\documentclass[11pt]{article}

\usepackage{amsmath,amsfonts,amssymb,amsthm,thmtools,thm-restate,mathabx,stmaryrd,enumerate,ytableau,hyperref,multirow}
\usepackage[boxed,linesnumbered]{algorithm2e}
\usepackage{cleveref}
\usepackage{fullpage}
\usepackage{color}

\usepackage{tikz}
\usepackage{ytableau}

\newif\ifdraft
\drafttrue

\DeclareFontFamily{OMX}{MnSymbolE}{}
\DeclareFontShape{OMX}{MnSymbolE}{m}{n}{
    <-6>  MnSymbolE5
   <6-7>  MnSymbolE6
   <7-8>  MnSymbolE7
   <8-9>  MnSymbolE8
   <9-10> MnSymbolE9
  <10-12> MnSymbolE10
  <12->   MnSymbolE12}{}
\DeclareSymbolFont{mnlargesymbols}{OMX}{MnSymbolE}{m}{n}
\SetSymbolFont{mnlargesymbols}{bold}{OMX}{MnSymbolE}{b}{n}
\DeclareMathDelimiter{\llangle}{\mathopen}{mnlargesymbols}{'164}{mnlargesymbols}{'164}
\DeclareMathDelimiter{\rrangle}{\mathclose}{mnlargesymbols}{'171}{mnlargesymbols}{'171}

\providecommand{\poly}{\operatorname{poly}}
\providecommand{\RR}{\mathbb{R}}
\providecommand{\Sym}[1][n]{S_{#1}}
\providecommand{\pms}[1][2n]{\mathcal{M}_{#1}}
\providecommand{\adeg}{\widetilde{\deg}}
\providecommand{\bs}{\mathit{bs}}
\providecommand{\fbs}{\mathit{fbs}}
\providecommand{\bits}{\{0,1\}}
\providecommand{\EE}{\mathbb{E}}
\providecommand{\id}{\mathsf{id}}

\providecommand{\universe}{\mathcal{U}}
\providecommand{\domain}{\mathcal{X}}
\providecommand{\queries}{\mathcal{Q}}
\providecommand{\answers}{\mathcal{A}}
\providecommand{\chunk}{\mathfrak{c}}
\providecommand{\Dparam}{\Delta}
\providecommand{\Cparam}{M}
\providecommand{\Bparam}{\beta}
\providecommand{\Bcparam}{\beta_{\chunk}}
\providecommand{\chars}{X}
\providecommand{\Sparam}{\alpha}
\providecommand{\Lparam}{N}
\providecommand{\Iparam}{I}
\providecommand{\dprime}{D}
\providecommand{\Cbound}{C}
\providecommand{\mparam}{m}
\providecommand{\oracle}{\mathcal{O}}
\providecommand{\oraclec}{T}

\declaretheorem[numberwithin=section]{theorem}
\declaretheorem[sibling=theorem]{lemma}
\declaretheorem[sibling=theorem]{corollary}
\declaretheorem[sibling=theorem]{proposition}

\providecommand{\subparagraphit}[1]{\subparagraph{\emph{#1}}}

\title{Complexity Measures on the Symmetric Group and Beyond}
\author{Neta Dafni, Yuval Filmus, Noam Lifshitz, Nathan Lindzey, Marc Vinyals}

\begin{document}
\maketitle

\begin{abstract}
We extend the definitions of complexity measures of functions to domains such as the symmetric group. The complexity measures we consider include degree, approximate degree, decision tree complexity, sensitivity, block sensitivity, and a few others. We show that these complexity measures are polynomially related for the symmetric group and for many other domains.

To show that all measures but sensitivity are polynomially related, we generalize classical arguments of Nisan and others. To add sensitivity to the mix, we reduce to Huang's sensitivity theorem using ``pseudo-characters'', which witness the degree of a function.

Using similar ideas, we extend the characterization of Boolean degree~1 functions on the symmetric group due to Ellis, Friedgut and Pilpel to the perfect matching scheme. As another application of our ideas, we simplify the characterization of maximum-size $t$-intersecting families in the symmetric group and the perfect matching scheme.
\end{abstract}

\section{Introduction} \label{sec:introduction}

A classical result in complexity theory states that a Boolean function $f\colon \bits^n \to \bits$ of degree~$d$ can be computed using a decision tree of depth $\poly(d)$. Conversely, a Boolean function computed by a decision tree of depth~$d$ has degree at most~$d$. Thus degree and decision tree complexity are polynomially related. Other complexity measures which are polynomially related to the degree include approximate degree, certificate complexity, and block sensitivity. Recently, Huang~\cite{Huang} added sensitivity to the list.

Can we prove similar results for Boolean functions on other domains? 
Such domains have been introduced to complexity theory in recent years: for example, O'Donnell and Wimmer~\cite{ODW} used Boolean functions on the so-called ``slice'' to construct optimal nets for monotone functions; Barak et al.~\cite{BGHMRS} used Boolean functions on the Reed--Muller code to construct and analyze the influential ``short code''; and recently, Khot, Minzer and Safra~\cite{KMS} proved the 2-to-2 conjecture using Boolean functions on the Grassmann scheme.

Although yet to see applications to complexity theory, perhaps the most appealing domain is the symmetric group. We say that a function $f\colon \Sym \to \RR$ has degree at most~$d$ if any of the following equivalent conditions hold:
\begin{enumerate}
\item $f(\pi)$ can be written as a linear combination of \emph{$d$-juntas}, which are functions depending on $\pi(i_1),\ldots,\pi(i_d)$ for some $i_1,\ldots,i_d \in [n]$.
\item Representing the input as a permutation matrix, $f$ can be written as a degree~$d$ polynomial in the entries of the matrix.
\item $f$ has Fourier-degree~$d$, that is, it is supported on isotypic components corresponding to partitions $\lambda$ with $\lambda_1 \geq n-d$.
\end{enumerate}
(The reader who is not familiar with representation theory can ignore the last definition.)

What is the correct generalization of decision tree? We take our inspiration from the work of Ellis, Friedgut and Pilpel~\cite{EFP}, which characterized the Boolean degree~$1$ functions on $\Sym$. These are functions that depend on some $\pi(i)$ or on some $\pi^{-1}(j)$. This suggests the following definition: a decision tree for functions on $\Sym$ is a decision tree with queries of the form ``$\pi(i)=?$'' and ``$\pi^{-1}(j)=?$''. Essentially the same definition (``matching decision trees'') is used in lower bounds on the pigeonhole principle~\cite{UrquhartFu}.

We show that this is a good definition by proving that degree and decision tree complexity are polynomially related for the symmetric group. In fact, we are able to generalize many other complexity measures to the symmetric group, and show that all of them are polynomially related:

\begin{theorem} \label{thm:intro-main-sn}
The following complexity measures (appropriately defined) are all polynomially related for Boolean functions over the symmetric group: degree, approximate degree, decision tree complexity, certificate complexity, unambiguous certificate complexity, sensitivity, block sensitivity, fractional block sensitivity, quantum query complexity.
\end{theorem}

Our results hold for many other domains, such as the perfect matching scheme (the set of all perfect matchings in $K_{2n}$) and balanced slices (the balanced slice consists of all vectors in $\bits^{2n}$ with equally many $0$s and $1$s, and is also known as the Johnson scheme $J(2n,n)$).

We prove \Cref{thm:intro-main-sn} and its generalizations in an abstract framework based on simplicial complexes. In this framework, every point in the domain is a set. For example:
\begin{enumerate}
\item Boolean cube: We identify each vector $x \in \bits^n$ with the set $\{ (i,x_i) : i \in [n] \}$.	
\item Symmetric group: We identify each permutation $\pi \in \Sym$ with the set $\{ (i,\pi(i)) : i \in [n] \}$.
\end{enumerate}
A function has degree~$d$ if it can be written as a linear combination of functions of the form ``the input set contains $S$'', where $|S| \leq d$; this generalizes the usual notion of degree in these two domains. Our decision trees allow any queries of the form ``which element of the set $Q$ does the input set contain?'', as long as there is a unique answer for every input.

With this setup in place, we are able to polynomially relate all complexity measures other than sensitivity by generalizing classical arguments, as presented by Buhrman and de~Wolf~\cite{BuhrmandeWolf}, for example. To add sensitivity to the mix, we reduce to Huang's sensitivity theorem~\cite{Huang} using basic representation theory.

Generalizing ideas of Gopalan et al.~\cite{GNSTW}, we also prove the following simple result:
\begin{theorem} \label{thm:intro-ball-sn}
If a function on the symmetric group has sensitivity~$s$, then it can be recovered from its evaluation on a ball of radius $O(s)$ around an arbitrary permutation.	
\end{theorem}


Using this, we show that low sensitivity functions can be computed efficiently:

\begin{theorem} \label{thm:intro-circuit}
If a function on the symmetric group has sensitivity $s$, then it can be computed using a circuit of size $n^{O(s)}$.	
\end{theorem}

This should be compared to a decision tree for the function, which corresponds to a balanced formula of size $n^{O(D)}$, where $D \geq s$ is the decision tree complexity.

\paragraph{Degree 1 functions} Our results show that in a wide variety of domains, Boolean degree~$1$ functions can be computed by constant depth decision trees. Can we say more?

Boolean degree~$1$ functions on the Boolean cube are \emph{dictators}, that is, depend on a single coordinate, and the same holds for functions on the balanced slice. Ellis, Friedgut and Pilpel~\cite{EFP} showed that the same holds for the symmetric group, with the correct interpretation of ``dictator'': a function depending only on some $\pi(i)$, or only on some $\pi^{-1}(j)$. In contrast, Filmus and Ihringer~\cite{FI1} showed that Boolean degree~$1$ functions on the Grassmann scheme ($k$-dimensional subspaces of an $n$-dimensional vector space over a finite field) could depend on two different ``data points''.

Among the domains we consider, in many cases Boolean degree~$1$ functions are trivially dictators. In some other cases, describing all Boolean degree~$1$ functions seems difficult. We identify one case in which the problem is feasible:

\begin{theorem} \label{thm:intro-deg1-pm}
A Boolean function on the perfect matching scheme has degree at most~$1$ if and only if it is one of the following: a constant function; a function depending on the match of some vertex $i$; or a function depending on whether the perfect matching intersects some triangle.
\end{theorem}

Incidentally, this is another example in which there are non-dictatorial degree~$1$ functions, namely those depending on intersections with a triangle.

We prove \Cref{thm:intro-deg1-pm} using polyhedral techniques. As in the proof of the corresponding result for the symmetric group by Ellis, Friedgut and Pilpel (which we paraphrase using our methods), we first characterize all \emph{nonnegative} degree~$1$ functions, using the classical characterization of supporting hyperplanes of the perfect matching polytope. To deduce the result for Boolean functions, we use a simple result from the theory of complexity measures: a degree~$1$ function has sensitivity at most~$1$.

The reader is perhaps wondering about nonnegative functions of higher degree. Can we say anything intelligent about them? It turns out that the answer is negative already for the Boolean cube: classical results on the Sherali--Adams hierarchy~\cite{GMT} show that there exist nonnegative degree~$2$ functions which, if written as nonnegative linear combinations of monomials over literals (that is, products of factors of the form $x_i$ and $1-x_j$), require degree $\Omega(n)$.

\paragraph{Application to Erd\H{o}s--Ko--Rado theory}
The work of Ellis, Friedgut and Pilpel, which has already been mentioned several times, is about intersecting families of permutations. A subset $\mathcal{F} \subseteq \Sym$ is \emph{$t$-intersecting} if any two $\pi_1,\pi_2 \in \mathcal{F}$ agree on the image of at least $t$~points. In other words, if we think of $\pi_1,\pi_2$ as sets (as in our setup), then $|\pi_1 \cap \pi_2| \geq t$. How large can a $t$-intersecting family be? One construction is a \emph{$t$-star}:
\[
 \mathcal{F} = \{ \pi \in \Sym : \pi(i_1) = j_1, \ldots, \pi(i_t) = j_t \}.
\]
Ellis et al.\ show that for large enough $n$ (depending on $t$), these families have the maximum possible size, and moreover uniquely so: every $t$-intersecting family of the maximum size $(n-t)!$ is a $t$-star. Unfortunately, their argument for the uniqueness claim is wrong, see~\cite{EFP-comment}. Uniqueness can be recovered from the work of Ellis~\cite{Ellis2}, which proves a much stronger result, and is quite complicated.

We give a much simpler proof of uniqueness, using the connection between degree and certificate complexity:

\begin{theorem} \label{thm:intro-uniqueness-sn}
For every $t,d$, the following holds for large enough $n$. If $f$ is the characteristic vector of a $t$-intersecting family and $\deg f \leq d$, then either $f$ is contained in a $t$-star, or the corresponding family contains $O((n-t-1)!)$ permutations.
\end{theorem}


Ellis et al.\ show that a $t$-intersecting family of size $(n-t)!$ must have degree~$t$ (for large enough~$n$), and so \Cref{thm:intro-uniqueness-sn} shows that for large enough~$n$, such a family must be a $t$-star.

\Cref{thm:intro-uniqueness-sn} generalizes to other domains for which similar intersection theorems are known, such as the perfect matching scheme~\cite{lindzey2018intersecting,LindzeyThesis}, and to cross-$t$-intersecting families.

We see \Cref{thm:intro-uniqueness-sn} as a contribution of theoretical computer science to extremal combinatorics. It illustrates the usefulness of the theory developed in this work.

\paragraph{Paper organization} We describe the basic setup in \Cref{sec:setup}, using two running examples: the Boolean cube and the symmetric group. We prove our main theorem, polynomially relating all complexity measures other than sensitivity, in \Cref{sec:main}. After describing several domains to which our techniques apply in \Cref{sec:domains}, we prove several sensitivity theorems (polynomially relating sensitivity to all other complexity measures) in \Cref{sec:sensitivity}. We discuss degree~$1$ functions in \Cref{sec:degree1}, the application to intersecting families in \Cref{sec:intersecting}, and circuit constructions in \Cref{sec:circuit}. We close the paper with \Cref{sec:open-questions}, in which we discuss some open questions.

\paragraph{Guide to the reader} Our main contribution is the setup described in \Cref{sec:setup}. A reader who is short on time can concentrate only on the running example of the symmetric group, and skip \Cref{sec:domains} altogether. The remaining sections are completely independent. Of these sections, \Cref{sec:main} and \Cref{sec:circuit} adapt known arguments, while \Cref{sec:sensitivity}, \Cref{sec:degree1}, and \Cref{sec:intersecting} contain novel arguments.

\paragraph{Acknowledgements} This project has received funding from the European Union's Horizon 2020 research and innovation programme under grant agreement No~802020-ERC-HARMONIC. We thank Nitin Saurabh for many helpful discussions.

\section{Basic setup} \label{sec:setup}

In this section we describe our general setup, using two running examples: the Boolean cube and the symmetric group.

\subsection{Defining a domain} \label{sec:setup-definition}

\paragraph{Domain} A \emph{domain} $\domain$ is a collection of subsets of some universe $\universe$, all of the same size~$n$, known as the \emph{dimension}.\footnote{In the theory of simplicial complexes, $\domain$ would actually have dimension $n-1$.}

\subparagraph{Boolean cube} We think of the Boolean cube $\bits^n$ as the product set
\[
 \domain = \bigtimes_{i=1}^n \{(i,0),(i,1)\}
\]
over the universe $\universe = \{(i,b) : i \in [n], b \in \bits\}$.

A vector $x \in \bits^n$ corresponds to the set $\{(i,x_i) : i \in [n]\}$.

\subparagraph{Symmetric group} We identify a permutation $\pi \in \Sym$ with its graph $\{ (i,\pi(i)) : i \in [n] \}$.

We can think of a permutation as a perfect matching in the bipartite graph $K_{n,n}$. In these terms, $\universe$ is the set of edges of $K_{n,n}$, and $\domain$ consists of all perfect matchings.

\paragraph{Query} In order to define decision trees, we will also need to define the notion of query. A \emph{query} $Q$ is a subset of $\universe$ which intersects each set in $\domain$ in exactly one point. With each domain we will associate a set $\queries$ of allowed queries. 

To avoid trivialities, we assume that $\queries$ satisfies the following property: $\bigcup \queries = \universe$, that is, every element of $\universe$ is an element of some query in $\queries$. This will ensure that any function can be represented as a decision tree.

\subparagraph{Boolean cube} Decision trees on the Boolean cube use queries of the form ``$x_i=?$''. In our formalism, such a query corresponds to the set $\{(i,0),(i,1)\}$, which is guaranteed to intersect each set in $\domain$ at exactly one element. Therefore
\[
 \queries = \{ \{(i,0), (i,1)\} : i \in [n] \}.
\]

\subparagraph{Symmetric group} The correct notion of decision trees for the symmetric group is hinted at by the characterization of Boolean degree~$1$ functions due to Ellis, Friedgut and Pilpel~\cite{EFP}, and has also appeared in the proof complexity literature, in the context of lower bounds on the pigeonhole principle~\cite{UrquhartFu}. The allowed queries are ``$\pi(i) = ?$'' and ``$\pi^{-1}(j) = ?$'', which in our setup translate to:
\[
 \queries = \{ \{(i,j) : j \in [n] \} : i \in [n] \} \cup \{ \{(i,j) : i \in [n] \} : j \in [n] \}.
\]

\subparagraph{Note} In order to completely define a domain, we need to specify both $\domain$ and $\queries$, though in practice we will refer to a domain using $\domain$ only, for brevity.

The set $\queries$ is not canonical: its choice determines the notion of decision tree used to define decision tree complexity, as well as the values of the parameters in \Cref{sec:setup-parameters}, which in turn affect our main theorem, \Cref{thm:main}, quantitatively. 
As an example, the symmetric group can also be viewed as a multislice (see \Cref{sec:domains-ms}), in which case only queries of the form ``$\pi(i)=?$'' are allowed.




\paragraph{Chunk size} Looking ahead, the ``blocks'' in the definition of block sensitivity correspond to removing $b$ elements from some $x \in \domain$ and replacing them with $b$ other elements. For example, on the Boolean cube, moving from $000$ to $101$ corresponds to replacing $(1,0),(3,0)$ with $(1,1),(3,1)$. The \emph{size} of the block is~$b$. In the definition of sensitivity, we require all blocks to have minimum size. Accordingly, we define the \emph{chunk size} $\chunk$ to be the minimal value of $|x \setminus y|$ for $x \neq y \in \domain$.

\subparagraph{Boolean cube} For two points $x,y$ on the Boolean cube, $|x \setminus y|$ is the Hamming distance between the vector representations of $x$ and $y$. Consequently, $\chunk = 1$.

\subparagraph{Symmetric group} The minimal number of changes required to move from one permutation to the other is~$2$ (corresponding to applying a transposition). Therefore, $\chunk = 2$.

\subparagraph{Note} We could have made $\chunk$ a free parameter, but for all domains for which we can prove a sensitivity theorem, we can prove it with respect to the current definition of $\chunk$.


\subsection{Complexity measures} \label{sec:setup-measures}

With the setup in hand, we can define the various complexity measures we are interested in, for a given $n$-dimensional domain $\domain$ over a universe $\universe$ with queries $\queries$. While many of the complexity measures make sense for arbitrary functions, we will only define them for Boolean functions, that is, functions $\domain \to \bits$.
Our selection of complexity measures matches Aaronson, Ben~David and Kothari~\cite[Table 2]{ABDK}. There are many other measures encountered in the literature, for example tree sensitivity~\cite{GSW} and quantum certificate complexity~\cite{Aaronson}, which we leave for future work.

All definitions that we give below coincide with the usual definitions in the case of the Boolean cube, as the reader can easily verify.

\paragraph{Degree and approximate degree} A \emph{polynomial} is a function $P\colon \domain \to \RR$ of the form
\[
 P(x) = \sum_{S \subseteq \universe} c_S \llbracket x \supseteq S \rrbracket,
\]
where $\llbracket x \supseteq S \rrbracket$ is the Boolean function which equals~$1$ if $x \supseteq S$. The \emph{degree} of a polynomial is the maximum size of a set $S$ such that $c_S \neq 0$.

The \emph{degree} of a Boolean function $f\colon \domain \to \bits$, denoted $\deg(f)$, is the minimum degree of a polynomial $P$ such that $f(x) = P(x)$ for all $x \in \universe$.

The \emph{$\epsilon$-approximate degree} of a Boolean function $f\colon \domain \to \bits$, denoted $\adeg_\epsilon(f)$, is the minimum degree of a polynomial $P$ such that $|f(x) - P(x)| \leq \epsilon$ for all $x \in \domain$. This definition is sensible for all $\epsilon \in (0,1/2)$. The \emph{approximate degree} is $\adeg(f) = \adeg_{1/3}(f)$.
(The constant $1/3$ can be replaced by any value in $(0,1/2)$; the approximate degree will change by at most a constant factor.)
Note that $\adeg_0(f) = \deg(f)$.

\subparagraph{Symmetric group} Recall that we represent a permutation $\pi \in S_n$ as a set of pairs $\{ (i,\pi(i)) : i \in [n] \}$. Equivalently, for every pair $(i,j)$, there is a Boolean variable $x_{i,j}$ indicating whether $\pi(i) = j$. The variables $x_{i,j}$ together form the permutation matrix representation of the input permutation.

The degree of a function $f\colon S_n \to \bits$ is the minimum degree of a polynomial in the variables $x_{i,j}$ representing the function. This notion of degree coincides with a natural spectral notion, as described in \Cref{sec:sensitivity-sn}. It also coincides with the following notion of \emph{junta degree}: the degree of $f$ is the minimal $d$ such that $f$ can be written as a linear combination of \emph{$d$-juntas}, where a $d$-junta is a function depending on $d$ data items of the form $\pi(i)$ or $\pi^{-1}(j)$.

\subparagraph{Notes} A given Boolean function can have several different polynomial representations. For example, if $Q$ is any query then $P(X) = \sum_{e \in Q} \llbracket e \in x \rrbracket - 1$ is a representation of the zero function. In many cases it is possible to impose more constraints on the polynomial, thus enforcing the representation to be unique. For example, over the Boolean cube, if we only allow sets $S$ with elements of the form $(i,1)$ then the representation will be unique. For the case of the slice, see~\cite{Filmus-basis}; a similar unique representation should exist for the symmetric group.

In the definition of a polynomial, it suffices to consider sets $S$ which are subsets of some $x \in \domain$. Such sets are known as \emph{faces} in the parlance of simplicial complexes (the sets in $\domain$ are known as \emph{facets}).

In the case of the Boolean cube, the maximum degree coincides with the dimension. In contrast, the maximum degree of a function on $\Sym$ is only $n-1$.

As mentioned in the introduction, this notion of degree (``spatial degree'') coincides with algebraic notions of degree (``spectral degree'') for all domains considered in this paper for which such notions exist. We expand on this in \Cref{sec:domains-degree}.

\paragraph{Certificate complexity} Let $f\colon \domain \to \bits$ be a Boolean function. A \emph{certificate} for a point $x \in \domain$ is a subset $c \subseteq x$ such that $f(y) = f(x)$ whenever $y \supseteq c$. In other words, a certificate for $x$ is a set of elements of $x$ which suffice to guarantee that $f$ will attain the value $f(x)$.

The \emph{certificate complexity} of $f$ at $x$, denoted $C(f,x)$, is the minimum size of a certificate for $x$. For $b \in \{0,1\}$, the \emph{$b$-certificate complexity} of $f$, denoted $C_b(f)$, is the maximum value of $C(f,x)$ over all points in $f^{-1}(b)$. The \emph{certificate complexity} of $f$, denoted $C(f)$, is $\max_x C(f,x)$.

\subparagraph{Symmetric group} The certificate complexity of a function $f\colon S_n \to \bits$ at a permutation $\pi$ is the minimum size of a \emph{partial permutation} $\rho \subseteq \pi$ (that is, an injective function whose domain is a subset of $[n]$) that forces the value of $f$, in the sense that $f(\sigma) = f(\pi)$ whenever $\sigma \supseteq \rho$. 

\subparagraph{Note on negative certificates} We can think of functions on $\domain$ as partial functions on $\bits^{\universe}$ with a fixed domain. From that perspective, it is natural to allow certificates which not only guarantee that some elements belong to the input, but also that some elements don't belong to the input. In other words, given a point $x \in \domain$, we can consider pairs $(c,d)$, where $c \subseteq x$ and $d \subseteq \overline{x}$. Such a pair is a certificate for $x$ if $f(y) = f(x)$ whenever $y \supseteq c$ and $\overline{y} \supseteq d$. The size of such a certificate is $|c|+|d|$.

Such ``negative'' certificates do not reduce the certificate complexity. Indeed, consider a certificate $(c,d)$ for some $x \in \domain$, and an element $e \in d$. Our assumption that $\bigcup \queries = \universe$ guarantees that $e$ participates in some query $Q \in \queries$. Since $Q$ is a query and $e \notin x$, we have $Q \cap x = \{e'\}$ for some $e' \neq e$. Conversely, if $e' \in y$ then $e \notin y$, since $Q$ intersects $y$ at a unique element. Therefore we can replace $e \in d$ with $e' \in c$, obtaining another certificate for $x$ of the same size. Eliminating all elements of $d$ in this way, we obtain a standard certificate for $x$ of the same size.

\paragraph{Unambiguous certificate complexity} Unambiguous certificate complexity is a notion closely related to decision tree complexity. A collection $\mathcal{C}$ of subsets of $\universe$ is a \emph{set of unambiguous certificates} for a Boolean function $f\colon \domain \to \bits$ if for each $x \in \domain$ there is a unique $y \in \mathcal{C}$ such that $y \subseteq x$, and this $y$ is a certificate for $x$.

The \emph{complexity} of a set $\mathcal{C}$ of unambiguous certificates is the size of the largest set in $\mathcal{C}$.
The \emph{unambiguous certificate complexity} of $f$, denoted $U(f)$, is the minimum complexity of a set of unambiguous certificates for $f$.

As in the case of certificate complexity, negative certificates do not reduce the unambiguous certificate complexity, for similar reasons. Indeed, consider a certificate with the negative guarantee $e \notin x$. Choosing some query $Q$ containing $x$, the negative guarantee can be expanded \emph{unambiguously} to positive guarantees $e' \in x$ for every $e' \in Q \setminus x$ which is consistent with the rest of the certificate.

\paragraph{Decision tree complexity} A \emph{decision tree} is a tree whose internal nodes are labeled by elements of $\queries$, whose edges are labeled by elements of $\universe$, and whose leaves are labeled by values in $\bits$. For an internal node $v$ labeled by $Q \in \queries$, let $\ell$ be the set of edge labels appearing in the path from the root to $v$. For each $a \in Q$, say that $a$ is \emph{feasible} for $v$ if some $x \in \domain$ contains $\ell \cup \{a\}$. We require $v$ to have exactly one child per feasible $a$, with the corresponding edge labeled $a$.

A decision tree computes a function in the natural way. The \emph{decision tree complexity} of a function $f\colon \domain \to \bits$, denoted $D(f)$, is the minimum depth (measured by edges) of a decision tree computing $f$.

The \emph{$\epsilon$-error randomized decision tree complexity} of $f$, denoted $R_\epsilon(f)$, is the minimum $R$ such that there is a probability distribution $\mathcal{D}$ on decision trees of depth at most $R$ such that $\Pr_{T \sim \mathcal{D}}[T(x) = f(x)] \geq 1-\epsilon$ for all $x \in \domain$. We define $R(f) = R_{1/3}(f)$. As in the case of approximate degree, the constant $1/3$ only affects the measure up to a constant factor, as long as $\epsilon \in (0,1/2)$.

A \emph{zero-error decision tree} $T$ for $f$ is a decision tree whose leaves are labeled $\{0,1,\ast\}$ (where $\ast$ has the interpretation ``do not know'') such that $T(x) \in \{f(x),\ast\}$ for all $x \in \domain$.
The \emph{zero-error randomized decision tree complexity} of $f$, denoted $R_0(f)$, is the minimum $R$ such that there is a probability distribution $\mathcal{D}$ on zero-error decision trees for $f$ of depth at most $R$ such that $\Pr_{T \sim \mathcal{D}}[T(x) = \ast] \leq 2/3$ form all $x \in \domain$. Once again, the constant $2/3$ only affects the measure up to a constant factor, and can be replaced by any other constant in $(0,1)$.

\subparagraph{Symmetric group} A decision tree for functions of $S_n$ uses queries of the form ``$\pi(i) = ?$'' and ``$\pi^{-1}(j) = ?$''. This notion coincides with the \emph{matching decision trees} appearing in~\cite{UrquhartFu}.

\subparagraph{Notes} In the case of the Boolean cube, the number of possible answers is always the same:~$2$. In contrast, in the case of the symmetric group, a query at depth $d$ only has $n-d$ answers (assuming no query is repeated along the way). This phenomenon is captured by the notion of feasibility.

Our assumption that $\bigcup \queries = \universe$ implies that there is a query algorithm that determines the input using at most $n$~queries, and in particular, the decision tree complexity of any function is at most~$n$. The query algorithm proceeds in $n$ rounds, each of which uncovers some element $a_i$ belonging to the input. In the $i$'th round, we ask an arbitrary query not containing $a_1,\ldots,a_{i-1}$, and so the answer $a_i$ differs from $a_1,\ldots,a_{i-1}$ by construction. To see that such a query must exist, let $a$ be some element in the input other than $a_1,\ldots,a_{i-1}$. Since $\bigcup \queries = \universe$, some query $Q$ contains $a$. Since $Q$ intersects the input at exactly one element, it cannot contain any of the elements $a_1,\ldots,a_{i-1}$.

\paragraph{Sensitivity and block sensitivity} The definition of block sensitivity is less intuitive than the definitions we have seen so far.

Let us start by recalling the usual definition of block sensitivity. The block sensitivity of a function $f\colon \bits^n \to \bits$ at a point $x$ is the maximum number of disjoint ``blocks'' $B_1,\ldots,B_s \subseteq [n]$ such that $f(x \oplus B_i) \neq f(x)$ for all $i$, where $x \oplus B_i$ is the result of flipping all the bits whose indices belong to $B_i$.

The operation $x \mapsto x \oplus B_i$ corresponds, in our formalism (viewing $x$ as a set), to removing from $x$ the points $\{(j,x_j) : j \in B_i\}$ and replacing them with $\{(j,1-x_j) : j \in B_i\}$. The fact that the blocks $B_i$ are disjoint corresponds to the sets $\{(j,x_j) : j \in B_i\}$ being disjoint.

Accordingly, we define the \emph{block sensitivity} of a function $f\colon \domain \to \bits$ at a point $x \in \domain$, denoted $\bs(f,x)$, to be the maximum number of points $y_1,\ldots,y_s \in \domain$ (corresponding to $x \oplus B_i$) such that (i) $f(y_i) \neq f(x)$ and (ii) the sets $x \setminus y_i$ are disjoint. The block sensitivity of $f$ is $\bs(f) = \max_{x \in \domain} \bs(f,x)$.

We define \emph{sensitivity} at a point $s(f,x)$ and global sensitivity $s(f)$ in the same way, using one additional constraint: $|x \setminus y_i| = \chunk$.

\subparagraph{Symmetric group} The block sensitivity of a function $f\colon S_n \to \bits$ at a permutation $\pi$ is the maximal number of permutations $\tau_j$ such that $f(\tau_j \pi) \neq f(\pi)$ and the sets $D(\tau_j) = \{ i \in [n] : \tau_j(i) \neq i \}$ are disjoint; we say that the permutations $\tau_j$ are disjoint.

This definition appears asymmetric, since we multiply $\pi$ by $\tau_j$ on the left. Note, however, that $\tau_j \pi = \pi \sigma_j$, where $\sigma_j = \pi^{-1} \tau_j \pi$ is a conjugate of $\tau_j$. Using the convention $(\alpha \beta)(i) = \alpha(\beta(i))$, we see that $D(\sigma_j) = \pi^{-1}(D(\tau_j))$, and so the permutations $\sigma_j$ are also disjoint. Therefore the definition of block sensitivity is, in fact, symmetric.

If we require moreover that the $\tau_j$ be transpositions, we get the definition of sensitivity. This definition is also symmetric, since $\sigma_j$, as a conjugate of $\tau_j$, is also a transposition.

\subparagraph{Note on block sensitivity} The definition of disjointness is asymmetric: we consider $x \setminus y_i$ but not $y_i \setminus x$. It turns out that if the sets $x \setminus y_i$ are disjoint, then so are the sets $y_i \setminus x$.

\begin{lemma} \label{lem:bs-symmetric}
If $x,y_1,\ldots,y_s \in \domain$ are such that $y_i \neq x$ for all $i$, and the sets $x \setminus y_i$ are disjoint, then so are the sets $y_i \setminus x$.	
\end{lemma}
\begin{proof}
Suppose that for some $i \neq j$, the sets $y_i \setminus x$ and $y_j \setminus x$ have a common element $a$. Since $\bigcup \queries = \universe$, there is a query $Q_a$ containing $a$. The query $Q_a$ intersects $x$ at a unique point $b \neq a$. Since $Q_a$ intersects $y_i$ at a unique point, necessarily $b \notin y_i$, and so $b \in x \setminus y_i$. The same argument shows that $b \in x \setminus y_j$, contradicting the assumption that these two sets are disjoint.
\end{proof}



\subparagraph{Notes on sensitivity} In the particular case of the symmetric group, another possible definition of sensitivity is as follows. Define the \emph{edge sensitivity} of $f\colon \Sym \to \bits$ at a point $x$, denoted $t(f,x)$, as the number of transpositions $\tau$ such that $f(x^\tau) \neq f(x)$. The edge sensitivity of $f$ is then $t(f) = \max_x t(f,x)$.

This definition is natural from the point of view of Boolean function analysis. In the case of the Boolean cube, it is well-known that average sensitivity equals total influence, where total influence is defined via a graph structure imposed on $\bits^n$, namely the hypercube. We can define total influence in a similar way for the symmetric group, using the transposition graph. We then get that the average edge-sensitivity is the same as the total influence~\cite{Wimmer}.

On the Boolean cube, total influence is bounded by the degree. On the symmetric group, total influence is bounded by the degree times $n$. This is reflected in the inequality $t(f) = O(n\cdot s(f))$. Indeed, it is not hard to show that $t(f,x) \leq 2n\cdot s(f,x)$. Unfortunately, we don't have a matching bound in the other direction. To see this, let $f(x)$ be the function ``the cycle decomposition of $x$ contains at least one 2-cycle of the form $(2i-1\;2i)$''. Then $s(f) \geq s(f,\mathrm{id}) = n/2$ is maximal, but simple case analysis shows that $t(f) = O(n)$. On the other hand, the sign function maximizes both the sensitivity and the edge sensitivity.

\paragraph{Fractional block sensitivity} Fractional block sensitivity, first defined by Tal~\cite{Tal} and by Gilmer, Saks and Srinivasan~\cite{GSS}, is a relaxation of block sensitivity obtained by relaxing an integer program to a linear program.

We can express the block sensitivity of a function $f\colon \domain \to \bits$ at a point $x \in \domain$ as the following integer program. The variables are $c_y \in \{0,1\}$, for each $y \in \domain$ such that $f(y) \neq f(x)$, which indicate a collection $y_1,\ldots,y_s$ of ``blocks''. We want to maximize $\sum_y c_y$ (the number of ``blocks'') under the constraints
\[
 \sum_{p \notin y} c_y \leq 1 \text{ for all } p \in x.
\]
These constraints express the condition ``for each $p \in x$, there can be at most one block $y_i$ such that $p \in x \setminus y_i$.'' This is the same as asking for the sets $x \setminus y_i$ to be disjoint. Hence the solution to this integer program is $\bs(f,x)$.

If we relax the constraint $c_y \in \{0,1\}$ to the linear constraint $0 \leq c_y \leq 1$ then we get a linear program whose solution we denote $\fbs(f,x)$, the \emph{fractional block sensitivity of $f$ at $x$}. We also denote $\fbs(f) = \max_{x \in \domain} \fbs(f,x)$.

Linear programming duality gives us another linear program for $\fbs(f,x)$. The variables are $d_p$ for each $p \in x$. The goal is to minimize $\sum_p d_p$ under the constraints $0 \leq d_p \leq 1$ and
\[
 \sum_{p \notin y} d_p \geq 1 \text{ for all } y \in \domain \text{ such that } f(y) \neq f(x).
\]

This measure is known as \emph{fractional certificate complexity}, and is similar to \emph{randomized certificate complexity} which had been defined by Aaronson~\cite{Aaronson} (the two measures are identical up to constant factors, as shown by Tal~\cite{Tal} and by Gilmer et al.~\cite{GSS}).
To understand the provenance of these terms, let us consider the corresponding integer program, in which the constraint $0 \leq d_p \leq 1$ is replaced with $d_p \in \{0,1\}$. The variables $d_p$ define a subset $z \subseteq x$. For each $y$, the constraint above expresses the condition ``if $f(y) \neq f(x)$ then $y$ is not a superset of $z$''. In other words, $z$ is a certificate for $x$. The solution of this integer program is thus $C(f,x)$.

The foregoing shows how to obtain the strange definition of block sensitivity from the natural definition of certificate complexity: start with an integer program for $C(f,x)$; relax it to a linear program; dualize; tighten it up to an integer program for $\bs(f,x)$. Hopefully this convinces the reader that our definition of block sensitivity is the correct one.

\paragraph{Quantum query complexity} We assume that the reader is familiar with quantum query complexity in the case of the Boolean cube; see for example Buhrman and de~Wolf~\cite{BuhrmandeWolf}. We introduce an alphabet $\answers$ of answers to queries. For example, in the case of the Boolean cube we can choose $\answers = \bits$, and in the case of the symmetric group we can choose $\answers = [n]$.

A quantum query algorithm operates on a triplet of quantum registers: an input register of width $\lceil \log \queries \rceil$ qubits, an output register of width $\lceil \log \answers \rceil$ qubits, and workspace of arbitrary width. The \emph{quantum query operator} $O$ is the unitary operator that maps $|q\rangle|y\rangle|z\rangle$ to $|q\rangle |y \oplus a \rangle |z\rangle$, where $a$ is the answer to query $q$. (The exact encoding of queries and answers will not make a difference.)

A \emph{quantum query algorithm} of complexity $T$ consists of $T+1$ unitary transformations $U_0,U_1,\ldots,U_T$. To apply the algorithm on an input, we initialize the registers to $|0\rangle |0\rangle |0\rangle$ and apply the operations $U_0,O,U_1,O,\allowbreak\ldots,\allowbreak U_{T-1},O,U_T$, in that order. Finally, we measure the first qubit (the choice of qubit to measure is arbitrary) and output the answer.

The \emph{exact quantum query complexity} $Q_E(f)$ is the minimum complexity of a quantum query algorithm that always computes $f$ correctly. The \emph{bounded-error quantum query complexity} $Q_\epsilon(f)$ is the minimum complexity of a quantum query algorithm that on every input, computes $f$ correctly with probability at least $1-\epsilon$. We define $Q(f) = Q_{1/3}(f)$. As in previous cases, changing $\epsilon$ only affects $Q_\epsilon(f)$ by at most a constant factor, as long as $\epsilon \in (0,1/2)$. Also, $Q_E(f) = Q_0(f)$.

\paragraph{Index of notation} We have defined quite a few complexity measures:
\begin{itemize}
\item $\deg(f),\adeg(f),\adeg_\epsilon(f)$: degree and approximate degree.
\item $C(f),U(f)$: certificate complexity (general and unambiguous).
\item $D(f),R(f),R_\epsilon(f),R_0(f)$: decision tree complexity (deterministic, randomized, and zero-error).
\item $s(f),\bs(f),\fbs(f)$: sensitivity, block sensitivity, fractional block sensitivity.
\item $Q_E(f), Q(f), Q_\epsilon(f)$: quantum query complexity (exact and bounded-error).
\end{itemize}

\paragraph{Simple relations among the measures}
Some inequalities among the measures we have considered follow directly from the definitions.

\begin{lemma} \label{lem:simple-relations}
The following hold for every function $f\colon \domain \to \bits$:
\begin{enumerate}[(a)]
\item $\adeg(f) \leq \deg(f)$. \label{item:adeg-deg}
\item $C(f) \leq U(f)$. \label{item:C-U}
\item $R_0(f) \leq D(f)$. \label{item:R0-D}
\item $s(f) \leq \bs(f)$. \label{item:s-bs}
\item $Q(f) \leq Q_E(f)$. \label{item:Q-QE}
\item $R(f) \leq R_0(f)$. \label{item:R-R0}
\item $\bs(f) \leq \fbs(f) \leq C(f)$. \label{item:bs-fbs-C}
\item $\deg(f) \leq D(f)$. \label{item:deg-D}
\item $U(f) \leq D(f)$.	 \label{item:U-D}
\item $\bs(f) \leq R_\epsilon(f)/(1-2\epsilon)$, and so $\bs(f) \leq 3R(f)$ and $\bs(f) \leq D(f)$. \label{item:bs-R}
\item $Q_E(f) \leq D(f)$. \label{item:QE-D}
\item $\adeg_\epsilon(f) \leq Q_\epsilon(f)$, and so $\adeg(f) \leq 2Q(f)$ and $\deg(f) \leq 2Q_E(f)$. \label{item:adeg-Q}
\end{enumerate}
\end{lemma}
\begin{proof}
The first five relations are trivial.

We have $R(f) \leq R_0(f)$ since given a distribution of zero-error decision trees for $f$, if we convert every $\ast$-leaf into a random coin toss then we get a distribution of decision trees whose error probability is at most $\Pr[\ast] \cdot \tfrac{1}{2} \leq \tfrac{1}{3}$.

We have $\bs(f) \leq \fbs(f) \leq C(f)$ since $\fbs$ is a linear programming relaxation of a maximization integer program for $\bs(f)$, and of a minimization integer program for $C(f)$.

We have $\deg(f) \leq D(f)$ since ``reaching a leaf at depth $d$'' is a  degree~$d$ monomial, and $f$ can be written as a sum of these monomials over all $1$-leaves of a decision tree for $f$.

We have $U(f) \leq D(f)$ since the root-to-leaf paths in a decision tree form a set of unambiguous certificates.

\smallskip

Let us now show that $\bs(f) \leq R_\epsilon(f)/(1-2\epsilon)$; substituting $\epsilon = 0$ gives $\bs(f) \leq D(f)$. Let $x$ be a point such that $\bs(f,x) = \bs(f)$, say as witnessed by $y_1,\ldots,y_{\bs(f)}$. Let $\mathcal{D}$ be a distribution over decision trees of depth at most $R_\epsilon(f)$ such that $\Pr_{T \sim \mathcal{D}}[T(y) = f(y)] \geq 1-\epsilon$ for all $y \in \domain$.

Denote by $p_i$ the probability, over $T \sim \mathcal{D}$, that $T$ asks a query which has a different answer on $x$ and on $y_i$. When asked on $x$, such a query must return some element of $x \setminus y_i$. Since the sets $x \setminus y_i$ are disjoint, each query can differentiate $x$ from at most one $y_i$. This implies that $\sum_i p_i \leq R_\epsilon(f)$.

On the other hand, the probability that $T(x) \neq T(y_i)$ (over $T \sim \mathcal{D}$) is at most $p_i$. Therefore
\[
 1 - \epsilon \leq \Pr_{T \sim \mathcal{D}}[T(x) = f(x)] \leq
 \Pr_{T \sim \mathcal{D}}[T(x) \neq T(y_i)] + \Pr_{T \sim \mathcal{D}}[T(y_i) \neq f(y_i)] \leq p_i + \epsilon,
\]
since $T(x) = T(y_i) = f(y_i)$ implies $T(x) \neq f(x)$. This shows that each $p_i$ is at least $1 - 2\epsilon$. Since $\sum_i p_i \leq R_\epsilon(f)$ and there are $\bs(f)$ many $p_i$, we conclude that $(1-2\epsilon)\bs(f) \leq R_\epsilon(f)$.

\smallskip

We move on to the relations involving quantum query complexity. To see that $Q_E(f) \leq D(f)$, we show how to simulate a decision tree using a quantum query algorithm with complexity $D(f)$.  We assume without loss of generality that all leaves are at depth $D(f)$. The workspace will contain the current node. The first unitary $U_0$ fixes the initial state to $|q(r)\rangle |0\rangle |r\rangle$, where $q(v)$ is the query at node $v$, and $r$ is the root. The unitaries $U_1,\ldots,U_{D(f)-1}$ are defined to mimic the decision tree: $U_i$ maps $|q(v)\rangle |a\rangle |v\rangle$ to $|q(v_a)\rangle |0\rangle |v_a\rangle$, where $v$ is a node at depth $i-1$ and $v_a$ is its child corresponding to answer $a$ (such a unitary exists, since we can extend the constraints to a permutation of the basic states). Finally, $U_{D(f)}$ maps $|q(v) \rangle |a\rangle |v\rangle$ to $|q(v_a)\rangle |0\rangle |v_a\rangle$, where $q(v_a)$ is the label of the leaf $v_a$. By construction, just before the $i$'th query, the state of the algorithm is $|q(v_{i-1})\rangle |0\rangle |v_{i-1}\rangle$, where $v_0,\ldots,v_{D(f)}$ is the path in the decision tree corresponding to the input. In particular, at the end the state will be $|q(v_{D(f)})\rangle |0\rangle |v_{D(f)}\rangle = |f(x)\rangle |0\rangle |v_{D(f)}\rangle$, where $x$ is the input, and so the algorithm outputs $f(x)$ correctly.

Finally, we show that $\adeg_\epsilon(f) \leq 2Q_\epsilon(f)$. Given a quantum query algorithm of complexity $T$, we prove inductively that the amplitudes of the state of the algorithms after $k$ queries are polynomials of degree at most $k$. Since the initial state is constant and unitary operations do not affect the degree, it suffices to show that the quantum query operator $O$ increases the degree by at most~$1$. Indeed, denoting by $\alpha(\cdot)$ the magnitude before applying $O$ and by $\beta(\cdot)$ the magnitude after applying $O$, we have
\[
 \beta(|q\rangle |y\rangle |z\rangle) = \sum_{a \in q} \llbracket a \in x \rrbracket \alpha(|q\rangle |y \oplus a\rangle |z\rangle).
\]
The amplitudes of the final state thus have degree at most~$T$. The probability that the quantum algorithm outputs~$1$ is a sum of squares of magnitudes, and so is a polynomial $P$ of degree at most~$2T$. If the quantum query algorithm outputs the correct answer with probability $1-\epsilon$, then $|P(x) - f(x)| \leq \epsilon$ for all $x \in \domain$, and so $\adeg_\epsilon(f) \leq 2Q_\epsilon(f)$.
\end{proof}

\subsection{Admissible domains} \label{sec:setup-parameters}

Our arguments, which show that all complexity measures (except for sensitivity) are polynomially related, only work for domains satisfying an additional condition, \emph{composability}. In addition, the big O constants involved depend on four parameters of the domain, which we define below. 

\subsubsection{Composability}

Recall our definition of block sensitivity: the block sensitivity of a function $f\colon \domain \to \bits$ at a point $x$ is the maximum $s$ such that there are points $y_1,\ldots,y_s \in \domain$ satisfying $f(y_i) \neq f(x)$ and that the sets $x \setminus y_i$ are disjoint. When lower-bounding the approximate degree in terms of block sensitivity, we need the ability to compose these ``blocks''.

A domain is \emph{composable} if whenever $x,y_1,\ldots,y_s \in \domain$ are such that $y_i \neq x$ for all $i$ and the sets $x \setminus y_i$ are disjoint, then
\[
 z := x \setminus \bigcup_{i=1}^s (x \setminus y_i) \cup \bigcup_{i=1}^s (y_i \setminus x) \in \domain.
\]
(Recall that by \Cref{lem:bs-symmetric}, the sets $y_i \setminus x$ are also disjoint.)

\smallskip

Composability implies that for each $s_1,\ldots,s_b \in \bits$, the following set is in $\domain$:
\[
 x(s_1,\ldots,s_b) = x \setminus \bigcup_{i\colon s_i=1} (x \setminus y_i) \cup \bigcup_{i\colon s_i=1} (y_i \setminus x) = (x \cap y_1 \cap \cdots \cap y_b) \cup \bigcup_{i\colon s_i=0} (x \setminus y_i) \cup \bigcup_{i\colon s_i=1} (y_i \setminus x).
\]
In other words, composability allows us to identify copies of the Boolean cube inside arbitrary domains. Using this, we can use results on the Boolean cube to deduce results on other domains. In particular, we will use this idea to obtain a lower bound on the approximate degree in terms of block sensitivity.

\paragraph{Criterion for composability} For some domains, such as the Boolean cube, composability is easy to prove directly.
For other domains such as the symmetric group, proving composability is less immediate. We will use the following simple criterion.

\begin{lemma} \label{lem:composability}
If $\domain$, viewed as as a subset of $\bits^\universe$, is the intersection of $\bits^\universe$ and an affine subspace, then $\domain$ is composable.
\end{lemma}
\begin{proof}
Given sets $x,y_1,\ldots,y_s \in \universe$ such that $y_i \neq x$ for all $i$ and the sets $x \setminus y_i$ are disjoint, we have to show that
\[
 z := x \setminus \bigcup_{i=1}^s (x \setminus y_i) \cup \bigcup_{i=1}^s (y_i \setminus x) \in \universe.
\]
\Cref{lem:bs-symmetric} shows that the sets $y_i \setminus x$ are also disjoint.

Since $\domain$ is the intersection of $\bits^\universe$ with an affine subspace, it suffices to show that for any linear form $\ell$ over $\RR^\universe$ such that $\ell(x) = \ell(y_1) = \cdots = \ell(y_s)$, we have $\ell(z) = \ell(x)$. By linearity,
\[
 \ell(x \setminus y_i) = \ell(x) - \ell(x \cap y_i) = \ell(y_i) - \ell(x \cap y_i) = \ell(y_i \setminus x).
\]
Since the sets $x \setminus y_i$ and the sets $y_i \setminus x$ are disjoint, $\ell(z) = \ell(x)$ immediately follows.
\end{proof}

\subparagraph{Boolean cube} It is easy to prove directly that the Boolean cube is composable. Let $S_i$ be the set of indices appearing in $x \setminus y_i$. Then $y_i$ results from $x$ by flipping the indices in $S_i$. The condition that the sets $x \setminus y_i$ are disjoint translates to the sets $S_i$ being disjoint. The vector $z$ results from $x$ by flipping the indices in $S_1 \cup \cdots \cup S_s$, and in particular, it lies in the Boolean cube. 

\subparagraph{Symmetric group} Using \Cref{lem:composability}, it is easy to show that the symmetric group is composable. Indeed, it is the set of solutions to the following linear system:
\begin{align*}
&\sum_{j=1}^n x_{i,j} = 1 && \text{ for all } i \in [n] \\
&\sum_{i=1}^n x_{i,j} = 1 && \text{ for all } j \in [n] \\
&x_{i,j} \in \bits && \text{ for all } i,j \in [n]
\end{align*}
This system states that the matrix formed by the elements $x_{i,j} \in \bits$ is bistochastic, and so a permutation matrix.

\subparagraph{Note} In both the Boolean cube and the symmetric group, the domain consists of all sets intersecting each query at exactly one point. But this is not the case for other domains, such as the slice.

\subsubsection{Four parameters} 

We now introduce four parameters which control our results quantitatively.

\paragraph{Maximum degree} The \emph{degree} of an element $a$ is the number of queries in $\queries$ mentioning it, that is, the number of $Q \in \queries$ such that $a \in Q$. Since $\bigcup \queries = \universe$, each element has degree at least~$1$. We denote the maximum degree of an element by $\Dparam$.

\subparagraph{Boolean cube} Each element $(i,j)$ is mentioned by exactly one query, ``$x_i = ?$''. Therefore $\Dparam = 1$.

\subparagraph{Symmetric group} Each element $(i,j)$ is mentioned by two queries: ``$\pi(i) = ?$'' and ``$\pi^{-1}(j) = ?$''. Therefore $\Dparam = 2$.

\paragraph{Conflict bound} A \emph{partial input} is a subset $c \subseteq \universe$ which is a subset of some set $x \in \domain$. Two partial inputs $c_1,c_2$ \emph{conflict} if no set in $\domain$ contains both.
The \emph{conflict bound} $\Cparam$ is the maximal value such that if $c_1,c_2$ are two conflicting partial inputs of size at most $\Cparam$, then there is a query $Q \in \queries$ which ``separates'' them, that is, intersects $c_1$ and $c_2$ at different elements.

\subparagraph{Boolean cube} Two partial inputs are conflicting if one of them contains $(i,0)$ and the other contains $(i,1)$. Any two such partial inputs can be separated by the query ``$x_i = ?$''. Therefore $\Cparam = n$.

\subparagraph{Symmetric group} Two partial inputs are conflicting if either one of them contains $(i,j_1)$ and the other $(i,j_2)$, in which case they are separated by ``$\pi(i) = ?$'', or one of them contains $(i_1,j)$ and the other $(i_2,j)$, in which case they are separated by ``$\pi^{-1}(j) = ?$''. Therefore again $\Cparam = n$.

To see this, suppose that $c_1,c_2$ are two partial inputs which are not separated by any query, and consider their union $c_1 \cup c_2$. By assumption, the union, considered as a set of edges of $K_{n,n}$, is a matching. Since every matching in $K_{n,n}$ can be completed to a perfect matching, we see that $c_1,c_2$ are not conflicting.

\subparagraph{Note} In both cases above $\Cparam = n$. However, on some domains $\Cparam$ is smaller. As an example, consider the ``slice'' $\binom{[n]}{k}$, which is the set of all vectors in $\bits^n$ with Hamming weight $k$. If $a+b > k$ then the two partial inputs $\{(1,1),\ldots,(a,1)\}$ and $\{(a+1,1),\ldots,(a+b,1)\}$ conflict but are not separated by any query. One can check that $\Cparam = \lfloor \min(k,n-k)/2 \rfloor$.

\paragraph{Sensitivity ratio} This double parameter is the most complicated to explain. When bounding the certificate complexity by the block sensitivity in the classical case, we need to show that block sensitivity is always witnessed by blocks whose size is at most the sensitivity. To do so, we show that a larger block can always be shortened by removing one of the elements. This is essentially because there are many ways of shortening a large block.

The corresponding property in our setup is a bit harder to state, and in fact we will have two different versions. The \emph{block sensitivity ratio} $\Bparam$ is the largest parameter such that for every distinct $x,y \in \domain$ there exist distinct $z_1,\ldots,z_s \in \domain$, with $s \geq \Bparam |x \setminus y|$, satisfying:
\begin{enumerate}[(a)]
\item $x \setminus z_i \subsetneq x \setminus y$	.
\item The sets $y \setminus z_i$ are disjoint.
\end{enumerate}
Here is the idea behind this definition. We start with two sets $x,y \in \domain$; one should think of $y$ as one of the sets in the definition of $\bs(f,x)$. We want to find many disjoint ways of bringing $y$ closer to $x$: these are the sets $z_1,\ldots,z_s$. The first constraint states that $z_i$ is closer to $x$ than $y$: we obtain $z_i$ from $y$ by ``fixing'' some of the disagreements with $x$. The second constraint states that the parts of $y$ that had to be fixed are disjoint for different $z_i$; this will be useful when relating $|x \setminus y|$ to  block sensitivity.

The \emph{sensitivity ratio} $\Bcparam$ is defined in the same way, but we also require $|y \setminus z_i| = \chunk$. This will allow us to relate $|x \setminus y|$ to sensitivity rather than block sensitivity.

\subparagraph{Boolean cube} Given two vectors $x = \{(i,x_i) : i \in [n]\}$ and $y = \{(i,y_i) : i \in [n]\}$, let $S$ be the set of coordinates on which they disagree; note that $|x \setminus y| = |S|$. For each $i \in S$, we define $z^i = (y \setminus \{(i,y_i)\}) \cup \{(i,x_i)\}$. That is, $z^i$ is equal to $y$ except at coordinate $i$, at which it agrees with $x$. Then $x \setminus z^i = (x \setminus y) \setminus \{(i,x_i)\}$, so the first property is satisfied, and $y \setminus z^i = \{(i,y_i)\}$, and so the second property is satisfied. This shows that $\Bparam = \Bcparam = 1$.

\subparagraph{Symmetric group} This case is more complicated. Let us assume for simplicity that $x$ is the identity permutation, and write $y$ as a product of disjoint non-trivial cycles of lengths $\ell_1,\ldots,\ell_m > 1$; note that $|x \setminus y| = \sum_i \ell_i$. Consider a specific cycle of $y$ of length $\ell$, without loss of generality $(1\;2\;\cdots\;\ell)$. The relevant part of $y$ is
\[
 (1,2),(2,3),\ldots,(\ell-1,\ell),(\ell,1).
\]
We form $z_i$ by ``shortcutting'' over $i$, that is, replacing $(i-1,i),(i,i+1)$ with $(i-1,i+1),(i,i)$. The new set $z_i$ is indeed a permutation. Moreover, $x \setminus z_i = (y \setminus z_i) \setminus \{(i,i)\}$, so the first property is satisfied. On the other hand, $y \setminus z_i = \{(i-1,i),(i,i+1)\}$, and so in order to satisfy the second property, we need to choose values of $i$ which are at least $2$ apart: $z_1,z_3$ and so on. The worst case is when $\ell = 3$, in which case we can only choose a single value of $i$. For general $\ell$, we can always choose at least $\ell/3$ many $i$'s which will satisfy the second property. Using the same construction for all cycles (the second property is automatically satisfied), we see that $\Bparam = \Bcparam = 1/3$.

\subparagraph{Note} In both examples, the constructions bounded $\Bcparam$ directly. However, in more complicated cases, similar arguments will result in  changes which will be larger than the $\chunk$ threshold.

\section{Relating all measures except sensitivity} \label{sec:main}

In this section, we prove our main theorem, relating the various complexity measures introduced in \Cref{sec:setup}.

\begin{restatable}{theorem}{main} \label{thm:main}
Let $(\domain,\universe,n)$ be a composable domain with parameters $\Dparam,\Cparam,\Bparam$.

Every function $f\colon \domain \to \bits$ satisfies:
\begin{gather}
\sqrt{\bs(f)/6} \leq \adeg(f) \leq \deg(f) \leq D(f) \notag \\
\bs(f) \leq \fbs(f) \leq C(f) \leq U(f) \leq D(f) \notag\\
\bs(f) \leq 3R(f) \leq 3R_0(f) \leq 3D(f) \notag \\
\adeg(f) \leq 2Q(f) \leq 2Q_E(f) \leq 2D(f) \notag \\
D(f) \leq \Bparam^{-2} \max\left(\Dparam,\frac{n}{\Cparam}\right) \bs(f)^4 \tag{+} \label{eq:main-relation}
\end{gather}

In particular, if $\Dparam = O(1)$, $\Cparam = \Omega(n)$, and $\Bparam = \Omega(1)$ then all complexity measures above are polynomially related.

Furthermore, we can strengthen \eqref{eq:main-relation} to
\[
 D(f) \leq \Bparam^{-2} \Dparam \bs(f)^4
\]
if any of the following conditions holds:
\begin{enumerate}[(a)]
\item $C(f) \leq \Cparam$ (implied by the same bound on $U(f),D(f)$). \label{item:bound-C}
\item $\bs(f) \leq \sqrt{\Bparam\Cparam}$ (implied by the same bound on $\fbs(f),R(f)/3,R_0(f)/3$). \label{item:bound-bs}
\item $\adeg(f) \leq \sqrt[4]{\Bparam\Cparam/36}$ (implied by the same bound on $\deg(f),Q(f)/2,Q_E(f)/2$). \label{item:bound-adeg}
\end{enumerate}

In particular, if $\Dparam = O(1)$, $\Bparam = \Omega(1)$, and one of the conditions above holds, then all complexity measures above are polynomially related.
\end{restatable}

In all cases we consider, we will have $\Dparam = O(1)$ and $\Bparam = \Omega(1)$. In some cases, such as the unbalanced slice, $\Cparam$ will be much smaller than $n$, and so the polynomial relation only holds for functions of low complexity. This is not just an artifact of our proof. As an extreme example, consider the subdomain of the Boolean cube consisting of all vectors with Hamming weight~$1$. Every function $f$ has certificate complexity $C(f) = 1$, but balanced functions $f$ have decision tree complexity $D(f) = \lfloor n/2 \rfloor$.

\smallskip

We also prove the following ``ball property'', generalizing a similar result in~\cite{GNSTW}. In this theorem, a ball of radius $r$ around a point $x \in \domain$ consists of all points $y \in \domain$ such that $|x \setminus y| = |y \setminus x| \leq r$.

\begin{restatable}{theorem}{ball} \label{thm:ball-property}
Let $(\domain,\universe,n)$ be a composable domain with chunk size $\chunk$ and sensitivity ratio $\Bcparam$.

If $f\colon \domain \to \bits$ has sensitivity $s = s(f)$, then $f$ can be recovered from its evaluation on a ball of radius $r = \Bcparam^{-1} (2s+1)$ around an arbitrary point.
\end{restatable}

(Recall that $s(f) \leq \bs(f)$ by \Cref{lem:simple-relations}\eqref{item:s-bs}.)

\smallskip

Finally, let us briefly discuss one property of Boolean functions on the Boolean cube which fails on other domains: the junta property. Nisan and Szegedy~\cite{NisanSzegedy} showed that Boolean degree~$d$ functions depend on at most $d2^{d-1}$ coordinates; this was later improved to $O(2^d)$~\cite{CHS,Wellens2019ATB}. Similar results were proved for the slice~\cite{FI2} and multislice~\cite{FODW}; the latter results hold unless the slice or multislice is extremely unbalanced.

In contrast, a similar junta property fails already for the symmetric group. Consider the function
\[
 f(x) = \sum_{i=2}^n x_{1i} x_{i1},
\]
which states that $1$ participates in a $2$-cycle. This function has degree~$2$, but depends on all ``coordinates'', that is, touches all queries. Contrast this with the case of Boolean degree~$1$ functions on the symmetric group: such functions are known to depend on the answer of a single query~\cite{EFP}.

\subsection{Main theorem}

Our proof of \Cref{thm:main} closely follows the exposition in the survey paper of Buhrman and de~Wolf~\cite{BuhrmandeWolf}. In some cases, we simply reduce to known results on the Boolean cube.

In all lemmas below, $f\colon \domain \to \bits$ is a Boolean function on some composable domain. We will make use of the parameters $\Dparam,\Cparam,\Bparam,\Bcparam$ defined in \Cref{sec:setup-parameters}, which we briefly recall:
\begin{itemize}
\item $\Delta$, the \emph{maximum degree	}, is the maximum number of queries that can mention a point.
\item $\Cparam$, the \emph{conflict bound}, is the maximum value such that if $c_1,c_2$ are two conflicting partial inputs of size at most $\Cparam$ then there is a query $Q$ that separates them.
\item $\Bparam$, the \emph{block sensitivity ratio}, is the maximum value such that for every distinct $x,y \in \domain$ there exist at least $\Bcparam |x \setminus y|$ distinct $z_i \in \domain$ such that $x \setminus z_i \subsetneq x \setminus y$ for all $i$, and the sets $y \setminus z_i$ are disjoint.
\item $\Bcparam$, the \emph{sensitivity ratio}, is defined similarly to $\Bparam$, with the additional promise that $|y \setminus z_i| \leq \chunk$ for all $i$.
\end{itemize}

We start by bounding the decision tree complexity in terms of the certificate complexity.

\begin{lemma} \label{lem:D-C}
If $C(f) \leq \Cparam$ then $D(f) \leq \Dparam C_0(f) C_1(f)$.

For all $f$ we have $D(f) \leq \max(\Dparam,n/\Cparam) C_0(f) C_1(f)$.
\end{lemma}
\begin{proof}
Let us start by noting that the second statement follows immediately from the first: if $C(f) > \Cparam$	then $C_0(f) C_1(f) > \Cparam$, and so the inequality follows from $D(f) \leq n$. Also, the first statement trivially holds when $f$ is constant.
We prove the first statement when $f$ is not constant by providing a query algorithm, \Cref{alg:query}.

\begin{algorithm} 
\SetKwFor{Repeat}{repeat}{:}{}
\KwIn{$x \in \domain$}
\KwOut{$f(x)$}
\BlankLine
initialize $z \gets \emptyset$ \tcp*[r]{$z$ is a partial input representing our current knowledge of $x$}
\Repeat{$C_0(f)$ times}{
  let $c$ be a $1$-certificate of $f$ not conflicting with $z$ \; \label{line:c}
  ask all queries containing elements in $c$, and add the answers to  $z$ \; \label{line:query}
  \lIf{$z$ is a $b$-certificate}{\Return{$b$}} \label{line:retb}
}
\caption{Query algorithm proving $D(f) \leq \Dparam C_0(f) C_1(f)$ whenever $C(f) \leq \Cparam$} \label{alg:query}
\end{algorithm}

It is easy to check that the algorithm makes at most $\Dparam C_0(f) C_1(f)$ queries (since $c$ contains at most $C_1(f)$ elements, and each of them appears in at most $\Dparam$ queries), and that it outputs the correct value in \cref{line:retb}. To complete the proof, we will show that in \cref{line:c} a certificate $c$ always exists, and that a value is always returned.

We start by showing that a certificate $c$ always exists in \cref{line:c}. This is the case in the first iteration, since we assume that $f$ is not constant. In any subsequent iteration, if every $1$-certificate of $f$ conflicts with $z$ then this shows that $z$ is a $0$-certificate, and so the algorithm would have returned~$0$ in \cref{line:retb} during the preceding iteration.

It remains to show that the algorithm always returns some value. Suppose first that $f(x) = 0$. We will show that each time that \cref{line:query} is executed, at least one new element of $c_0$ is added to $z$. Therefore, after at most $C_0(f)$ iterations, $z$ will contain $c_0$, and so \cref{line:retb} will return $0$.

Consider, therefore, some iteration of the loop, and let $c$ be the $1$-certificate considered in \cref{line:c}. Since $|c|,|c_0| \leq C(f) \leq \Cparam$, there exists a query $Q \in \queries$ which separates $c$ and $c_0$, say $Q \cap c = \{a\}$ and $Q \cap c_0 = \{a_0\}$. The query $Q$ gets asked in \cref{line:query} since $a \in c$. In the same line, $a_0$ gets added to $z$. To complete the proof, it suffices to show that $a_0$ did not belong to $z$. Indeed, if $a_0$ had belonged to $z$, then the certificate $c$ couldn't have been chosen in \cref{line:c}, since $c$ and $z$ would have conflicted: every set in $\domain$ intersects $Q$ exactly once, and so no such set can contain both $a$ and $a_0$.

Finally, let us consider the case $f(x) = 1$. Above we have shown that if $f(x) = 0$ then after at most $C_0(f)$ iterations of the loop, $z$ contains a $0$-certificate. Therefore if execution reaches the last iteration and \cref{line:retb} does not return~$0$, then necessarily $f(x) = 1$. In other words, at the conclusion of the final iteration, $z$ is a $1$-certificate, and so \cref{line:retb} returns~$1$, completing the proof.
\end{proof}

In order to bound certificate complexity in terms of other measures, we first need to show that the size of every ``minimal block'' can be bounded in terms of sensitivity or block sensitivity.

\begin{lemma} \label{lem:minimal-block}
Suppose that $f(x) \neq f(y)$. There exists $z \in \domain$ satisfying the following properties:
\begin{enumerate}[(a)]
\item $f(z) = f(y)$. \label{item:mb-1}
\item $x \setminus z \subseteq x \setminus y$. \label{item:mb-2}
\item $|x \setminus z| \leq \min(\Bcparam^{-1} s(f), \Bparam^{-1} \bs(f))$. \label{item:mb-3}
\end{enumerate}
\end{lemma}
\begin{proof}
Among all $z \in \domain$ satisfying \eqref{item:mb-1} and \eqref{item:mb-2}, we choose one that minimizes $|x \setminus z|$. We will show that this $z$ satisfies \eqref{item:mb-3}.

By definition of $\Bcparam$, we can find distinct $w_1,\ldots,w_s \in \domain$, where $s \geq \Bcparam |x \setminus z|$, such that
\begin{enumerate}[(i)]
\item $x \setminus w_i \subsetneq x \setminus z$	 for all $i$.
\item The sets $z \setminus w_i$ are disjoint.
\item $|z \setminus w_i| \leq \chunk$.
\end{enumerate}
Since $x \setminus w_i \subseteq x \setminus y$ and $|x \setminus w_i| < |x \setminus z|$, the minimality of $|x \setminus z|$ guarantees that $f(w_i) \neq f(z)$. This implies that $s \leq s(f)$ by definition of sensitivity, and so $|x \setminus z| \leq \Bcparam^{-1} s \leq \Bcparam^{-1} s(f)$.

The proof that $|x \setminus z| \leq \Bparam^{-1} \bs(f)$ is very similar, and left to the reader.
\end{proof}

Armed with the preceding lemma, we can bound certificate complexity in terms of sensitivity and block sensitivity.

\begin{lemma} \label{lem:C-s-bs}
For all $f$ we have $C(f) \leq \min(\Bcparam^{-1} s(f) \bs(f),\Bparam^{-1} \bs(f)^2)$.	
\end{lemma}
\begin{proof}
We will bound $C(f,x)$ for every $x \in \domain$. Let $b = \bs(f,x) \leq \bs(f)$, say witnessed by $y_1,\ldots,y_b$. Thus $f(y_1),\ldots,f(y_b) \neq f(x)$, and the sets $x \setminus y_i$ are disjoint. \Cref{lem:minimal-block} shows that we can find $z_1,\ldots,z_b$ such that $f(z_1),\ldots,f(z_b) \neq f(x)$, the sets $x \setminus z_i$ are disjoint, and furthermore $|x\setminus z_1|,\ldots,|x\setminus z_b| \leq \min(\Bcparam^{-1} s(f), \Bparam^{-1} \bs(f))$.

We claim that $c = (x \setminus z_1) \cup \cdots \cup (x \setminus z_b)$ is a certificate for $x$. Indeed, suppose that $f(w) \neq f(x)$ although $w \supseteq c$. For each $i$, since $w \supseteq c \supseteq x \setminus z_i$, we get that $x \setminus w$ is disjoint from $x \setminus z_i$. But this implies that $z_1,\ldots,z_b,w$ also satisfy the conditions in the definition of $\bs(f,x)$, and we reach a contradiction.

Clearly $|c| \leq \bs(f) \cdot \min(\Bcparam^{-1} s(f), \Bparam^{-1} \bs(f))$, completing the proof.
\end{proof}

To complete the chain of inequalities, we bound block sensitivity in terms of approximate degree.

\begin{lemma} \label{lem:bs-adeg}
For all $f$ we have $\bs(f) \leq 6\adeg(f)^2$.	
\end{lemma}
\begin{proof}
The proof is by reduction to the classical case of the Boolean cube.

We will show that for every set	$x \in \domain$, the bound $\bs(f,x) \leq 6\adeg(f)^2$ holds. Suppose that $b = \bs(f,x)$, and let this be witnessed by $y_1,\ldots,y_b$. Thus $f(y_i) \neq f(x)$ and the sets $x \setminus y_i$ are disjoint. \Cref{lem:bs-symmetric} shows that the sets $y_i \setminus x$ are also disjoint, and so all of the following are disjoint:
\[
 x \cap y_1 \cap \cdots \cap y_b, \quad x \setminus y_1, \ldots, x \setminus y_b, \quad y_1 \setminus x, \ldots, y_b \setminus x.
\]
Furthermore, by composability, for each $s_1,\ldots,s_b \in \bits$, the following set is in $\domain$:
\[
 x(s_1,\ldots,s_b) = (x \cap y_1 \cap \cdots \cap y_b) \cup \bigcup_{i\colon s_i=0} (x \setminus y_i) \cup \bigcup_{i\colon s_i=1} (y_i \setminus x).
\]

Here is another way to describe the sets $x(s_1,\ldots,s_b)$. Let $s_1,\ldots,s_b \in \bits$, and define an assignment to variables $x_e$, for $e \in \universe$:
\begin{equation} \tag{$\ast$} \label{eq:substitution}
 x_e =
 \begin{cases}
 	1 & \text{if } e \in x \cap y_1 \cdots \cap y_b, \\
 	\overline{s_i} & \text{if } e \in x \setminus y_i, \\
 	s_i & \text{if } e \in y_i \setminus x, \\
 	0 & \text{otherwise}.
 \end{cases}
\end{equation}
The set $x(s_1,\ldots,s_b)$ results from applying this assignment, and interpreting the result as the characteristic vector of a set.

Let us now define a function on the Boolean cube $\bits^b$:
\[
 g(s_1,\ldots,s_b) = f(x(s_1,\ldots,s_b)).
\]
By construction, $\bs(g) = b$: indeed, $g(0,\ldots,0) = f(x)$ while $g(0,\ldots,0,1,0,\ldots,0) = f(y_i)$, where the $1$ is in the $i$'th coordinate. We will soon show that $\adeg_\epsilon(g) \leq \adeg_\epsilon(f)$ for all $\epsilon$. The lemma then follows from the well-known bound $\bs(g) \leq 6\adeg(g)^2$~\cite[Lemma 3.5]{NisanSzegedy}.

Recall that $\adeg_\epsilon(f)$ is the minimum degree of a polynomial $P$ such that $|f(x) - P(x)| \leq \epsilon$ for all $x \in \domain$. If we apply substitution~\eqref{eq:substitution} to $P$, we get a polynomial $Q$ such that $|g(s) - Q(s)| = |f(x(s)) - P(x(s))| \leq \epsilon$ for all $s \in \{0,1\}^b$. Since~\eqref{eq:substitution} is an affine substitution, clearly $\deg Q \leq \deg P$. Therefore $\adeg_\epsilon(g) \leq Q$.
%
\end{proof}

We will need a strengthening of this result for the case $\deg(f) = 1$ in \Cref{sec:degree1-pm}.

\begin{lemma} \label{lem:bs-deg-1}
If $\deg(f) \leq 1$	then $\bs(f) \leq 1$.
\end{lemma}
\begin{proof}
Construct the function $g\colon \bits^{\bs(f)} \to \bits$ as in the proof of \Cref{lem:bs-adeg}. By construction, $\deg(g) \leq \deg(f) \leq 1$ and $\bs(g) = \bs(f)$. A Boolean degree~$1$ function on the Boolean cube is a dictator, and in particular, its block sensitivity is at most~$1$. Therefore $\bs(f) = \bs(g) \leq 1$.
\end{proof}

We can now prove \Cref{thm:main}.

\main*

\begin{proof}
The first four lines of inequalities follow from \Cref{lem:simple-relations} and \Cref{lem:bs-adeg}:
\begin{gather*}
\sqrt{\bs(f)/6} \stackrel{\ref{lem:bs-adeg}}{\leq}
\adeg(f) \stackrel{\ref{lem:simple-relations}\eqref{item:adeg-deg}}{\leq}
\deg(f) \stackrel{\ref{lem:simple-relations}\eqref{item:deg-D}}{\leq}
D(f), \\
\bs(f) \stackrel{\ref{lem:simple-relations}\eqref{item:bs-fbs-C}}{\leq}
\fbs(f) \stackrel{\ref{lem:simple-relations}\eqref{item:bs-fbs-C}}{\leq}
C(f) \stackrel{\ref{lem:simple-relations}\eqref{item:C-U}}{\leq}
U(f) \stackrel{\ref{lem:simple-relations}\eqref{item:U-D}}{\leq}
D(f), \\
\bs(f) \stackrel{\ref{lem:simple-relations}\eqref{item:bs-R}}{\leq}
3R(f) \stackrel{\ref{lem:simple-relations}\eqref{item:R-R0}}{\leq}
3R_0(f) \stackrel{\ref{lem:simple-relations}\eqref{item:R0-D}}{\leq}
3D(f), \\
\adeg(f) \stackrel{\ref{lem:simple-relations}\eqref{item:adeg-Q}}{\leq}
2Q(f) \stackrel{\ref{lem:simple-relations}\eqref{item:Q-QE}}{\leq}
2Q_E(f) \stackrel{\ref{lem:simple-relations}\eqref{item:QE-D}}{\leq}
2D(f).
\end{gather*}
The fourth line \eqref{eq:main-relation} follows by combining \Cref{lem:D-C} and \Cref{lem:C-s-bs}:
\[
 D(f) \stackrel{\ref{lem:D-C}}\leq \max\left(\Dparam,\frac{n}{\Cparam}\right) C(f)^2 \stackrel{\ref{lem:C-s-bs}}\leq \Bparam^{-2} \max\left(\Dparam,\frac{n}{\Cparam}\right) \bs(f)^4.
\]
If $C(f) \leq \Cparam$, then according to \Cref{lem:D-C} we can replace $\max(\Dparam,n/\Cparam)$ with $\Dparam$. This explains condition~\eqref{item:bound-C}. Condition~\eqref{item:bound-bs} follows from the inequality $C(f) \leq \Bparam^{-1} \bs(f)^2$ of \Cref{lem:C-s-bs}, and condition~\eqref{item:bound-adeg} then follows from the inequality $\bs(f) \leq 6\adeg(f)^2$ of \Cref{lem:bs-adeg}.
\end{proof}

\subsection{Ball property}

Our proof of \Cref{thm:ball-property} closely follows the argument of Gopalan et al.~\cite{GNSTW}.

\ball*

\begin{proof}
Suppose that we are given the values of $f$ at all points at distance at most $r$ from some point $x \in \domain$, where the distance between two points $x,y \in \domain$ is $d(x,y) = |x \setminus y| = |y \setminus x|$. We will show that the value of $f(y)$ at a point at distance $d \geq r$ from $x$ can be determined from the value of $f$ at points at distance less than $d$ from $x$, and so the entire function $f$ can be constructed step by step from its values on the initial ball.

Let $y$ be a point at distance $d \geq r$ from $x$. By definition of $\Bcparam$, we can find distinct $z_1,\ldots,z_t \in \domain$, with $t \geq \Bcparam d(x,y) \geq 2s+1$, such that $d(x,z_i) < d$, the sets $y \setminus z_i$ are disjoint, and $|y \setminus z_i| = \chunk$. We will show that $f(y)$ is the majority value of $f(z_1),\ldots,f(z_{2s+1})$, completing the proof.

Indeed, suppose that $f(y)$ were not the majority value of $f(z_1),\ldots,f(z_{2s+1})$. This means that at least $s+1$ of the points $z_1,\ldots,z_{2s+1}$ satisfy $f(z_i) \neq f(y)$. However, since the sets $y \setminus z_i$ are disjoint and $|y \setminus z_i| \leq \chunk$, that would imply that $s(f,y) \geq s+1$, contradicting $s(f) = s$.
\end{proof}

\section{Examples of domains} \label{sec:domains}

So far we have seen two examples of domains: the Boolean cube, and the symmetric group. In this section we describe three families of domains: product domains (generalizing the Boolean cube), perfect matching domains (generalizing the symmetric group), and multislices. We close the section by briefly discussing spectral notions of degree.

\subsection{Product domains} \label{sec:domains-product}

The prototypical product domain is the Boolean cube $\bits^n$. Much of the theory of the Boolean cube extends to the so-called Hamming scheme $\{1,\ldots,m\}^n$. Here we consider a slightly more general domain, in which the number of values in each coordinate could depend on the coordinate:
\[
 H(m_1,\ldots,m_n) = \bigtimes_{i=1}^n \{1,\ldots,m_i\}.
\]

\paragraph{Formal definition and composability}

The universe is
\[
 \universe = \{ (i,j) : 1 \leq i \leq n, 1 \leq j \leq m_i \}.
\]

The sets in the domain are
\[
 \domain = \bigl\{ \{ (1,j_1), \ldots, (n,j_n) \} : 1 \leq j_i \leq m_i \bigr\}.
\]

As in the case of the Boolean cube, the chunk size is $\chunk = 1$.

\subparagraphit{Queries} Thinking of the input as a vector $x_1,\ldots,x_n$, we allow queries of the form ``$x_i = ?$''. Formally,
\[
 \queries = \bigl\{ \{(i,1),\ldots,(i,m_i)\} : 1 \leq i \leq n \bigr\}.
\]

\subparagraphit{Composability} The domain $\domain$ consists of all sets intersecting each query at exactly one point, hence is composable due to \Cref{lem:composability}.

\paragraph{Parameters} We now calculate the maximum degree, conflict bound, and sensitivity ratio.

\subparagraphit{Maximum degree} Every element $(i,j) \in \universe$ only participates in the query ``$x_i = ?$'', so $\Dparam = 1$.

\subparagraphit{Conflict bound} Two partial inputs are conflicting if they specify different values for some coordinate $i$. Such inputs can be separated by the query ``$x_i = ?$''. Therefore $\Cparam = n$.

\subparagraphit{Sensitivity ratio} Given two vectors $x$ and $y$ at distance $d = |x \setminus y|$, we can form $d$ vectors $z^i$, for each coordinate $i$ for which $x_i \neq y_i$, defined by $z^i = (y \setminus \{(i,y_i)\}) \cup \{(i,x_i)\}$. These vectors satisfy $x \setminus z^i = (x \setminus y) \setminus \{(i,x_i)\}$, and the sets $y \setminus z^i = \{(i,y_i)\}$ are disjoint. Hence $\Bparam = \Bcparam = 1$.

\paragraph{Main result} Applying \Cref{thm:main}, we deduce the following corollary:

\begin{corollary} \label{cor:main-product}
All complexity measures considered in the paper, except for sensitivity, are polynomially related for all functions on $H(m_1,\ldots,m_n)$. Furthermore, the polynomial relations do not depend on the values of $m_1,\ldots,m_n$.
\end{corollary}

In fact, sensitivity is also polynomially related to the other measures, as we show in \Cref{sec:sensitivity-product}.

\subsection{Perfect matching domains} \label{sec:domains-pm}

Perfect matching domains generalize the symmetric group and the perfect matching scheme (the set of all perfect matchings in $K_{2n}$) to hypergraphs.

Let $\lambda = \lambda_1,\ldots,\lambda_m$ be a sequence of positive integers. 
The domain $P(n;\lambda)$ consists of all perfect hypermatchings in a $|\lambda|$-uniform hypergraph, where $|\lambda| = \lambda_1 + \cdots + \lambda_m$. The vertex set of the hypergraph is partitioned into $m$ parts $P_1,\ldots,P_m$, where $P_i$ contains $\lambda_i n$ vertices. The hyperedges consist of a choice of $\lambda_i$ vertices from the the $i$'th part $P_i$, for each $i$. Every perfect hypermatching consists of $n$ hyperedges.

Special cases include the symmetric group $P(n;1,1)$ (perfect matchings in $K_{n,n}$) and the perfect matching scheme $P(n;2)$ (perfect matchings in $K_{2n}$).

To avoid trivialities, we assume that $n,|\lambda| \geq 2$.

\paragraph{Formal definition and composability} For $i \in \{1,\ldots,m\}$, define
\[
 P_i = \{(i,j) : 1 \leq j \leq \lambda_i n \},
\]
and let $P = P_1 \cup \cdots \cup P_m$.
The universe is
\[
 \universe = \{ S \subseteq P : |S \cap P_i| = \lambda_i \}.
\]

The sets in the domain are
\[
 \domain = \bigl\{ \{S_1,\ldots,S_n\} \subseteq \universe : S_1 \cup \cdots \cup S_n = \universe \bigr\}.
\]

The chunk size is $\chunk = 2$. Indeed, given a set $A$, if we remove $S_j$, then the only way to complete it to a set in $\domain$ is by adding $S_j$ back. In contrast, we can switch the $P_1$-parts of $S_1$ and $S_2$ in order to get a point $B$ at distance $|B \setminus A| = 2$.

\subparagraphit{Queries} The queries we allow are of the form ``which hyperedge does vertex $(i,j)$ participate in?''. This generalizes the queries we considered in the case of the symmetric group, and corresponds to the queries ``which vertex is $i$ connected to?'' in the case of the perfect matching scheme. Formally,
\[
 \queries = \bigl\{ \{ S \in \universe : S \ni v \} : v \in P \bigr\}.
\]

\subparagraphit{Composability} Although described slightly differently, $\domain$ consists of all sets intersecting each query at exactly one vertex, hence is composable due to \Cref{lem:composability}.

\paragraph{Parameters} We now bound the maximum degree, conflict bound, and block sensitivity ratio.

\subparagraphit{Maximum degree} A hyperedge $S \in \universe$ appears in one query per each vertex $v \in S$. Hence $\Dparam = |\lambda|$.

\subparagraphit{Conflict bound} Suppose that $C_1,C_2$ are two partial inputs, which we think of as hypermatchings. If $C_1 \setminus C_2$ and $C_2 \setminus C_1$ mention different vertices (that is, different elements of $P$) then the union $C_1 \cup C_2$ is also a hypermatching. It is not hard to check that every hypermatching can be extended to a perfect hypermatching, and so $C_1,C_2$ do not conflict. This means that if $C_1,C_2$ do conflict, then there must exist a vertex $v \in P$ which belongs to different hyperedges in $C_1$ and $C_2$. Hence the two partial inputs can be separated using the question ``which hyperedge does vertex $v$ participate in?''. This shows that $\Cparam = n$. 

\subparagraphit{Block sensitivity ratio} This is the only nontrivial part in the analysis of perfect matching domains. Our bound on the block sensitivity ratio takes inspiration from the case of the symmetric group, but the argument looks rather different.

Consider two perfect hypermatchings $A,B \in \domain$. Recall that our goal is to find as many $C_i$ as possible (in terms of $|A \setminus B|$) such that $A \setminus C_i \subsetneq A \setminus B$, and the sets $B \setminus C_i$ are disjoint.

The idea is to pick a hyperedge $S \in A \setminus B$, and form another perfect hypermatching $C_S$ by modifying $B$ so that it contains $S$. Later on we will show that we can choose many different hyperedges for which the corresponding perfect hypermatchings are disjoint.

Let $T_1,\ldots,T_r$ be the hyperedges in $B \setminus A$ which share vertices with $S$; note that $r \geq 2$, since otherwise $S \in A \cap B$. To form $C_S$ from $B$, we start by replacing $T_1$ with $S$. This means that we are no longer covering vertices in $T_1 \setminus S$, and vertices in $S \setminus T_1$ are covered twice; the two sets contain an equal number of vertices in each part. The vertices covered twice all appear in the sets $T_2,\ldots,T_r$. We modify $T_2,\ldots,T_r$ by replacing the vertices appearing twice with the vertices $T_1 \setminus S$, replacing each vertex by another vertex in the same part. Thus $B \setminus C_S = \{T_1,\ldots,T_r\}$.

To obtain $C_S$ from $B$, we have modified only hyperedges which conflict with $S$, and in particular do not belong to $A$. This shows that $A \setminus C_S \subseteq A \setminus B$. Since $S \in C_S$ by construction, in fact $A \setminus C_S \subsetneq A \setminus B$.

Let us illustrate this process using the case of the symmetric group. We start with two permutations $A,B$. We choose some edge $S = \{(1,i),(2,j)\} \in A \setminus B$ (encoding that $A(i) = j$). The vertices $(1,i),(2,j)$ appear in two edges of $B$, say $T_1 = \{(1,i),(2,J)\}$ and $T_2 = \{(1,I),(2,j)\}$. The permutation $C_S$ is given by
\[
 C_S = B \setminus \bigl\{\{(1,i),(2,J)\},\{(1,I),(2,j)\}\bigr\} \cup \bigl\{\{(1,i),(2,j)\},\{(1,I),(2,J)\}\bigr\}.
\]

The sets $B \setminus C_S$ are not necessarily disjoint. In order for two sets $B \setminus C_{S_1}$ and $B \setminus C_{S_2}$ to intersect, there needs to be a set $T \in B \setminus A$ which intersects both $S_1$ and $S_2$. Since each set $S \in A \setminus B$ intersects at most $|\lambda|$ sets $T \in B \setminus A$ and vice versa, we see that each set $B \setminus C_{S_1}$ intersects at most $|\lambda|(|\lambda|-1)$ sets $B \setminus C_{S_2}$. We immediately obtain $\Bparam \geq 1/(|\lambda|(|\lambda|-1)+1)$.

\subparagraphit{Sensitivity ratio} When forming $C_S$ from $B$, we remove $r \leq |\lambda|$ sets from $B$. Hence if $|\lambda| \leq 2$, then $|B \setminus C_S| \leq \chunk$, implying that we get a bound on $\Bcparam$, namely $\Bcparam \geq 1/3$. This holds for both the symmetric group and the perfect matching scheme.

\paragraph{Main result} Applying \Cref{thm:main}, we deduce the following corollary:

\begin{corollary} \label{cor:main-pm}
All complexity measures considered in the paper, except for sensitivity, are polynomially related for all functions on $P(n;\lambda)$. Furthermore, the polynomial relations do not depend on the value of $n$.
\end{corollary}

When $|\lambda| = 2$, sensitivity is also polynomially related to the other measures, as we show in \Cref{sec:sensitivity-sn} and \Cref{sec:sensitivity-pm}.

\subsection{Multislices} \label{sec:domains-ms}

The \emph{slice}, or \emph{Johnson scheme}, consists of all vectors in the Boolean cube $\bits^n$ with fixed Hamming weight. We consider a multicolored generalization of the slice, known as the \emph{multislice}.

Let $\lambda = \lambda_1,\ldots,\lambda_m$ be a sequence of positive integers summing to $n$; to avoid trivialities, we assume that $m \geq 2$. The multislice $M(\lambda)$ is the subset of $\{1,\ldots,m\}^n$ consisting of all vectors having exactly $\lambda_i$ coordinates labeled~$i$. The set of all functions on the multislice is known in the representation theory of the symmetric group as the \emph{permutation module} $M^\lambda$.

\paragraph{Formal definition and composability} The universe is
\[
 \universe = \{ 1,\ldots, n\} \times \{1, \ldots, m\}.
\]

The sets in the domain are
\[
 \domain = \{ S \subseteq \universe : \#\{ c \colon (i,c) \in S \} = 1, \#\{ i \colon (i,c) \in S \} = \lambda_c \}.
\]
This definition shows that the domain is composable, by \Cref{lem:composability}.

We sometimes think of the domain as the set of vectors $x=(x_1,\ldots,x_n)$, where $x_i \in \{1,\ldots,m\}$, with $\lambda_c$ many elements ``colored'' $c$.

The chunk size is $\chunk = 2$. Indeed, if we ``uncolor'' an element, there is only one way to color it, so $\chunk > 1$. In contrast, we can switch the colors of two elements, and so $\chunk = 2$.

\subparagraphit{Queries} We allow queries of the form ``$x_i = ?$'', that is, ``what is the color of $x_i$?''. Formally:
\[
 \queries = \{ \{i\} \times \{1,\ldots,m\} : 1 \leq i \leq n \}.
\]

\paragraph{Parameters} We now bound the maximum degree, conflict bound, and block sensitivity ratio.

\subparagraphit{Maximum degree} Every element $(i,j)$ appears in exactly one query ``$x_i = ?$'', and so $\Dparam = 1$.

\subparagraphit{Conflict bound} In contrast to the previous domains we consider, the conflict bound in this case is smaller than~$n$. If two conflicting partial inputs $a,b$ disagree on the color of some coordinate $i$, then they can be separated by the query ``$x_i = ?$'', but this need not be the case in general: another way for two partial inputs to conflict is if in total, they specify more than $\lambda_c$ coordinates of color $c$; in all other cases, it is not hard to see that $a,b$ do not conflict. The only way to guarantee that two partial inputs do not specify more than $\lambda_c$ coordinates of some color $c$ is to limit them to size $\Cparam = \lfloor \min(\lambda)/2 \rfloor$, where $\min(\lambda) = \min(\lambda_1,\ldots,\lambda_m)$.

\subparagraphit{Sensitivity ratio} Let us start with the case of the slice, that is $m = 2$. In this case, we are able to get slightly better bounds.

Suppose that $x,y \in \domain$ are at distance $d = |x \setminus y|$. This means that there are $d$ indices whose color is different. Since the total number of indices of each color is the same in both vectors, there must be $d/2$ indices $i_1,\ldots,i_{d/2}$ such that $x_{i_t}=0$ and $y_{i_t}=1$, and $d/2$ indices $j_1,\ldots,j_{d/2}$ where the opposite happens.

We define $z^t$, for $t \in \{1,\ldots,d/2\}$, to result from $y$ by switching the colors of indices $i_t$ and $j_t$. Thus $x \setminus z^t = (y \setminus z^t) \setminus \{(i_t,0),(j_t,1)\}$ and $y \setminus z^t = \{(i_t,1),(j_t,0)\}$, showing that the sets $y \setminus z^t$ are disjoint. This shows that $\Bcparam = 1/2$ when $m = 2$.

\smallskip

For general multislices, we will show that $\Bcparam \geq 1/3$ by reduction to the sensitivity ratio of the symmetric group. Suppose that $x,y \in \domain$ are at distance $d = |x \setminus y|$. Let $\pi \in \Sym$ be a permutation such that $y_i = x_{\pi(i)}$, and furthermore $\pi(i) = i$ if $y_i = x_i$.

Let $\id$ be the identity permutation. By construction, $|\pi \setminus \id| = d$. We have shown in \Cref{sec:domains-pm} that $\Bcparam \geq 1/3$ for $\Sym$, and so applying the sensitivity ratio property to the pair $\id,\pi$ we obtain $\ell \geq d/3$ permutations $\sigma_1,\ldots,\sigma_\ell$ such that
$\id \setminus \sigma_t \subsetneq \id \setminus \pi$, the sets $\pi \setminus \sigma_t$ are disjoint, and $|\pi \setminus \sigma_t| = 2$.

Define $z^t \in \domain$ by $z^t_i = x_{\sigma_t(i)}$, so that
\begin{align*}
&x \setminus y = \{(i,x_i) : x_i \neq x_{\pi(i)}\}, \\
&x \setminus z^t = \{(i,x_i) : x_i \neq x_{\sigma_t(i)}\}, \\
&y \setminus z^t = \{(i,x_{\pi(i)}) : x_{\pi(i)} \neq x_{\sigma_t(i)}\}.
\end{align*}
We have to prove the following properties: $x \setminus z^t \subsetneq x \setminus y$, the sets $y \setminus z^t$ are disjoint, and $|y \setminus z^t| = 2$.

If $x_i \neq x_{\sigma_t(i)}$ then certainly $i \neq \sigma_t(i)$, and so $(i,i) \in \id \setminus \sigma_t$. Therefore $(i,i) \in \id \setminus \pi$, and so $\pi(i) \neq i$. By construction, this implies that $x_i \neq x_{\pi(i)}$. This shows that $x \setminus z^t \subset x \setminus y$.

By assumption, $\id \setminus \sigma_t \subsetneq \id \setminus \pi$, and so there exists $i$ such that $\pi(i) \neq i$ but $\sigma_t(i) = i$. Since $\pi(i) \neq i$, by construction $x_i \neq x_{\pi(i)}$. On the other hand, clearly $x_i = x_{\sigma_t(i)}$, showing that $x \setminus z^t \subsetneq x \setminus y$, proving the first property.

If $(i,x_{\pi(i)}) \in y \setminus z^t$ then $\pi(i) \neq \sigma_t(i)$, and so $(i,\pi(i)) \in \pi \setminus \sigma_t$. Since the sets $\pi \setminus \sigma_t$ are disjoint, so are the sets $y \setminus z^t$, proving the second property. Since $|\pi \setminus \sigma_t| = 2$, also $|y \setminus z^t| = 2$ (since we already know that $y \neq z^t$ and $2$ is the minimal distance), proving the third property.

\paragraph{Main result} Applying \Cref{thm:main}, we deduce the following corollary:

\begin{corollary} \label{cor:main-ms-balanced}
Suppose that $\lambda = \lambda_1,\ldots,\lambda_m$ is a sequence of positive integers summing to $n$, and let $\lambda_{\min} = \min(\lambda_1,\ldots,\lambda_m)$.

All complexity measures considered in the paper, except for sensitivity, are polynomially related for all functions on $M(\lambda)$ whose degree is at most $O_m(\lambda_{\min}^{1/4})$. 

Furthermore, for any constant $c > 0$, if $\lambda_{\min} \geq cn$ then all complexity measures considered in the paper, except for sensitivity, are polynomially related for all functions on $M(\lambda)$ (without any constraint on the degree). Moreover, the polynomial relations do not depend on the value of $n$ (but do depend on $m$ and $c$).
\end{corollary}

In fact, sensitivity can also be added to the list of polynomially related measures, as we show in \Cref{sec:sensitivity-ms}.

\subsection{Spectral notions of degree} \label{sec:domains-degree}

Analysis of Boolean functions~\cite{ODonnell} studies functions on the Boolean cube $\bits^n$ from a spectral perspective. The starting point is the Fourier expansion of a function $f\colon \bits^n \to \RR$:
\[
 f = \sum_{S \subseteq [n]} \hat{f}(S) \chi_S, \text{ where } \chi_S(x_1,\ldots,x_n) = \prod_{i \in S} (-1)^{x_i}.
\]
It is natural to partition the Fourier characters $\chi_S$ into levels according to their size, and to define
\[
 f^{=d} = \sum_{|S|=d} \hat{f}(S) \chi_S. 
\]
The degree of $f \neq 0$ is defined to be the maximal $d$ such that $f^{=d} \neq 0$. This spectral notion of degree coincides with the spatial notion of degree appearing in \Cref{sec:setup-measures}.

Stated differently, we can decompose the space of real-valued functions on $\bits^n$ into $n+1$ subspaces $V_0,\ldots,V_n$, where $V_d$ is spanned by the characters $\chi_S$ for $|S|=d$. The degree of a function $f$ is the maximal $d$ such that the projection of $f$ to $V_d$ is non-zero. Furthermore, the subspaces $V_0,\ldots,V_n$ are orthogonal.

Going in the other direction, given the spatial notion of degree appearing in \Cref{sec:setup-definition}, we can define the space $V_{\leq d}$ of all functions of degree at most $d$, and can then recover $V_d$ as the orthogonal complement of $V_{\leq d-1}$ inside $V_{\leq d}$.

\smallskip

As we explain in \Cref{sec:sensitivity-sn,sec:sensitivity-pm,sec:sensitivity-ms}, there are natural spectral notions of degree for the symmetric group, the perfect matching scheme, and multislices. In all of these cases, the space of real-valued functions on the domain decomposes into isotypic components with respect to the action of the symmetric group. These isotypic components are indexed by partitions, and they can be grouped into levels according to the maximal part. These levels are the analogs of the subspaces $V_d$ described above.

The decomposition $V_0,\ldots,V_n$ for real-valued functions on the Boolean cube can be obtained in a similar way by considering the action of the hyperoctahedral group, which is the symmetry group of the Boolean cube (as a graph or polytope); see for example \cite[\S2.8.1]{Bachoc}.

The theory can likely be extended to other homogeneous product domains, that is, product domains of the form $[m]^n$, by considering the action of the wreath product $\Sym[m] \wr \Sym[n]$.

A different way of recovering the level decomposition in some cases is using the theory of distance regular graphs~\cite{BCN} and cometric association schemes. For example, homogeneous product domains correspond to Hamming schemes, and slices correspond to Johnson schemes.

\smallskip

The partition of the Fourier characters of the Boolean cube into levels mirrors the partition of the Boolean cube itself into levels according to Hamming weight. Moreover, there is a bijection between Fourier characters and points of the Boolean cube, given by $\chi_S \leftrightarrow 1_S$. We explore similar phenomena in \Cref{apx:rsk}, using the framework of generalized permutations and the Robinson--Schensted--Knuth correspondence.

\section{Sensitivity theorems} \label{sec:sensitivity}

\Cref{thm:main} shows that many different complexity measures are all polynomially related. Sensitivity is conspicuously missing. It had long been conjectured that sensitivity can also be added to the mix, but only recently Huang~\cite{Huang} managed to prove this.

\begin{theorem}[Huang's sensitivity theorem] \label{thm:huang}
If $f$ is a function on the Boolean cube then $s(f) \geq \sqrt{\deg(f)}$.	
\end{theorem}

In this section we show that similar sensitivity theorems hold for many of the domains considered in \Cref{sec:domains}: all product domains, the symmetric group, the perfect matching scheme, and multislices. We do so by reducing the sensitivity theorem on these domains to Huang's sensitivity theorem, using basic facts from representation theory. Our technique is unable to capture hypergraphical perfect matching domains, since the representation theory of these domains is not well-understood (this is related to the notorious Schur plethysm problem, see~\cite[Theorem A2.8]{Stanley2}; see also~\cite[\S7.8]{Kerber}).

We introduce our method in \Cref{sec:sensitivity-method}, and prove the various sensitivity theorems in the subsequent subsections.

\subsection{The method} \label{sec:sensitivity-method}

The basic idea of our method is that if a function on some domain has degree~$d$, then this is often witnessed by some pseudo-character $\chi\colon \domain \to \{-1,0,1\}$, which decomposes the domain into parts which behave like a Boolean cube of some possibly smaller dimension $\dprime$.

Given a domain $\domain$, we will say that a pseudo-character $\chi\colon \domain \to \{-1,0,1\}$ is \emph{$\bits^{\dprime}$-inducing} if the support of $\chi$ can be partitioned into subsets $C_i$ of size $2^{\dprime}$, each of them accompanied with a bijection $\phi_i\colon \bits^{\dprime} \to C_i$ satisfying:
\begin{enumerate}[(a)]
\item If $a \in \bits^{\dprime}$ and $b_1,\ldots,b_{\dprime}$ are its neighbors (at Hamming distance~$1$) then $|\phi_i(a) \setminus \phi_i(b_j)| = \chunk$, and furthermore the sets $\phi_i(a) \setminus \phi_i(b_j)$ are disjoint; we say that $\phi_i$ \emph{maps neighbors to disjoint neighbors}. \label{item:neighbors}
\item If $a,b \in \bits^{\dprime}$ are neighbors then $\chi(\phi_i(a)) = - \chi(\phi_i(b))$, that is, the restriction of $\chi$ to $C_i$ behaves like the parity character. \label{item:parity}
\end{enumerate}

If a function $f\colon \domain \to \bits$ has nontrivial correlation with a $\bits^{\dprime}$-inducing pseudo-character, then we can apply Huang's sensitivity theorem to obtain a lower bound on the sensitivity of $f$.

\begin{lemma} \label{lem:sensitivity-chi}
Let $\domain$ be a domain with chunk size $\chunk$. Suppose that a function $f\colon \domain \to \bits$ has nontrivial correlation with a $\bits^{\dprime}$-inducing pseudo-character $\chi$, that is,
\[
 \sum_{x \in \domain} f(x) \chi(x) \neq 0.
\]
Then $s(f) \geq \sqrt{\dprime}$.
\end{lemma}
\begin{proof}
Let $C_i$ be the partition of the support of $\chi$ promised by the definition of $\bits^{\dprime}$-inducing, and let $\phi_i$ be the corresponding embedding functions. Since
\[
 \sum_{x \in \domain} f(x) \chi(x) = \sum_i \sum_{x \in C_i} f(x) \chi(x),
\]	
there must exist some part $C_i$ such that
\[
 \sum_{x \in C_i} f(x) \chi(x) \neq 0.
\]
Define a function $g\colon \bits^{\dprime} \to \bits$ by $g = f \circ \phi_i$. Since $\chi \circ \phi_i$ is the Fourier character $\chi_{\{1,\ldots,\dprime\}}$ and $\phi_i$ is a bijection, it follows that $\deg(g) = \dprime$. Huang's sensitivity theorem shows that $s(g) \geq \sqrt{\dprime}$, and in particular, there is a point $y \in \bits^{\dprime}$ which has at least $s(g)$ neighbors $z_1,\ldots,z_{s(g)} \in \bits^{\dprime}$ satisfying $g(y) \neq g(z_j)$. Since $|\phi_i(y) \setminus \phi_i(z_j)| = \chunk$ and the sets $\phi_i(y) \setminus \phi_i(z_j)$ are disjoint, this shows that $s(f) \geq s(g) \geq \sqrt{\dprime}$.
\end{proof}

In order to apply this lemma, we will need our domains to possess pseudo-characters witnessing the degrees of all functions on the domain. Specifically, we will say that a domain is \emph{$\Sparam$-witnessing} if for every $d \leq n$, the space of functions of degree at most $d$ is spanned by a collection $\chars_d$ of pseudo-characters, and every degree $d$ pseudo-character in $\chars_d$ is $\bits^{\dprime}$-inducing for some $\dprime \geq \Sparam d$.

\begin{theorem} \label{thm:sensitivity}
If a domain $\domain$ is $\Sparam$-witnessing, then every function $f\colon \domain \to \bits$ satisfies
\[
 s(f) \geq \sqrt{\Sparam \deg(f)}.
\]
\end{theorem}
\begin{proof}
Suppose that $\deg(f) = d$. If $f = 0$ then $d = 0$, and so there is nothing to prove. Otherwise, $\sum_{x \in \domain} f(x)^2 \neq 0$, and so, since $f(x)$ is a linear combination of pseudo-characters in $\chars_d$ and has degree $d$, there must be some degree~$d$ pseudo-character $\chi \in \chars_d$ such that $\sum_{x \in \domain} f(x) \chi(x) \neq 0$. By assumption, $\chi$ is $\bits^{\dprime}$-inducing for some $\dprime \geq \Sparam d$, and so \Cref{lem:sensitivity-chi} immediately shows that $s(f) \geq \sqrt{\dprime} \geq \sqrt{\Sparam d}$.
\end{proof}

\subsection{Product domains} \label{sec:sensitivity-product}

Let us recall the definition of product domains:
\[
 H(m_1,\ldots,m_n) = \{1,\ldots,m_1\} \times \cdots \times \{1,\ldots,m_n\}.
\]
We can assume without loss of generality that $m_1,\ldots,m_n \geq 2$. We think of each element of $H(m_1,\ldots,m_n)$ as an $n$-dimensional vector $v$ whose $i$'th entry satisfies $v_i \in \{1,\ldots,m_i\}$.

By definition, a function has degree~$d$ if it is a linear combination of functions of the form
\[
 \llbracket x_{i_1} = j_1 \rrbracket \cdots \llbracket x_{i_e} = j_e \rrbracket,
\]
where $e \leq d$ and we can assume, without loss of generality, that the indices $i_t$ are all different. (Recall that $\llbracket E \rrbracket$ is the indicator variable of the condition $E$.)

We will use the identity
\[
 \llbracket x_i = j \rrbracket = \frac{1}{m_i} \sum_{j'=1}^{m_i} (\llbracket x_i = j \rrbracket - \llbracket x_i = j' \rrbracket) + \frac{1}{m_i} = \frac{1}{m_i} \sum_{j' \neq j} (\llbracket x_i = j \rrbracket - \llbracket x_i = j' \rrbracket) + \frac{1}{m_i}.
\]
This identity shows that every degree~$d$ function is a linear combination of functions of the form
\[
 (\llbracket x_{i_1} = j_1 \rrbracket - \llbracket x_{i_1} = j'_1 \rrbracket) \cdots (\llbracket x_{i_e} = j_e \rrbracket - \llbracket x_{i_e} = j'_e \rrbracket),
\]
where $e \leq d$, all indices $i_t$ are different, and $j_t \neq j'_t$. All such functions are pseudo-characters, that is, are $\{-1,0,1\}$-valued, and they form the collection $\chars_d$. Such a pseudo-character has degree~$d$ if $e = d$ (this requires a short proof, and follows from standard facts about Fourier expansion in Abelian groups).

We claim that the pseudo-character
\[
 \chi =  (\llbracket x_{i_1} = j_1 \rrbracket - \llbracket x_{i_1} = j'_1 \rrbracket) \cdots (\llbracket x_{i_d} = j_d \rrbracket - \llbracket x_{i_d} = j'_d \rrbracket)
\]
is $\bits^d$-inducing. First notice that the support of $\chi$ consists of all vectors in which $x_{i_t} \in \{j_t,j'_t\}$. The support naturally breaks into subcubes $C_y$, where the index $y$ specifies the values on all coordinates other than $i_1,\ldots,i_d$:
\[
 C_y = \{ x \in \domain : x_{i_t} \in \{j_t,j'_t\}, x_i = y_i \text{ for all other } i\}.
\]
For each $C_y$, we have to construct a bijection $\phi_y\colon \bits^d \to C_y$ that maps neighbors to disjoint neighbors and such that $\chi \circ \phi_y$ is the sign character. The bijection is very simple: for $z \in \bits^d$, $\phi_y(z)$ is the vector given by
\[
 \phi_y(z)_i =
 \begin{cases}
 	j_t & \text{if $i = i_t$ and $z_t = 0$}, \\
 	j'_t & \text{if $i = i_t$ and $z_t = 1$}, \\
 	y_i & \text{otherwise}.
 \end{cases}
\]
It is easy to check that $\phi_i$ maps neighbors to disjoint neighbors. For the other property, notice that if $z,w$ are neighbors differing in the $t$'th coordinate then $\chi(\phi_y(w))$ results from $\chi(\phi_y(z))$ by flipping the sign of the factor $\llbracket x_{i_t} = j_t \rrbracket - \llbracket x_{i_t} = j'_t \rrbracket$, and so $\chi(\phi_z(w)) = -\chi(\phi_y(w))$, as needed.

\smallskip

The foregoing shows that $\domain$ is $1$-witnessing. Applying \Cref{thm:sensitivity}, we obtain a generalization of Huang's sensitivity theorem to all product domains:

\begin{corollary} \label{cor:sensitivity-product}
If $\domain$ is a product domain and $f\colon \domain \to \bits$ then
\[
 s(f) \geq \sqrt{\deg(f)}.
\]	
\end{corollary}

Together with \Cref{cor:main-product}, this shows that all complexity measures considered in this paper are polynomially related for functions on product domains.

\subsection{Symmetric group} \label{sec:sensitivity-sn}

Our treatment of the symmetric group will involve basic representation theory, from a slightly unusual perspective. The representation theory of the symmetric group gives an orthogonal decomposition of the space of all real-valued functions on the symmetric group:
\[
 \mathbb{R}[\Sym] = \bigoplus_{\lambda \vdash n} V^\lambda,
\]
where $\lambda \vdash n$ means that $\lambda$ is a partition of $n$, that is, a nonincreasing sequence of positive integers summing to~$n$. The subspaces $V^\lambda$ are known as \emph{isotypic components}. Before describing how these isotypic components look like, let us mention a theorem of Ellis, Friedgut and Pilpel~\cite[Theorem 7]{EFP} relating degree and the isotypic decomposition:

\begin{theorem} \label{thm:sn-spectral}
The degree of a function $f\colon \Sym \to \mathbb{R}$ is the maximum $d$ such that $f$ has a nonzero component in an isotypic component $V^\lambda$ with $\lambda_1 = n-d$.	
\end{theorem}

We now describe a spanning set for $V^\lambda$, which can be decoded, for example, from the treatment in Sagan~\cite[Chapter 2]{Sagan}. Let $\lambda = (\lambda_1,\ldots,\lambda_m)$. A \emph{tabloid of shape $\lambda$} consists of the numbers $1,\ldots,n$ arranged in $m$ lines, the $i$'th line containing $\lambda_i$ numbers. For example, here is a pair of tabloids of shape $(3,2)$:
\ytableausetup{tabloids}
\begin{align*}
&\begin{ytableau}
1 & 2 & 3 \\
4 & 5	
\end{ytableau}	
&\begin{ytableau}
4 & 5 & 3 \\
1 & 2	
\end{ytableau}	
\end{align*}
A pair of tabloids defines a Boolean function on $\Sym$, which indicates that the each row in the first tabloid is mapped by the input permutation to the corresponding row in the second tabloid. For example, the function corresponding to the pair of tabloids above equals $1$ whenever the input permutation $\pi$ satisfies $\pi(\{1,2,3\}) = \{4,5,3\}$ and $\pi(\{4,5\}) = \{1,2\}$; we stress that the order of numbers in each row doesn't affect the definition of the function. Let us denote the function defined in this way for a pair of tabloids $A,B$ by $e_{A,B}$.

For each pair of tabloids $A,B$ of shape $\lambda$ we define a pseudo-character $\chi_{A,B}$, as follows. We fix $A$, and consider all ways to permute the columns of $B$, forming various tabloids $B^\sigma$. We sum the functions $e_{A,B^\sigma}$, twisted by the sign of $\sigma$:
\[
 \chi_{A,B} = \sum_\sigma (-1)^\sigma e_{A,B^\sigma},
\]
where the sum is over all ways to permute the columns of $B$, and $(-1)^\sigma$ is the sign of $\sigma$. As an illustration, in our running example, representing $e_{A,B^\sigma}$ by the tabloid $B^\sigma$, the pseudo-character $\chi_{A,B}$ is equal to
\[
\begin{ytableau}
4 & 5 & 3 \\
1 & 2
\end{ytableau}
-
\begin{ytableau}
1 & 5 & 3 \\
4 & 2
\end{ytableau}
-
\begin{ytableau}
4 & 2 & 3 \\
1 & 5
\end{ytableau}
+
\begin{ytableau}
1 & 2 & 3 \\
4 & 5
\end{ytableau}
\]

The punch line is that $V^\lambda$ is spanned by the pseudo-characters $\chi_{A,B}$, where $A,B$ go over all tabloids of shape $\lambda$. Moreover, if we restrict $A,B$ to \emph{standard} tabloids, in which the numbers increase in each row and each column, then we get a basis for $V^\lambda$; our argument does not use this fact.

Accordingly, we define $\chars_d$ as follows:
\[
 \chars_d = \{ \chi_{A,B} : \text{$A,B$ are of shape $\lambda$, where $\lambda \vdash n$ and $\lambda_1 \geq n-d$} \}.
\]
\Cref{thm:sn-spectral} shows that $\deg(\chi_{A,B})=d$ if $A,B$ have shape $\lambda$, where $\lambda_1 = n-d$. Let us consider such a pseudo-character $\chi = \chi_{A,B}$. We will show that $\chi$ is $\bits^{\dprime}$-inducing for some $\dprime \geq d/2$.

\smallskip

Consider an empty tabloid of shape $\lambda$. Go over all columns, and partition each column into pairs and, if the column has odd length, a singleton. For a column of length $\ell+1$, this results in $\lceil \ell/2 \rceil \geq \ell/2$ pairs. The total number of positions beyond the first row is $d$, so in total this gives $\dprime \geq d/2$ pairs. We identify each pair with a transposition: if the corresponding entries in $A$ are $x,y$, then the transposition switches $\pi(x)$ and $\pi(y)$, an operation that we denote by $\pi^{(x\;y)}$. By construction, the transpositions are disjoint (affect disjoint sets of elements). We denote the set consisting of these transposition by $T$.

The main observation driving our approach is that if $\tau \in T$ then $\chi(\pi) = -\chi(\pi^\tau)$. This follows directly from the definition of $\chi_{A,B}$, since applying a transposition changes the sign of the column permutation. This suggests looking at the subgroup $\langle T \rangle$ generated by the transposition in $T$, and considering the cosets
\[
 C_\pi = \{ \pi^{\sigma} : \sigma \in \langle T \rangle \}.
\]
(Note that every coset is described by $2^{\dprime}$ different permutations $\pi$.) Each such coset is either entirely in the support of $\chi$, or entirely outside its support. The cosets $C_\pi$ in the support thus partition the support of $\chi$.

Let us consider some coset $C_\pi$ in the support of $T$, and write $T = \{\tau_1,\ldots,\tau_{\dprime}\}$. We define a bijection $\phi_\pi\colon \bits^{\dprime} \to C_\pi$ as follows: $\phi_\pi(x) = \pi^\sigma$, where $\sigma$ is the product of $\tau_i$ for all $i$ such that $x_i = 1$. If $x,y \in \bits^{\dprime}$ are neighbors then $\phi_\pi(x),\phi_\pi(y)$ differ by a transposition $\tau \in T$, so $\phi_\pi$ maps neighbors to neighbors. Since the transpositions in $T$ are disjoint, $\phi_\pi$ actually maps neighbors to disjoint neighbors. As already noted above, if $x,y \in \bits^{\dprime}$ are neighbors then also $\phi_\pi(x) = -\phi_\pi(y)$. This completes the proof that $\chi$ is $\bits^{\dprime}$-inducing, for some $\dprime \geq d/2$.

\smallskip

We have shown that the symmetric group is $1/2$-witnessing. Applying \Cref{thm:sensitivity}, we generalize Huang's sensitivity theorem to the symmetric group:

\begin{corollary} \label{cor:sensitivity-sn}
If $f\colon \Sym \to \bits$ then
\[
 s(f) \geq \sqrt{\deg(f)/2}.
\]
\end{corollary}

Together with \Cref{cor:main-pm}, this shows that all complexity measures considered in this paper are polynomially related for functions on the symmetric group.

\medskip

The main insight behind the proof of the sensitivity theorem for the symmetric group is that the pseudo-characters $\chi_{A,B}$ witness the degree of a function. As another illustration of this point of view, we prove the following lemma in \Cref{apx:snstar}.
In this lemma, a \emph{$t$-star} is a set of the form
\[
 \{ \pi \in \Sym : \pi(i_1) = j_1, \ldots, \pi(i_t) = j_t \}.
\]

\begin{restatable}{lemma}{snstar} \label{lem:sn-star}
Suppose that $\mathcal{F}$ is a subset of $\Sym$ contained in a $t$-star, and its characteristic function has degree at most $t$. Then either $\mathcal{F}$ is empty, or it is equal to the $t$-star.
\end{restatable}

\subsection{Perfect matching scheme} \label{sec:sensitivity-pm}

Our treatment of the perfect matching scheme also involves representation theory. We will follow the exposition in Lindzey's Ph.D. thesis~\cite{LindzeyThesis}, using the notation $\pms$ for the collection of all perfect matchings in $K_{2n}$; see also Ceccherini-Silberstein et al.~\cite[Chapter 11]{CSST}. As in the case of the symmetric group, there is an orthogonal decomposition
\[
 \mathbb{R}[\pms] = \bigoplus_{\lambda \vdash n} V^{2\lambda},
\]
where $2\lambda$ is the partition obtained by doubling each part in $\lambda$.
(Note that $\lambda$ is a partition of $n$ rather than of $2n$.) Lindzey~\cite[Theorem 5.1.1]{LindzeyThesis} related degree to this decomposition:

\begin{theorem} \label{thm:pms-spectral}
The degree of a function $f\colon \pms \to \mathbb{R}$ is the maximum $d$ such that $f$ has a nonzero component in an isotypic component $V^{2\lambda}$ with $\lambda_1 = n-d$.
\end{theorem}

Lindzey~\cite[Theorem 5.2.6]{LindzeyThesis} gave a spanning set for $V^{2\lambda}$, in terms of tabloids of shape $2\lambda$. Given a \emph{single} tabloid $A$ of shape $2\lambda$, the function $e_A$ is the indicator of all perfect matchings in which the two vertices in each edge lie in the same row. For example, here is a tabloid of shape $2(3,2)$:
\[
 \begin{ytableau}
 1 & 2 & 3 & 4 & 5 & 6 \\ 7 & 8 & 9 & 10
 \end{ytableau}
\]
The corresponding function $e_A$ equals $1$ whenever the perfect matching consists of a perfect matching over the vertices $\{1,\ldots,6\}$ together with a perfect matching over the vertices $\{7,\ldots,10\}$.

For each tabloid $A$ of shape $2\lambda$ we define a pseudo-character $\chi_A$ by considering all ways of permuting the columns of $A$, and summing them according to the sign of $\sigma$:
\[
 \chi_A = \sum_\sigma (-1)^\sigma e_{A^\sigma},
\]
where $\sigma$ goes over all ways of permuting the columns of $A$. We define
\[
 \chars_d = \{ \chi_A : \text{$A$ is of shape $2\lambda$, where $\lambda \vdash n$ and $\lambda_1 \geq n-d$} \}.
\]
\Cref{thm:pms-spectral} shows that $\chi_A$ has degree $d$ if $A$ is of shape $2\lambda$, where $\lambda_1 = n-d$. We will show that each such pseudo-character is $\bits^{\dprime}$-inducing for some $\dprime \geq d$.

Go over the columns of $A$, and partition each column into pairs and, possible, a singleton. For a column of length $\ell+1$, this results in at least $\ell/2$ pairs. Since there are $2d$ elements beyond the first row, this results in $\dprime \geq d$ pairs. We think of each pair $(a\;b)$ as a transposition, which acts on a perfect matching by switching the mates of vertices $a$ and $b$ (if $a$ and $b$ are an edge, this has no effect). We denote the set consisting of these transpositions by $T$. By construction, if $m$ is a matching and $\tau \in T$ then $\chi_A(m^\tau) = -\chi_A(m)$. 

We can decompose $\pms$ according to the subgroup $\langle T \rangle$ generated by $T$:
\[
 C_m = \{ m^\sigma : \sigma \in \langle T \rangle \}.
\]
Each $C_m$ is either wholly contained in the support of $\chi$, or wholly outside. For every $C_m$ inside the support, we can construct a bijection $\phi_m \bits^{\dprime} \to C_m$ just as in the case of the symmetric group, concluding that $\chi_A$ is $\bits^{\dprime}$ inducing for some $\dprime \geq d$.

\smallskip

We have shown that the perfect matching scheme is $1$-witnessing. Applying \Cref{thm:sensitivity}, we generalize Huang's sensitivity theorem to the perfect matching scheme:

\begin{corollary} \label{cor:sensitivity-pm}
If $f\colon \pms \to \bits$ then
\[
 s(f) \geq \sqrt{\deg(f)}.
\]
\end{corollary}

Together with \Cref{cor:main-pm}, this shows that all complexity measures considered in this paper are polynomially related for functions on the perfect matching scheme.

\subsection{Multislices} \label{sec:sensitivity-ms}

The final domain we consider is the multislice $M(\mu)$, where $\mu$ is a sequence of $m$ positive integers summing to $n$. This is the collection of all vectors in $\{1,\ldots,m\}^n$ in which exactly $\mu_i$ coordinates have the value $i$.

When $\mu$ consists of $n$ many $1$s, then $M(\mu)$ is just the symmetric group. In other words, the symmetric group is a multislice. Conversely, a mild generalization of the argument for the symmetric group applies to general multislices.

The representation theory of the symmetric group refers to multislices as \emph{permutation modules} $M^\mu$. The isotypic decomposition generalizes:
\[
 \mathbb{R}[M^\mu] = \bigoplus_{\mu \trianglelefteq \lambda} V^\lambda,
\]
where $\mu \trianglelefteq \lambda$ means that $\mu$ is dominated by $\lambda$, a partial order whose exact definition is immaterial here. We stress that the isotypic components $V^\lambda$ depend on $\mu$. The relation between degree and this decomposition was worked out by Filmus, O'Donnell and Wu~\cite[Claim 27]{FODW}:

\begin{theorem} \label{thm:ms-spectral}
The degree of a function $f\colon M^\mu \to \mathbb{R}$ is the maximum $d$ such that $f$ has a nonzero component in an isotypic component $V^\lambda$ with $\lambda = n-d$.	
\end{theorem}

The spanning set described in Sagan~\cite[Chapter 2]{Sagan} and mentioned in \Cref{sec:sensitivity-sn} actually applies to arbitrary multislices. In the case of the symmetric group, we had a pair of tabloids, each filled with the numbers $1$ to $n$. In this case, the first tabloid is the same, and the second one has \emph{content} $\mu$, that is, it contains $\mu_i$ copies of the number $i$. For example, if $\mu = (3,3)$ and $\lambda = (4,2)$, then one possible pair of tabloids is
\begin{align*}
&\begin{ytableau}
1&2&3&4\\5&6	
\end{ytableau}
&\begin{ytableau}
1&1&1&2\\2&2	
\end{ytableau}
\end{align*}
The function $e_{A,B}$ corresponding to such a pair is the indicator of all vectors in which the multiset of values in the indices appearing in each row of the first tabloid is the contents of the same row in the second tabloid. For example, the pair indicated above corresponds to all vectors $x \in M(3,3)$ satisfying $\llangle x_1,x_2,x_3,x_4 \rrangle = \llangle 1,1,1,2 \rrangle$ and $\llangle x_5,x_6 \rrangle = \llangle 2,2 \rrangle$, where $\llangle \cdot \rrangle$ defines a multiset. 

We can now define $\chi_{A,B}$ just as in the case of the symmetric group:
\[
 \chi_{A,B} = \sum_\sigma (-1)^\sigma e_{A,B^\sigma},
\]
where $\sigma$ goes over all ways to permute the columns. We can define $\chars_d$ in exactly the same way:
\[
 \chars_d = \{ \chi_{A,B} : \text{$A,B$ are of shape $\lambda$, where $\lambda \vdash n$ and $\lambda_1 \geq n-d$} \}
\]
Following the exact same argument as for the symmetric group, we deduce that the multislice is $1/2$-witnessing. Applying \Cref{thm:sensitivity}, we obtain a sensitivity theorem for multislices:

\begin{corollary} \label{cor:sensitivity-ms}
If $f\colon M(\mu) \to \bits$ then
\[
 s(f) \geq \sqrt{\deg(f)/2}.
\]	
\end{corollary}

Together with \Cref{cor:main-ms-balanced}, this shows that all complexity measures considered in this paper are polynomially related for functions on balanced multislices, and for low complexity functions on arbitrary multislices.

\section{Degree 1 functions} \label{sec:degree1}

\Cref{thm:main} shows that in a composable domain with parameters $\Dparam,\Cparam,\Bparam$ such that $\Cparam \geq 1$, a Boolean degree~$1$ function can be described by a decision tree of depth at most $\Bparam^{-2} \Dparam$ (the constant~$6$ can be removed by replacing \Cref{lem:bs-adeg} with \Cref{lem:bs-deg-1}). In many cases, we can actually say more. For example, it is a classical result that a Boolean degree~$1$ function on the Boolean cube is a \emph{dictator}, that is, depends on at most one coordinate. Ellis, Friedgut and Pilpel~\cite[Corollary 2]{EFP} extended this to the symmetric group:

\begin{restatable}{theorem}{sndictator} \label{thm:sn-dictator}
If $f\colon \Sym \to \bits$ has degree at most $1$ then either $f$ depends only on some $\pi(i)$, or it depends only on some $\pi^{-1}(j)$, where $\pi$ denotes the input permutation.	
\end{restatable}

Similarly, Filmus and Ihringer~\cite{FI1} extended the dictator result to many domains, including product domains and multislices.

Our main goal in this section is to prove a counterpart of \Cref{thm:sn-dictator} for the perfect matching scheme $\pms$:

\begin{restatable}{theorem}{pmdictator} \label{thm:pm-dictator}
If $f\colon \pms \to \bits$ has degree at most $1$ then $f$ either depends only on which vertex gets matched to some vertex $i$, or on whether the perfect matching intersects some triangle $\{i,j\},\{j,k\},\{k,i\}$.	
\end{restatable}

Since a perfect matching intersects a triangle at most once, Boolean functions depending on whether the input intersects a triangle indeed have degree~$1$.

One can capture both theorems using a single formulation: if $f$ is a Boolean degree~$1$ function on $\Sym$ or on $\pms$, then there is a collection of mutually intersecting edges $E$ in the corresponding graph ($K_{n,n}$ for $\Sym$, $K_{2n}$ for $\pms$) such that $f$ depends only on whether the input perfect matching intersects $E$. Conversely, all such functions have degree~$1$.

\smallskip

Our general approach follows that of Ellis, Friedgut and Pilpel: we first describe all \emph{nonnegative} degree~$1$ functions on the perfect matching scheme, and from this deduce the characterization of Boolean degree~$1$ function.

Describing all nonnegative degree~$1$ functions on a domain is essentially the same as determining the \emph{H-representation} of the polytope whose vertices are the characteristic vectors of sets in the domain; an $H$-representation is simply a set of linear inequalities whose common solution in the polytope.

For the case of the symmetric group, the relevant polytope is the Birkhoff polytope, whose H-representation is implicitly used by Ellis et al. For the case of the perfect matching scheme, the relevant polytope is the perfect matching polytope of the complete graph, first described by Edmonds~\cite{Edmonds}.

We illustrate our method by reproving \Cref{thm:sn-dictator} in \Cref{sec:degree1-sn}. The more complicated proof of \Cref{thm:pm-dictator} appears in \Cref{sec:degree1-pm}.

\subsection{Symmetric group} \label{sec:degree1-sn}

We defined the symmetric group $\Sym$ in \Cref{sec:domains-pm} as a collection of sets over the universe $\{\{(1,i),(2,i)\} : 1 \leq i,j \leq n\}$. It will be more convenient to change the universe to the set of all pairs $(i,j)$, where $1 \leq i,j \leq n$. A permutation is then a collection of pairs which contains exactly one pair of the form $(i,\cdot)$ for each $i$, and exactly one pair of the form $(\cdot,j)$ for each $j$. We will identify such a set with its characteristic vector $x$, indexed by pairs $i,j$.

Suppose that $f\colon \Sym \to \RR$ has degree at most $1$. By definition, this means that it is a linear combination of the functions $1,x_{i,j}$. Since $1 = x_{1,1} + \cdots + x_{1,n}$, we can eliminate $1$, and deduce that $f$ can be written in the form
\[
 f = \sum_{i,j=1}^n c_{i,j} x_{i,j}.
\]
This representation is not unique: indeed, the space of all functions of degree at most~$1$ has dimension $(n-1)^2+1$, whereas here we have $n^2$ parameters. Ellis, Friedgut and Pilpel~\cite[Theorem 28]{EFP} show that if $f$ is nonnegative, then we can choose a representation in which $c_{ij} \geq 0$.\footnote{Ellis, Friedgut and Pilpel also claim a similar result for larger degrees (their Theorem 27). However, the theorem is false, as explained in~\cite{EFP-comment}.} We will show this using a (superficially) different argument, and deduce \Cref{thm:sn-dictator}.

The \emph{Birkhoff polytope} $B_n$ is the convex hull of the characteristic vectors of all permutations in $\Sym$. The Birkhoff--von~Neumann theorem shows that the polytope has the following H-representation:
\begin{align*}
&x_{i,j} \geq 0 && \text{for all } 1 \leq i,j \leq n \\
&\sum_{j=1}^n x_{i,j} = 1 && \text{for all } 1 \leq i \leq n \\
&\sum_{i=1}^n x_{i,j} = 1 && \text{for all } 1 \leq j \leq n
\end{align*}

If $f \geq 0$ on $\Sym$, then since $f$ is linear, $f \geq 0$ on the entire polytope $B_n$. Therefore $f^{-1}(0)$ is a (possibly empty) face of $B_n$. In contrast, every face of $B_n$ is specified by a set of tight inequalities. This allows us to describe all nonnegative functions on $\Sym$.

\begin{theorem} \label{thm:sn-nonneg}
If $f\colon \Sym \to \mathbb{R}$ is nonnegative then $f$ is a nonnegative linear combination of the functions $x_{ij}$.	
\end{theorem}
\begin{proof}
We prove the theorem by induction on the support of $f$. The base case, $f = 0$, is trivial.

Let $\min(f) = \min_{x \in \Sym} f(x)$ be the minimum value of $f$.
If $\min(f) > 0$, then we can represent $f - \min(f)$ as a nonnegative linear combination of the functions $x_{ij}$. Since $1 = x_{11} + \cdots + x_{1n}$, it follows that $f$ is also a nonnegative linear combination of the functions $x_{ij}$.

Now suppose that $\min(f) = 0$ but $f \neq 0$. Note first that since the minimum of $f$ over $B_n$ is attained at a vertex, $f$ is nonnegative on the entire polytope $B_n$. Consider the set $F = \{ x \in B_n : f(x) = 0 \}$. Since $F$ is the set of points at which the linear function $f$ is minimized, $F$ is a face of $B_n$, that is,
\[ F = B_n \cap \{ x_{i_1,j_1} = \cdots = x_{i_m,j_m} = 0 \}, \]
for some set of pairs $(i_1,j_1),\ldots,(i_m,j_m)$; note that $m \geq 1$ since $f \neq 0$.

Let $(i,j) = (i_1,j_1)$. The definition of $F$ implies that if $x_{i,j} = 1$ then $f(x) \neq 0$. Let $m$ be the minimal value of $f$ among all points satisfying $x_{i,j} = 1$ (it doesn't matter if we take the minimum over $\Sym$ or over $B_n$, since the minimum is the same). Then $f - mx_{i,j}$ is nonnegative and has smaller support, and so we can represent it as a linear combination of the functions $x_{ij}$. It follows that $f$ can be represented in the same way.
\end{proof}

We deduce \Cref{thm:sn-dictator} following the argument of Ellis, Friedgut and Pilpel.

\sndictator*

\begin{proof}
Every Boolean function is a fortiori nonnegative, and so \Cref{thm:sn-nonneg} shows that we can write
\[
 f = \sum_{i,j} c_{i,j} x_{i,j},
\]
where $c_{i,j} \geq 0$. If $c_{i,j} > 0$ then $f(x) > 0$ whenever $x_{i,j} = 1$. Since $f$ is Boolean, this shows that in fact $f(x) = 1$ whenever $x_{i,j} = 1$, and so the function $f - x_{i,j}$ is also a Boolean degree~$1$ function. A simple induction on the support thus shows that we can write
\[
 f = \sum_{(i,j) \in S} x_{i,j},
\]
for some subset $S$ of pairs. Any two pairs in $S$ must conflict, that is, cannot belong to the same permutation. As edges of $K_{n,n}$, this means that they intersect at a vertex. It is easy to check that a set of edges are pairwise intersecting if and only if they are either all of the form $(i,\cdot)$ or all of the form $(\cdot,j)$. In the former case, $f$ depends only on $\pi(i)$, and in the latter case, $f$ depends only on $\pi^{-1}(j)$.
\end{proof}

In this argument, we were lucky that the functions $x_{ij}$ are all Boolean. This won't be the case for the perfect matching scheme, a difficulty which will require an additional argument to address.

\subsection{Perfect matching scheme} \label{sec:degree1-pm}

The perfect matching scheme $\pms$, as defined in \Cref{sec:domains-pm}, is the collection of all perfect matchings of $K_{2n}$, considered as sets of edges. We will identify a perfect matching with its characteristic vector $x$, indexed by \emph{unordered} pairs $i,j$. As in the case of the symmetric group, using the identity $1 = x_{1,2} + \cdots + x_{1,2n}$ we can represent each degree~$1$ function $f$ (non-uniquely) as a linear combination
\[
 \sum_{i,j} c_{i,j} x_{i,j}.
\]

The \emph{perfect matching polytope} $P_n$ is the convex hull of the characteristic vectors of all perfect matchings in $\pms$. Edmonds~\cite{Edmonds} determined the H-representation of $P_n$ in terms of the cut functions
\[
 \delta_S(x) = \sum_{i \in S} \sum_{j \notin S} x_{i,j}.
\]
Edmonds' theorem~\cite[Chapter 25]{Schrijver} states that the H-representation of $P_n$ is:
\begin{align*}
&x_{i,j} \geq 0 && \text{for all } 1 \leq i \neq j \leq 2n \\
&\delta_{\{i\}}(x) = 1 && \text{for all } 1 \leq i \leq 2n \\
&\delta_S(x) \geq 1 && \text{for all $S \subseteq \{1,\ldots,2n\}$ such that $3 \leq |S| \leq n$ is odd}
\end{align*}
It will be more useful to replace the constraint $\delta_S(x) \geq 1$ with the equivalent constraint $d_S(x) \geq 0$, where
\[
 d_S(x) = \frac{\delta_S(x) - 1}{2}.
\]
If $x$ is a vertex of $P_n$ then $d_S(x)$ is always an integer; indeed, $|S| - \delta_S(x)$ is the number of vertices in $S$ matched inside $S$, which is always even.

We can repeat the argument of \Cref{thm:sn-nonneg}, extending it to the perfect matching scheme.

\begin{theorem} \label{thm:pm-nonneg}
If $f\colon \pms \to \mathbb{R}$ is nonnegative then $f$ is a nonnegative linear combinations of the functions $x_{i,j}$ and $d_S(x)$, where $3 \leq |S| \leq n$ is odd.	
\end{theorem}
\begin{proof}
The proof is by induction on the support of $f$. The base case, $f = 0$, is trivial. If $\min(f) > 0$, then the result follows by inducting on $f - \min(f)$, using the formula $1 = x_{1,2} + \cdots + x_{1,2n}$.

Suppose now that $\min(f) = 0$ and $f \neq 0$, and let $F = \{x \in P_n : f(x) = 0\}$. Since $F$ is the set of points minimizing the linear function $f$ over $P_n$, we see that $F$ is a face of $P_n$. Therefore $F$ is the intersection of $P_n$ with equations of the form $x_{i,j} = 0$ and $d_S(x) = 0$, the latter for odd $3 \leq |S| \leq n$. Since $f \neq 0$, there must be at least one such equation.

If the defining equations of $F$ include an equation $x_{i,j} = 0$, then this means that whenever $x_{i,j} = 1$, we must have $f(x) > 0$. Let $m$ be the minimum value of $f$ on all points at which $x_{i,j} = 1$. Then $f - mx_{i,j}$ is nonnegative and has smaller support, so can be represented inductively as a linear combination of the required form, and the same holds for $f$.

Similarly, if the defining equations of $F$ include an equations $d_S(x) = 0$, then this means that whenever $d_S(x) \geq 1$, we must have $f(x) > 0$. Let $m$ be the minimum value of $f(x)/d_S(x)$ over all points at which $d_S(x) \neq 0$. Then $f - md_S$ is nonnegative and has smaller support, so can be represented inductively in the required form, and the same holds for $f$.
\end{proof}

In order to deduce \Cref{thm:pm-dictator}, we would like to use an argument similar to that of \Cref{thm:sn-dictator}. However, direct replication of the argument fails, since the function $d_S(x)$ is only Boolean if $|S| = 3$. In order to rule out the appearance of $d_S(x)$ for larger $S$, we appeal to \Cref{lem:bs-deg-1}, which states that $\deg(f) \leq 1$ implies $\bs(f) \leq 1$.

\pmdictator*

\begin{proof}
Since the function $f$ is nonnegative, it can be written as
\[
 f = \sum_{i,j} c_{i,j} x_{i,j} + \sum_S c_S d_S,
\]
where $c_{i,j},c_S \geq 0$, and the second sum is over all sets $S$ whose size is odd and satisfies $3 \leq |S| \leq n$.

If $f = 1$ then the theorem is trivial. We will show that if $f \neq 1$ then $c_S = 0$ whenever $|S| \geq 5$.

Suppose, to the contrary, that $c_S > 0$ for $S = \{1,\ldots,2k+1\}$, where $k \geq 2$. Since $f \neq 1$, there is some perfect matching $y$ at which $f(y) = 0$, and so $d_S(y) = 0$. This means that $y$ matches $2k$ of the vertices of $S$, say $\{1,2\},\ldots,\{2k-1,2k\} \in Y$, where $Y$ is the perfect matching whose characteristic vector is $y$. The vertex $2k+1$ is matched to some other vertex, say $2k+2$. Since $2n \geq 2|S| \geq 4k+2 \geq 2k+6$, the perfect matching $Y$ must contain at least two more edges, say $\{2k+3,2k+4\}$ and $\{2k+5,2k+6\}$. So
\[
 \{1,2\},\ldots,\{2k+5,2k+6\} \in Y.
\]
We construct two new perfect matchings:
\begin{align*}
Z &= Y \setminus \{\{1,2\},\{2k+3,2k+4\}\} \cup \{\{1,2k+3\},\{2,2k+4\}\}, \\
W &= Y \setminus \{\{3,4\},\{2k+5,2k+6\}\} \cup \{\{3,2k+5\},\{4,2k+6\}\}.	
\end{align*}
The characteristic vectors $y,z$ of these perfect matchings satisfy $d_S(z) = d_S(w) = 1$, and so $f(z), f(w) > 0$, implying that $f(z) = f(w) = 1$. This shows that $s(f,y) \geq 2$, contradicting the bound $\bs(f) \leq 1$ given by \Cref{lem:bs-deg-1}.

\smallskip

It follows that $f$ is a nonnegative linear combination
\[
 f = \sum_{i,j} c_{i,j} x_{i,j} + \sum_{|S|=3} c_S d_S.
\]
If $c_{i,j} > 0$ then $f(x) > 0$ and so $f(x) = 1$ whenever $x_{i,j} = 1$, implying that $f - x_{i,j}$ is still a Boolean degree~$1$ function. The same holds for $c_S$, since $d_S$ is Boolean when $|S|=3$. A simple induction shows that we can write $f$ as a sum of $x_{i,j}$ and $d_S$. Since $f$ is Boolean, in this sum no two terms can equal~$1$ at the same time.

The term $x_{i,j}$ equals $1$ if the perfect matching contains the edge $\{i,j\}$; we will denote this ``edge event'' by $E_{i,j}$. The term $c_{i,j,k}$ equals $1$ if the perfect matching contains none of the edges $\{i,j\},\{j,k\},\{k,i\}$; we will denote this ``triangle event'' by $\Delta_{i,j,k}$. To complete the proof, we need to consider when two such events, of either type, cannot co-occur. Note that triangle events are only relevant when $2n \geq 6$ (since $|S| \leq n$).

\smallskip

We start by considering two triangle events $\Delta_{S_1},\Delta_{S_2}$. There are three cases to consider, depending on the size of the intersection $S_1 \cap S_2$:
\begin{description}
\item[$S_1,S_2$ disjoint:] Say $S_1 = \{1,2,3\}$ and $S_2 = \{4,5,6\}$. There is a perfect matching containing the edges $\{1,4\},\{2,5\},\{3,6\}$, in which both events occur.
\item[$S_1,S_2$ share an edge:] Say $S_1 = \{1,3,4\}$ and $S_2 = \{2,3,4\}$. There is a perfect matching containing the edges $\{1,2\},\{3,5\},\{4,6\}$, in which both events occur.
\item[$S_1,S_2$ share a vertex:] Say $S_1 = \{1,2,5\}$ and $S_2 = \{3,4,5\}$. There is a perfect matching containing the edges $\{1,3\},\{2,4\},\{5,6\}$, in which both events occur.
\end{description}

This means that the sum representing $f$ contains at most one triangle event.

We move on to consider a triangle event $\Delta_S$ and an edge event $E_T$. Again there are three cases to consider, depending on the size of the intersection $S \cap T$:
\begin{description}
\item[$S,T$ disjoint:] Say $S = \{1,2,3\}$ and $T = \{4,5\}$. Suppose first that $2n \geq 8$. In this case there is a perfect matching containing the edges $\{1,6\},\{2,7\},\{3,8\},\{4,5\}$, in which both events occur. \\
When $2n = 6$, the two events cannot co-occur. A simple case analysis shows that $\Delta_{1,2,3} \lor E_{4,5}$ is equivalent to $E_{1,6} \lor E_{2,6} \lor E_{3,6}$.
\item[$S,T$ share a vertex:] Say $S = \{1,2,3\}$ and $T = \{3,4\}$. There is a perfect matching containing the edges $\{1,5\},\{2,6\},\{3,4\}$, in which both events occur.
\item[$S,T$ share an edge:] Say $S = \{1,2,3\}$ and $T = \{1,2\}$. In this case the events are clearly mutually exclusive. If both of them \emph{do not occur} then the perfect matching doesn't contain $\{1,2\}$ but contains either $\{1,3\}$ or $\{2,3\}$, and vice versa: if the perfect matching contains either $\{1,3\}$ or $\{2,3\}$ then both events do not occur. Therefore $\Delta_{1,2,3} \lor E_{1,2}$ is equivalent to $E_{3,4} \lor \cdots \lor E_{3,2n}$.
\end{description}

This means that the sum representing $f$ either consists of a single triangle event, or otherwise is equivalent to another sum consisting only of edge events. In the latter case, the edges must pairwise intersect, and so are either all adjacent to a single vertex, or else form a triangle.
\end{proof}

\section{Intersecting families} \label{sec:intersecting}

Consider some domain $\domain$. A subset $\mathcal{F} \subseteq \domain$ is \emph{$t$-intersecting} if any two sets $S_1,S_2 \in \mathcal{F}$ have at least $t$ elements in common: $|S_1 \cap S_2| \geq t$. If $t = 1$, then we call $\mathcal{F}$ an \emph{intersecting} family. There is also a bipartite version: two subsets $\mathcal{F}_1,\mathcal{F}_2 \subseteq \domain$ are \emph{cross-$t$-intersecting} (or \emph{cross-intersecting} when $t=1$) if any $S_1 \in \mathcal{F}_1$ and $S_2 \in \mathcal{F}_2$ have at least $t$ elements in common.

Erd\H{o}s, Ko and Rado~\cite{EKR} determined the maximum size of an intersecting family for the domain $\binom{[n]}{k}$, which consists of all subsets of $\{1,\ldots,n\}$ of size $k$. This domain is similar to, but not identical with, the multislice $M(k,n-k)$ (the difference is that we consider $\overline{S_1 \triangle S_2}$ instead of $S_1 \cap S_2$). Ahlswede and Khachatrian~\cite{AK1,AK2} extended this to arbitrary $t$ (for the same domain), and also considered~\cite{AK3} \emph{$t$-agreeing} families in $\{1,\ldots,m\}^n$, which are just $t$-intersecting families in $H(m,\ldots,m)$ ($n$ many copies).

Ellis, Friedgut and Pilpel~\cite{EFP} studied $t$-intersecting families in the symmetric group, following earlier work on the case $t=1$~\cite{DezaFrankl,CameronKu,LaroseMalvenuto}. They showed that for every fixed $t$ and large enough $n$ (large enough in terms of $t$), the maximum size of a $t$-intersecting family is $(n-t)!$, matching the size of \emph{$t$-stars}, which are families of the form
\[
 \{ \pi \in \Sym: \pi(i_1) = j_1, \ldots, \pi(i_t) = j_t \}.
\]
Furthermore, they claimed that for large enough $n$, the $t$-stars are the only $t$-intersecting families of size $(n-t)!$. Unfortunately, their proof is wrong when $t \geq 2$, as pointed out in~\cite{EFP-comment}, although the result does follow from subsequent work of Ellis~\cite{Ellis2}. The argument of Ellis actually proves a much stronger structural result, and is quite complicated. Our goal in this section is to present an alternative proof of this property, known as \emph{uniqueness}, using the connection between degree and certificate complexity.

\smallskip

Ellis et al.\ also studied cross-$t$-intersecting families in the symmetric group, showing that for every fixed $t$ and large enough $n$, the maximum product of sizes of two cross-$t$-intersecting families is $(n-t)!^2$. They claimed that the only extremal examples are when the two families are the same $t$-star. Again the argument is flawed, but the result follows from the work of Ellis. Our technique also applies to this result.

Our arguments apply to other domains as well, such as the perfect matching scheme, simplifying the characterization of maximum size $t$-intersecting families of perfect matchings due to Lindzey~\cite{lindzey2018intersecting,LindzeyThesis}.

\subsection{Main result} \label{sec:intersecting-result}

\paragraph{Setup} Our result is stated in terms of two parameters of a domain $(\domain,\universe,n)$.

\subparagraph{Maximum size of links} For any $S \subseteq \universe$, the \emph{link} of $S$ is
\[
 \domain_S = \{ T : T \in \domain, T \supseteq S \}. 
\]
(We typically consider only non-empty links.) A link $\domain_S$ is a \emph{$t$-link} if $|S|=t$.

We define $\Lparam_t$ to be the maximum size of a $t$-link.

\subparagraph{Intersection bound} Suppose that $x \in \domain$ is a set that $t$-intersects all $y \in \domain$ containing some partial input $C$ (in our case, $C$ will be a $1$-certificate of a $t$-intersecting family). One way in which this can happen is if $C$ itself $t$-intersects $x$, but this can fail if $C$ is very large. For example, in the symmetric group, the identity permutation intersects the single permutation containing $\{1,2\},\{2,3\},\ldots,\{n-1,1\}$. Intersection bounds are bounds on the size of $C$ which guarantee that this strange situation does not happen.

For any integer $t \geq 1$, the \emph{intersection bound} $\Iparam_t$ is the maximal value such that for any $x \in \domain$ and any partial input $C$ of size at most $\Iparam_t$, if $x$ $t$-intersects all total inputs extending $C$ then $|x \cap C| \geq t$.

\medskip

We can now state the main result.

\begin{theorem} \label{thm:intersecting}
Consider a domain $(\domain,\universe,n)$ with parameters $\Lparam_t,\Iparam_t$.

If $\mathcal{F}_1,\mathcal{F}_2 \subseteq \domain$ are cross-$t$-intersecting, and the characteristic function $f_1\colon \domain \to \bits$ of $\mathcal{F}_1$ satisfies $C(f_1) \leq \Iparam_t$, then either $\mathcal{F}_1$ is contained in a $t$-link, or the size of $\mathcal{F}_2$ can be bounded:
\[
 |\mathcal{F}_2| \leq \binom{C(f_1)}{t} C(f_1) \Lparam_{t+1}.
\]
\end{theorem}

Before proving this result, let us indicate why it is useful in the case of the symmetric group (the case of the perfect matching scheme is similar), concentrating on the case $\mathcal{F}_1 = \mathcal{F}_2 = \mathcal{F}$. Using a spectral argument, Ellis et al.~\cite{EFP} showed that if $|\mathcal{F}| = (n-t)!$ and $n$ is large enough, then $\deg(f) \leq t$, which by \Cref{cor:main-pm} implies that $C(f) = O_t(1)$. Therefore either $\mathcal{F}$ is a $t$-star, or it has size $O((n-t-1)!)$, which for large enough $n$ is smaller than $(n-t)!$.

\begin{proof}
We can assume that $\mathcal{F}_2$ is non-empty, since otherwise the result is trivial.

Suppose that $\mathcal{F}_1$ is not contained in any $t$-link. In particular, it is not empty. Let $A$ be an arbitrary $1$-certificate of $f$. Any set $x \in \mathcal{F}_2$ $t$-intersects all total inputs extending $A$, and so, since $|A| \leq C(f_1) \leq \Iparam_t$, we see that $|x \cap A| \geq t$, and in particular $|A| \geq t$.

Let $S \subseteq A$ be an arbitrary subset of $A$ of size $t$. Since $\mathcal{F}_1 \not\subseteq \domain_S$ by assumption, there must be some $x_S \in \mathcal{F}_1$ which doesn't contain $S$. Let $B_S$ be a $1$-certificate for $x_S$. Since $x_S \not\supseteq S$ but $x \supseteq B_S$, necessarily $B_S \not\supseteq S$.

Now consider an arbitrary $x \in \mathcal{F}_2$. Since $x$ $t$-intersects all total inputs extending $A$ and $|A| \leq C(f_1) \leq \Iparam_t$, $x$ must contain some subset $S$ of $A$ of size $t$. Similarly, $x$ must contain some subset $T$ of $B_S$ of size $t$. Since $B_S$ doesn't contain $S$, the subset $T$ must contain some element $e \notin S$. By construction, $x$ belongs to the $(t+1)$-link $\domain_{S \cup \{e\}}$.

There are $\binom{|A|}{t} \leq \binom{C(f_1)}{t}$ choices for $S$ and $|B_S| \leq C(f_1)$ choices for $e$, and so $\mathcal{F}_2$ is covered by $\binom{C(f_1)}{t} C(f_1)$ many $(t+1)$-links, of total size at most $\binom{C(f_1)}{t} C(f_1) \Lparam_{t+1}$.
\end{proof}


In \Cref{sec:intersecting-ib}, we calculate the intersection bounds of the various domains considered in \Cref{sec:domains}. In \Cref{sec:intersecting-spectral} we describe the spectral method used by Ellis, Friedgut and Pilpel~\cite{EFP} and by Lindzey~\cite{lindzey2018intersecting,LindzeyThesis}, and show how to apply \Cref{thm:intersecting} in this setting.

\subsection{Intersection bounds} \label{sec:intersecting-ib}

In this section, we compute or bound the intersection bounds for all domains considered in \Cref{sec:domains}.

\paragraph{Product domains} Recall that the product domain $H(m_1,\ldots,m_n)$ consists of all sets of the form
\[
 \{ (1,j_1), \ldots, (n,j_n) \}, \text{ where } j_i \in \{1,\ldots,m_i\}.
\]
To avoid trivialities, we assume that $m_1,\ldots,m_n \geq 2$.

\begin{lemma} \label{lem:intersection-bounds-product}
The intersection bounds of $H(m_1,\ldots,m_n)$ are $\Iparam_t = \infty$ for all $t$.	
\end{lemma}
\begin{proof}
Suppose that $C$ is a partial input, and that $x$ is an input that $t$-intersects all extensions of $C$. We can think of $x$ as a vector $x_1,\ldots,x_n$, where $x_i \in \{1,\ldots,m_i\}$. Similarly, $C$ is a partial vector. 	We extend $C$ to an input $y$ according to the following rule: if $C_i$ is undefined, we choose $y_i$ to be some element different from $x_i$; this is possible since $m_i \geq 2$. This guarantees that $|x \cap C| = |x \cap y| \geq t$.
\end{proof}

\paragraph{Perfect matching domains} Recall that a perfect matching domain $P(n; \lambda)$ is given by an integer $n \geq 2$ and a sequence $\lambda = \lambda_1,\dots,\lambda_m$ of positive integers summing to $|\lambda| \geq 2$. The domain consists of all $|\lambda|$-uniform hypermatchings in the complete ``$\lambda$-partite'' hypergraph on $|\lambda|n$ vertices. In more detail, the vertices are partitioned into $m$ parts $P_1,\ldots,P_m$, the $i$'th part containing $\lambda_i n$ vertices $(i,1),\ldots,(i,\lambda_i n)$, and the hyperedges consist of a choice of $\lambda_i$ elements from $P_i$ for each $i \in \{1,\ldots,m\}$.

\begin{lemma} \label{lem:intersection-bounds-pm}
The intersection bounds of $P(n; \lambda)$ are $\Iparam_t = n-2$ for all $t \leq n-2$.
\end{lemma}
\begin{proof}
We start by showing that $\Iparam_t \geq n-2$. Let $C$ be a partial input of size at most $n-2$, and let $x$ be an input which $t$-intersects all total inputs extending $C$.

Denote by $V$ the set of vertices mentioned by hyperedges in $C$; thus $C$ is a perfect hypermatching of $V$. Let $x|_{\overline{V}}$ be the set of hyperedges in $x$ only involving vertices outside of $V$. We can complete $x|_{\overline{V}}$ to a perfect hypermatching $z$ of $\overline{V}$.

Without loss of generality, suppose that $\overline{V} \cap P_i$ consists of the vertices $(i,1),\ldots,(i,\lambda_i r)$, where $r = n - |C| \geq 2$, and that $z$ consists of the hyperedges
\[
 \{ (i,1), \ldots, (i,\lambda_i) : 1 \leq i \leq m\},
 \{ (i,\lambda_i+1), \ldots, (i,2\lambda_i) : 1 \leq i \leq m\}, \dots
\]
Let $w$ consist of the hyperedges
\[
 \{ (i, \lambda_ir), (i,1), \ldots, (i,\lambda_i-1) : 1 \leq i \leq m\},
 \{ (i,\lambda_i), \ldots, (i,2\lambda_i-1) : 1 \leq i \leq m\}, \dots
\]
Notice that $z$ and $w$ are disjoint; we can think of $w$ as a generalized derangement.\footnote{Recall that a \emph{derangement} is a permutation in $\Sym$ without fixed points. In our setting, we can think of a derangement as a permutation in $\Sym$ which is disjoint from the identity permutation $\mathsf{id} \in \Sym$. More generally, given a domain $\domain$, we can define $a \in \domain$ to be a derangement with respect to $b \in \domain$ if $a \cap b = \emptyset$. Under this definition (and working inside an appropriate link of $P(n;\lambda)$), $w$ is a derangement with respect to $z$.}

Since $C \cup w$ is a total input extending $C$, we have $|C \cap x| = |(C \cup w) \cap x| \geq t$.

\smallskip

To complete the proof, we show that $\Iparam_t \leq n-2$. Let $x,y$ be two inputs such that $|x \cap y| = t$; such inputs can be constructed using generalized derangements. Let $C$ be obtained from $y$ by removing one of the hyperedges in $x \cap y$. By construction, $|x \cap C| = t-1$. On the other hand, $y$ is the only total input extending $C$, and so $x$ $t$-intersects all total inputs extending $C$.
\end{proof}

\paragraph{Multislices} Recall that a multislice $M(\lambda)$ is specified by a sequence $\lambda_1,\ldots,\lambda_m$ of positive integers summing to $n$, where $m \geq 2$. The multislice consists of all vectors $x \in \{1,\ldots,m\}^n$ (encoded as sets $\{ (i,x_i) : 1 \leq i \leq n \}$) having exactly $\lambda_c$ coordinates equal to $c$.

\begin{lemma} \label{lem:intersection-bounds-ms}
The intersection bounds of $M(\lambda)$ satisfy $\Iparam_t \geq n - 2\max(\lambda_1,\ldots,\lambda_m)$.	
\end{lemma}
\begin{proof}
Let $C$ be a partial input of size at most $n - 2\max(\lambda_1,\ldots,\lambda_m)$, and let $x$ be an input which $t$-intersects all total inputs extending $C$.

The link of $C$ is another multislice $M(\mu)$, where $\mu_1 + \cdots + \mu_m = n - |C|$, and
\[
 \max(\mu_1,\dots,\mu_m) \leq \max(\lambda_1,\dots,\lambda_m) \leq \frac{n-|C|}{2}.
\]

Consider a bipartite graph with $n - |C|$ vertices on both sides. We color the first $\mu_1$ vertices on each side by color~$1$, the following $\mu_2$ vertices by color~$2$, and so on. We connect a vertex of color~$i$ on the left to all vertices on the right of colors different from~$i$.

We claim that any set $S$ on the left has at least $|S|$ neighbors on the right, and so the graph has a perfect matching by Hall's criterion. Indeed, if $S$ contains vertices of more than one color, then its neighborhood consists of all vertices on the right. If all vertices in $S$ are colored~$i$ then $|S| \leq \mu_i$ and $S$ has exactly $n - |C| - \mu_i$ neighbors, which is at least $|S|$ since $\mu_i \leq (n-|C|)/2$.

Using the perfect matching whose existence is promised by Hall's criterion, we can extend $C$ to a total input $y$ which disagrees with $x$ on all indices outside of those mentioned by $C$. It follows that $|x \cap C| = |x \cap y| \geq t$.
\end{proof}

When $\lambda_1 = \cdots = \lambda_m = 1$, this bound coincides with the bound for the symmetric group implied by \Cref{lem:intersection-bounds-pm}.


\subsection{Spectral technique} \label{sec:intersecting-spectral}

The results of Ellis, Friedgut and Pilpel~\cite{EFP} and of Lindzey~\cite{lindzey2018intersecting,LindzeyThesis} are proved using a spectral technique known as the \emph{weighted Hoffman bound}, pioneered in this context by Wilson~\cite{Wilson} and Frankl--Wilson~\cite{FranklWilson}, and later cast in a different form by Friedgut~\cite{Friedgut}; see also the monograph of Godsil and Meagher~\cite{GodsilMeagher}.

In a nutshell, in order to prove a $t$-intersecting theorem for a domain $\domain$, the idea is to construct a $\domain \times \domain$ matrix $A$, supported on pairs on non-$t$-intersecting elements, satisfying certain spectral properties.

\begin{restatable}{definition}{defad} \label{def:adjacency}
A real $\domain \times \domain$ matrix $A$ is \emph{$t$-good} if the following properties hold:
\begin{enumerate}
\item $A$ is symmetric.
\item If $|x \cap y| \geq t$ then $A(x,y) = 0$.
\item $A \mathbf{1} = \mathbf{1}$, where $\mathbf{1}$ is the constant~$1$ vector.
\item If $\deg(f) \leq t$ and $\EE[f] = 0$, where $\EE[f] = \sum_{x \in \domain} f(x)/|\domain|$, then $Af = -\omega f$, where
\[
 \omega = \frac{\Lparam_t}{|\domain| - \Lparam_t}.
\]
\item If $Af = \lambda f$ and $\deg(f) > t$ then $|\lambda| < \omega$.
\end{enumerate}
\end{restatable}

We remark that if we replace $|\lambda| < \omega$ with $\lambda > -\omega$ in the final property, then we can still recover all results about $t$-intersecting families; the stronger guarantee $|\lambda| < \omega$ is only needed to handle cross-$t$-intersecting families.

Ellis, Friedgut and Pilpel~\cite[Theorem 26]{EFP} constructed a $t$-good matrix for $\domain = \Sym$ for large enough $n$ (as a function of $t$). Similarly, Lindzey~\cite{lindzey2018intersecting,LindzeyThesis} constructed a $t$-good matrix for $\domain = \pms$ for large enough $n$ (as a function of $t$).

Wilson~\cite{Wilson} constructed a matrix satisfying a similar condition for $\domain = M(n-k,k)$; however, his definition of $t$-intersecting is different from ours. Frankl and Wilson~\cite{FranklWilson} constructed such a matrix when $\domain$ is the Grassmann scheme, a domain which doesn't fall into our framework (see \Cref{sec:open-questions}). Friedgut~\cite{Friedgut} constructed such a matrix for the Boolean cube $\bits^n$ under a biased measure, also using a different definition of $t$-intersecting.

It is well-known that the existence of a $t$-good operator implies that a $t$-intersecting family contains at most $\Lparam_t$ elements. For completeness, we reproduce the proof in \Cref{apx:spectral}.

\begin{restatable}{proposition}{proib} \label{pro:intersecting-bound}
Suppose that there exists a $t$-good matrix for $\domain$. Then a $t$-intersecting family contains at most $\Lparam_t$ points, and furthermore, the characteristic function of a $t$-intersecting family of size $\Lparam_t$ has degree at most~$t$.
\end{restatable}

Using \Cref{thm:intersecting}, we can show that for large enough $n$, the bound $\Lparam_t$ is achieved only by $t$-links.

\begin{theorem} \label{thm:intersecting-uniqueness}
Suppose that there exists a $t$-good matrix for $\domain$, and every function of degree at most $t$ has certificate complexity at most $\Cbound_t \leq \Iparam_t$. If $\Cbound_t \binom{\Cbound_t}{t} \Lparam_{t+1} < \Lparam_t$ then any $t$-intersecting family of size $\Lparam_t$ is a $t$-link.
\end{theorem}
\begin{proof}
Let $\mathcal{F}$ be a $t$-intersecting family of size $\Lparam_t$, and let $f$ be its characteristic function. According to \Cref{pro:intersecting-bound}, $\deg(f) \leq t$.	 By assumption, $C(f) \leq \Iparam_t$, and so \Cref{thm:intersecting} shows that either $\mathcal{F}$ is contained in a $t$-link, or $|\mathcal{F}| \leq \Cbound_t \binom{\Cbound_t}{t} \Lparam_{t+1} < \Lparam_t$. The second case cannot happen by assumption. We conclude that $\mathcal{F}$ is contained in a $t$-link. Since $|\mathcal{F}| = \Lparam_t$, it has to be the complete $t$-link.
\end{proof}

We can extend this result to cross-$t$-intersecting families. We start with a cross-intersecting version of \Cref{pro:intersecting-bound}, which is also well-known. For completeness, we include the proof in \Cref{apx:spectral}.

\begin{restatable}{proposition}{proibcross} \label{pro:intersecting-bound-cross}
Suppose that there exists a $t$-good matrix for $\domain$.
If $\mathcal{F},\mathcal{G}$ are two cross-$t$-intersecting families then $\sqrt{|\mathcal{F}| \cdot |\mathcal{G}|} \leq \Lparam_t$. Furthermore, if equality holds then $\mathcal{F} = \mathcal{G}$, and the common characteristic function has degree at most $t$.
\end{restatable}

As a simple corollary, we can generalize \Cref{thm:intersecting-uniqueness} to cross-$t$-intersecting families.

\begin{theorem} \label{thm:intersecting-uniqueness-cross}
Suppose that there exists a $t$-good matrix for $\domain$, and every function of degree at most $t$ has certificate complexity at most $\Cbound_t \leq \Iparam_t$. Assume that $\Cbound_t \binom{\Cbound_t}{t} \Lparam_{t+1} < \Lparam_t$. If $\mathcal{F}$ and $\mathcal{G}$ are cross-$t$-intersecting families and $\sqrt{|\mathcal{F}| \cdot |\mathcal{G}|} = \Lparam_t$ then $\mathcal{F} = \mathcal{G}$ is a $t$-link.
\end{theorem}
\begin{proof}
\Cref{pro:intersecting-bound-cross} implies that $\mathcal{F} = \mathcal{G}$, and so $\mathcal{F}$ is a $t$-intersecting family of size $\Lparam_t$. The result now follows from \Cref{thm:intersecting-uniqueness}.	
\end{proof}

If $\domain = \Sym$ or $\domain = \pms$, then $\Cbound_t = O(t^8)$ by \Cref{thm:main}. Since $\Lparam_t/\Lparam_{t+1} = \Theta(n)$ and $\Iparam_t = n-2$ in both cases, the premises of \Cref{thm:intersecting-uniqueness} and \Cref{thm:intersecting-uniqueness-cross} other than the existence of a $t$-good matrix hold as long as $n \geq K t^{8(t+1)}$, for some absolute constant $K$. Ellis, Friedgut and Pilpel~\cite{EFP} and Lindzey~\cite{lindzey2018intersecting,LindzeyThesis} showed that a $t$-good matrix holds for large enough~$n$, and so we conclude the following corollary.

\begin{theorem} \label{thm:intersecting-application}
Let $\domain = \Sym$ or $\domain = \pms$. If $n$ is large enough (as a function of $t$), then every $t$-intersecting family has size at most $(n-t)!$ (if $\domain = \Sym$) or $(2n-2t-1)!! = (2n-2t-1)(2n-2t-3)\cdots(1)$	 (if $\domain = \pms$). Furthermore, this bound is attained only by $t$-links.

An identical bound holds on the geometric mean of the sizes of two cross-$t$-intersecting families. The bound is attained only if both families are the same $t$-link.
\end{theorem}

\section{Circuits} \label{sec:circuit}

Boolean circuits typically compute functions on the Boolean cube $\bits^n$. In this section, we consider Boolean circuits on arbitrary composable domains $(\domain,\universe,n)$. The input to such a circuit is the characteristic vector of some $x \in \domain$, encoded as a Boolean vector in $\bits^\universe$.
In particular cases, other input encodings are possible: for example, we can encode a permutation $\pi \in \Sym$ as a list of values $\pi(1),\ldots,\pi(n)$ encoded in binary. However, in order to keep the discussion as generic as possible, we only consider the ``unary'' encoding described above.

Among the various possible notions of Boolean circuits, we will consider circuits over the AND, OR, NOT basis with unbounded fan-in (``circuits'') and formulas over the same basis with bounded fan-in (``formulas'').

If a function has low decision tree complexity, then it can be computed by a shallow formula, by emulating the decision tree.

\begin{theorem} \label{thm:D-formula}
Every function $f\colon \domain \to \bits$ on a composable domain $(\domain,\universe,n)$ can be computed using a formula of depth $O(D(f)\log \mparam)$, where $\mparam = \max_{Q \in \queries} |Q|$ is the maximum number of answers to a query.
\end{theorem}
\begin{proof}
Consider a decision tree for $f$, and let $Q$ be the question asked at the root. For every answer $a \in Q$, let $f_a$ be the corresponding restriction of $f$. We construct the desired formula recursively using the identity
\[
 f(x) = \bigvee_{a \in Q} (x_a \land f_a(x)). \qedhere
\]
\end{proof}

This shows that a function $f$ on the Boolean cube can be computed by a formula of depth $O(D(f))$, and a function $f$ on $\Sym$ can be computed by a formula of depth $O(D(f)\log n)$.

\smallskip

A different construction, due to Gopalan et al.~\cite{GNSTW}, employs the ball property (\Cref{thm:ball-property}) to compute low sensitivity functions using small circuits and low depth formulas. We generalize the circuit construction, leaving the formula construction for future work.

We start by introducing two useful pieces of notations. The \emph{distance} between two elements $x,y \in \domain$ is $d(x,y) = |x \setminus y| = |y \setminus x|$. The \emph{ball} $B(x,r)$ at radius $r$ around $x$ consists of all points whose distance from $x$ is at most $r$.

Recall that the ball property states that a function $f\colon \domain \to \bits$ can be recovered from its values on a ball of radius $\Bcparam^{-1}(2s(f)+1)$ around an arbitrary point $x$. The idea is that given an arbitrary point $y$ outside the ball, using the sensitivity ratio property we can find $\ell = 2s(f)+1$ neighbors $z_1,\ldots,z_\ell$ of $x$ (differing by disjoint chunks) which are closer to $x$. Computing $f(z_1),\ldots,f(z_\ell)$ recursively, we can recover $f(y)$ by taking a majority vote.

The first step in making this approach algorithmic is to find the points $z_1,\ldots,z_\ell$. For specific domains such as the symmetric group, this can be done by closely following the sensitivity ratio argument. For arbitrary domains, if we are willing to pay slightly in the value of $\ell$, we can find these points using a \emph{membership oracle} $\oracle$, which checks whether a given vector $w \in \bits^\universe$ belongs to $\domain$. All domains considered in \Cref{sec:domains} have such membership oracles which can be computed using polynomial size circuits. For example, in the case of $\Sym$ we can take
\[
 \oracle(w) = \bigwedge_{i=1}^n \bigvee_{j=1}^n w_{(i,j)} \land \bigwedge_{j=1}^n \bigvee_{i=1}^n w_{(i,j)} \land
 \bigwedge_{i=1}^n \bigwedge_{1 \leq j_1 < j_2 \leq n} (\lnot w_{(i,j_1)} \lor \lnot w_{(i,j_2)}).
\]


\begin{lemma} \label{lem:Bcparam-alg}
Let $(\domain,\universe,n)$ be a composable domain with chunk size $\chunk$ and sensitivity ratio $\Bcparam$. Suppose that the membership oracle $\oracle$ can be implemented using a circuit of size $\oraclec$.

If $x,y \in \domain$ are two points at distance $d = |x \setminus y|$, then using a circuit of size $|\universe|^{O(\chunk)} \oraclec$
we can find $\ell \geq (\Bcparam/\chunk)d$ points $z_1,\ldots,z_\ell$ such that $d(x,z_i) < d$, the sets $y \setminus z_i$ are disjoint, and $|y \setminus z_i| = \chunk$.
\end{lemma}
\begin{proof}
There are $N \leq \binom{n}{\chunk}\binom{|\universe|-n}{\chunk} \leq |\universe|^{O(\chunk)}$ many sets of size $n$ at distance $\chunk$ from $y$. Using the membership oracle $\oracle$, we construct a list $Z$ of elements $z \in \domain$ satisfying $d(x,z) < d(x,y)$ and $d(y,z) = \chunk$. We then construct the list $z_1,\ldots,z_\ell$ by going over the elements in $Z$, adding each element $z$ such that $y \setminus z$ is disjoint from $y \setminus z_i$ for all $z_i$ already on the list.

Since the domain has sensitivity ratio $\Bcparam$, we know that there are $r \geq \Bcparam d$ elements $w_1,\ldots,w_r \in Z$ such that the sets $y \setminus w_i$ are disjoint. When our greedy algorithm encounters each $w_i$, it either adds it, or cannot do so since $y \setminus w_i$ intersects some $y \setminus z_j$. Since $|y \setminus z_j| = \chunk$ and the sets $y \setminus w_i$ are disjoint, each $z_j$ can ``spoil'' at most $\chunk$ many $w_i$. It follows that $\ell \geq r/\chunk$.
\end{proof}

Given this tool, we can construct the desired circuit.

\begin{theorem} \label{thm:s-circuit}
Let $(\domain,\universe,n)$ be a composable domain with chunk size $\chunk$ and sensitivity ratio $\Bcparam$. Suppose that the membership oracle $\oracle$ can be implemented using a circuit of size $\oraclec$.

Every function $f\colon \domain \to \bits$ of sensitivity $s = s(f)$ can be computed using a circuit of size $|\universe|^{O(\Bcparam^{-1}\chunk s)} \oraclec$.
\end{theorem}
\begin{proof}
Fix an arbitrary origin $x_0$, and let $r = \Bcparam^{-1}\chunk(2s+1)$. Hardcode the values of $f$ on $B(x,r)$, whose size is at most $|\universe|^{O(r)}$.

Given an input $w \notin B(x,r)$, we can find a point $z \in \domain$ such that $d(x,z) < d(x,w)$ and $d(z,w) = \chunk$ by trying all $|\universe|^{O(\chunk)}$ possible vectors at distance $\chunk$ from $w$ and using the membership oracle. In this way we can compute a path $x_1, \ldots, x_t = w$, where $x_1 \in B(x,r)$, $d(x,x_i) < d(x,x_{i+1})$, and $d(x_i,x_{i+1}) = \chunk$. Since $d(x,w) \leq n$, the path has length $t \leq n$.

We will compute $f(w)$ as follows. For each $i$, we compute $f$ on $B(x_i,r+\chunk)$ given its values on $B(x_i,r)$, in $\chunk$ steps. Since $d(x_i,x_{i+1}) = \chunk$, the triangle inequality shows that $B(x_{i+1},r) \subseteq B(x_i,r+\chunk)$, and so this gives us $f$ on $B(x_{i+1},r)$. In particular, once we have computed $f$ on $B(x_t,r)$, we know $f(w)$.

We compute $f$ on $B(x_i,r+\delta+1)$ given its values on $B(x_i,r+\delta)$ using \Cref{lem:Bcparam-alg}. Let $y \in B(x_i,r+\delta+1) \setminus B(x_i,r+\delta)$, so that $d(x_i,y) = r+\delta+1 \geq r$. Applying \Cref{lem:Bcparam-alg}, we find $\ell \geq 2s+1$ points $z_1,\ldots,z_\ell$ such that $z_j \in B(x_i,r+\delta)$, the sets $y \setminus z_i$ are disjoint, and $|y \setminus z_i| = \chunk$. Taking the majority of $f(z_1),\ldots,f(z_{2s+1})$, we deduce $f(y)$, as in the proof of \Cref{thm:ball-property}.

Altogether, the number of points in the balls $B(x_i,r+\chunk)$ is at most $t \cdot |\universe|^{O(r+\chunk)} = |\universe|^{O(\Bcparam^{-1}\chunk s)}$. The dominant term in the circuit size is the application of \Cref{lem:Bcparam-alg} for each of these points, and so we are led to the stated complexity.
\end{proof}

This shows that a function $f$ on the Boolean cube $\bits^n$ or on the symmetric group $\Sym$ can be computed by a circuit of size $n^{O(s(f))}$. \Cref{thm:D-formula}, in contrast, gives a circuit of size $2^{O(D(f))}$ in the former case, and of size $n^{O(D(f))}$ in the latter case. Since $s(f) \leq D(f)$ by \Cref{thm:main}, \Cref{thm:s-circuit} potentially improves on \Cref{thm:D-formula}.

\section{Open questions} \label{sec:open-questions}

\paragraph{Robust relations} The most interesting open direction, in our view, is proving a robust version of \Cref{thm:main}. Concretely, \Cref{thm:main} describes the structure of Boolean low degree functions: they correspond to shallow decision trees. What can we say about Boolean functions which are ``almost'' low degree, that is, are close to some (not necessarily Boolean) low-degree function?

To make this question precise, let us say that a function $f\colon \domain \to \bits$ is \emph{$\epsilon$-close to degree~$d$} if there exists a degree~$d$ function $g\colon \domain \to \mathbb{R}$ such that $\mathbb{E}[(f-g)^2] \leq \epsilon$, where the expectation is with respect to the uniform distribution over $\domain$. Does this imply that $f$ is close to a function computed by a decision tree of depth $\operatorname{poly}(d)$? More ambitiously, is $f$ close to a \emph{Boolean} degree $d$ function?

In the case of the Boolean cube, this has been answered in the affirmative by Kindler and Safra~\cite{KindlerS2002,Kindler2003}, who showed that for every \emph{fixed} $d$, if a Boolean function $f$ is $\epsilon$-close to degree $d$, then $f$ is $O(\epsilon)$-close to a Boolean degree~$d$ function; this was extended to slices (multislices of the form $M(k,n-k)$) by Keller and Klein~\cite{KellerKlein}. In ongoing work with Dor Minzer, the result of Kindler and Safra is extended to arbitrary $d$ by making the proof of \Cref{thm:main} robust.

Ellis, Filmus and Friedgut~\cite{EFF1,EFF2,EFF3} considered this question for the case of the symmetric group, showing (among else) that if a Boolean function is close to degree~$1$ then it is close to a Boolean degree~$1$ function. They also proved initial results for higher-degree functions, but these apply only to very sparse functions.

We conjecture that if a Boolean function on the symmetric group is close to degree~$d$, then it is close to a function computable by a decision tree of depth $\operatorname{poly}(d)$; and perhaps even close to a Boolean degree~$d$ function. Furthermore, we expect the same to hold for the perfect matching scheme.

\paragraph{$q$-analogs} Another interesting open direction is generalizing our framework to \emph{$q$-analogs}, that is, domains such as the Grassmann scheme (consisting of all $k$-dimensional subspaces of an $n$-dimensional vector space over $\mathbb{F}_q$), the bilinear scheme (all $n \times m$ matrices over $\mathbb{F}_q$), and the general linear group (all $n \times n$ invertible matrices over $\mathbb{F}_q$); all of these come with natural notions of degree.

What separates the domains considered in our paper and $q$-analogs in the natural symmetries they possess. All domains considered in this paper have a natural action of the symmetric group on them. In contrast, in the case of $q$-analogs the role of the symmetric group is played by the general linear group.
This is related to the fact that the Grassmann scheme is not a simplicial complex.

Filmus and Ihringer~\cite{FI1} have initiated the study of $q$-analogs from this perspective. They classified the Boolean degree~$1$ functions on the Grassmann scheme for $q=2,3,4,5$ (whenever $k,n-k$ are larger than a small constant), and proposed a conjectured classification of Boolean degree~$1$ functions on the Grassmann and bilinear schemes for all $q$. 

At present we can neither extend \Cref{thm:main} to $q$-analogs nor find nontrivial counterexamples.

\paragraph{Circuit complexity} We have briefly touched upon circuit complexity in \Cref{sec:circuit}, leaving the construction of shallow formulas for low sensitivity functions for future work.

We would like to highlight one particular question in this vein which we find intriguing: what is the complexity of calculating the sign of a permutation? The sign function is analogous to the parity function: both are the unique characters of maximum degree.

The sign of a permutation is the parity of the number of inversions. Using this, one can construct a circuit of size $O(n^2)$ for the sign function (which is optimal up to a constant factor), and a formula of size $O(n^5 \log n)$. However, the best formula lower bound we could come up with is only $\Omega(n^3)$, obtained using Khrapchenko's technique. It would be interesting to close this gap.

\paragraph{More relations} The literature contains many other relations between complexity measures. For example, Midrijanis~\cite{Midrijanis} proved that $D(f) \leq 2\deg(f)^3$, and Kulkarni and Tal~\cite{KulkarniTal} used the same method to prove that $R_0(f) = O(R(f)^2 \log R(f)^2)$. So far we have been unable to extend these arguments, which rely on \emph{maxonomials}, to our setting.

\bibliographystyle{alpha}
\bibliography{biblio}

\appendix

\section{Application of pseudo-characters} \label{apx:snstar}

In this section, we illustrate the utility of pseudo-characters by proving the following lemma. Before stating the lemma, let us recall that a $t$-star is a family of the form
\[
 \{ \pi \in \Sym : \pi(i_1) = j_1, \ldots, \pi(i_t) = j_t \}.
\]
(This coincides with the definition of \emph{$t$-link} in \Cref{sec:intersecting}.)

\snstar*

\begin{proof}
Suppose for concreteness that the $t$-star consists of all permutations containing $(1,1),\ldots,(t,t)$.
Let $f$ be the characteristic function of $\mathcal{F}$. Denote by $f_s$ the restriction of $f$ to the $s$-star $\mathcal{S}_s$ consisting of all permutations containing $(1,1),\ldots,(s,s)$.

We claim that if $f_1 \neq 0$ then $\deg(f_1) < \deg(f)$. Indeed, let $\deg(f_1) = d$. Our work in \Cref{sec:sensitivity-sn} shows that $f_1$ has nonzero correlation with some $\chi_{A,B}$ over $\mathcal{S}_1$, where $A,B$ is a pair of tabloids of shape $\lambda$, with $\lambda_1 = n-1-d$. Form new tabloids $A',B'$ by adding one more box filled with $1$ at the very bottom. This is a pair of tabloids of shape $\mu = \lambda,1$, which satisfies $\mu_1 = n-1-d = n-(d+1)$. Since $f$ vanishes outside of $\mathcal{S}_1$, we have
\[
 \sum_{x \in \Sym} f(x) \chi_{A',B'}(x) =
 \sum_\sigma (-1)^{\sigma} \sum_{x \in \mathcal{S}_1} f_1(x) e_{A',(B')^\sigma}(x),
\]
where $\sigma$ varies over all permutations of the columns of $B'$. If $\sigma$ moves the bottom box of $B'$ then $e_{A',(B')^\sigma}(x) = 0$ for all $x \in \mathcal{S}_1$, so we need not consider such $\sigma$. All remaining $\sigma$ permute the columns of $B$ and satisfy $e_{A',(B')^\sigma} = e_{A,B}$ over $\mathcal{S}_1$, and so
\[
 \sum_{x \in \Sym} f(x) \chi_{A',B'}(x) = \sum_{x \in \mathcal{S}_1} f(x) \chi_{A,B}(x) \neq 0.
\]
This shows that $\deg(f) \geq d+1$.

The same argument shows that $\deg(f_t) < \cdots < \deg(f)$ unless $f_t = 0$, and so $f_t$ is constant. Since $f_t$ is Boolean, either $f_t = 0$ or $f_t = 1$.
\end{proof}

We comment that the result can be proven, for large enough $n$ (as a function of $t$), using the methods outlined in \Cref{sec:intersecting-spectral}. Indeed, suppose that a $t$-good matrix $A$ exists for $\Sym$ (Ellis, Friedgut and Pilpel~\cite{EFP} showed that this holds for large enough $n$, as a function of $t$). Let $\mathcal{F}$ be a subset of a $t$-star whose characteristic function $f$ has degree at most~$t$. Since $\mathcal{F}$ is a subset of a $t$-star, it is automatically $t$-intersecting. The proof of \Cref{pro:intersecting-bound} shows that
\[
 0 = \langle f,Af \rangle = \mu^2 - \omega(\mu - \mu^2),
\]
where $\omega = (n-t)!/(n! - (n-t)!)$. Thus either $\mu = 0$ (and so $\mathcal{F} = \emptyset$) or $\mu = \omega/(1+\omega) = (n-t)!/n!$ (and so $\mathcal{F}$ is a $t$-star).

\section{More on the spectral technique} \label{apx:spectral}

\Cref{sec:intersecting-spectral} outlines a spectral approach for bounding the size of $t$-intersecting families and for characterizing $t$-intersecting families of maximum size.
The input to the spectral approach is a matrix satisfying the properties given in the following definition.

\defad*

(In this definition, $\Lparam_t$ is the maximum size of a \emph{$t$-link}, as defined in \Cref{sec:intersecting-result}.)

In this brief section, we explain how the existence of a $t$-good matrix implies an upper bound on the size of $t$-intersecting and cross-$t$-intersecting families. The proofs below are standard and appear in many papers, and are reproduced here for the sake of completeness.

\smallskip

We start by showing that a $t$-good matrix implies a bound on the size of $t$-intersecting families. The argument is known as the \emph{weighted Hoffman bound}, which is closely related to the Lov\'asz $\theta$ function.

\proib*
\begin{proof}
Let $A$ be a $t$-good matrix, let $\mathcal{F}$ be a $t$-intersecting family, and let $f$ be its characteristic vector. Define $\mu = \EE[f] = |\mathcal{F}|/|\domain|$.

Let $V_0$ consist of all constant functions, and let $V_1$ consist of all functions of degree at most~$t$ orthogonal to $V_0$. Both of these are eigenspaces of $A$. Denote by $V_2,\ldots,V_m$ its remaining eigenspaces. Since $A$ is symmetric, these eigenspaces are orthogonal.
Denote the corresponding eigenvalues by $\lambda_0,\ldots,\lambda_m$. Thus $\lambda_0 = 1$, $\lambda_1 = -\omega$, and $|\lambda_i| < \omega$ for $i \geq 2$.

We will use the following notations, for functions $g,h$ on $\domain$: $\langle g,h \rangle = \EE[gh]$, and $\|g\|^2 = \langle g,g \rangle$.

Let $f_0,\ldots,f_m$ be the projections of $f$ into $V_0,\ldots,V_m$, respectively. Since
\[
 \langle \mu \mathbf{1}, f - \mu \mathbf{1} \rangle = \mu \EE[f] - \mu^2 \EE[\mathbf{1}] = 0,
\]
we see that $f_0 = \mu \mathbf{1}$. Also, orthogonality guarantees that
\[
 \|f_0\|^2 + \cdots + \|f_m\|^2 = \|f\|^2 = \EE[f^2] = \EE[f] = \mu.
\]

Since $\mathcal{F}$ is $t$-intersecting,
\[
 \langle f, Af \rangle = \frac{1}{|\domain|} \sum_{x,y \in \domain} A(x,y) f(x) f(y) = 0.
\]
On the other hand, since $Af_i = \lambda_i f_i$ and the eigenspaces are orthogonal,
\[
 0 = \langle f, Af \rangle = \sum_{i=0}^m \langle f_i, \lambda_i f_i \rangle = \sum_{i=0}^m \lambda_i \|f_i\|^2 \geq \mu^2 - \omega \sum_{i=1}^m \|f_i\|^2 = \mu^2 - \omega (\mu - \mu^2),
\]
with equality if and only if $f_i = 0$ for all $i \geq 2$, that is, $\deg f \leq t$.

We thus obtain the inequality $(1+\omega) \mu \leq \omega$, hence
\[
 \mu \leq \frac{\omega}{1+\omega} = \frac{\Lparam_t}{|\domain|}.
\]
In other words, $|\mathcal{F}| \leq \Lparam_t$. Furthermore, if equality holds then $\deg f \leq t$.
\end{proof}

The argument extends to the cross-$t$-intersecting case, by throwing in several applications of the Cauchy--Schwartz inequality.

\proibcross*
\begin{proof}
We will use the same notation as in the proof of \Cref{pro:intersecting-bound}.	Let $\mu = \EE[f]$ and $\nu = \EE[g]$. Since $\mathcal{F}$ and $\mathcal{G}$ are cross-$t$-intersecting, we have
\[
 0 = \langle f, Ag \rangle = \sum_{i=0}^m \lambda_i \langle f_i,g_i \rangle = \mu \nu + \sum_{i=1}^m \lambda_i \langle f_i,g_i \rangle.
\]
The Cauchy--Schwartz inequality shows that $|\langle f_i,g_i \rangle| \leq \|f_i\| \cdot \|g_i\|$, and so
\begin{equation} \tag{$\ast$} \label{eq:tight}
 0 = \mu \nu + \sum_{i=1}^m \lambda_i \langle f_i,g_i \rangle \geq
 \mu \nu - \omega \sum_{i=1}^m \|f_i\| \cdot \|g_i\|.
\end{equation}
Applying the Cauchy--Schwartz inequality to the sum gives
\[
 0 \geq \mu \nu - \omega \sqrt{\sum_{i=1}^m \|f_i\|^2} \sqrt{\sum_{i=1}^m \|g_i\|^2} = \mu \nu - \omega \sqrt{\mu-\mu^2} \sqrt{\nu-\nu^2} \geq \mu \nu - \omega (\sqrt{\mu\nu} - \mu\nu),
\]
using the following consequence of the AM-GM inequality (applied twice):
\[
 \sqrt{(\mu - \mu^2)(\nu - \nu^2)} =
 \sqrt{\mu\nu} \sqrt{(1-\mu)(1-\nu)} \leq
 \sqrt{\mu\nu} \left(1 - \frac{\mu+\nu}{2}\right) \leq
 \sqrt{\mu\nu} (1 - \sqrt{\mu\nu}).
\]
As in the proof of \Cref{pro:intersecting-bound}, this shows that $\sqrt{\mu\nu} \leq \Lparam_t/|\domain|$, and so $\sqrt{|\mathcal{F}| \cdot |\mathcal{G}|} \leq \Lparam_t$.

If $\sqrt{|\mathcal{F}| \cdot |\mathcal{G}|} = \Lparam_t$ then the AM-GM inequality is tight, implying that $\mu = \nu$. Furthermore, the inequalities in \eqref{eq:tight} are tight, implying that $\deg(f),\deg(g) \leq t$ (since $|\lambda_i| < \omega$ for $i \geq 2$) and that $f_1,g_1$ are parallel (since the Cauchy--Schwartz inequality $\langle f_1,g_1 \rangle \leq \|f_1\| \cdot \|g_1\|$ is tight). Since $f_1 = f - \mu$ and $g_1 = g - \nu = g - \mu$ are both $\{-\mu,1-\mu\}$-valued, necessarily $f_1 = g_1$, and so $f = g$. (Note that $f_1 = -g_1$ is impossible even when $\mu = 1/2$, since then $\langle f_1,g_1 \rangle = -\|f_1\| \cdot \|g_1\|$.)
\end{proof}

\section{Generalized permutations and the RSK correspondence} \label{apx:rsk}

\paragraph{Boolean cube} Every function $f\colon \bits^n \to \RR$ has a unique representation as a linear combination of Fourier characters:
\[
 f = \sum_{S \subseteq [n]} \hat{f}(S) \chi_S, \text{ where } \chi_S(x_1,\ldots,x_n) = \prod_{i \in S} x_i.
\]
There is a natural correspondence between Fourier characters and the Boolean cube itself, given by $\chi_S \leftrightarrow 1_S$. More generally, if $G$ is a finite Abelian group (in this case, $\bits^n$ under bitwise XOR), then $\hat{G}$ (the \emph{dual group}, the group of characters under multiplication) is isomorphic to $G$.

Let us say that a function on the Boolean cube has \emph{pure degree $d$} if it has degree~$d$ and is orthogonal to all functions of smaller degree. It turns out that the space of pure degree~$d$ functions is spanned by the Fourier characters $\chi_S$, where $S$ goes over all sets of size $d$. We say that these Fourier characters belong to \emph{level $d$}.

The correspondence between Fourier characters and the Boolean cube induces a correspondence between the decomposition of the Fourier characters into levels and the decomposition of the Boolean cube according to Hamming weight.

In this appendix we show how to extend this correspondence to domains which can be described by \emph{generalized permutations}, such as the symmetric group $\Sym$, via the RSK correspondence (and its special case, the Robinson--Schensted correspondence).

\paragraph{Symmetric group} In \Cref{sec:sensitivity-sn} we described the analog of the Fourier expansion on the symmetric group:
\[
 \RR[\Sym] = \bigoplus_{\lambda \vdash n} V^\lambda,
\]
where $\lambda$ goes over all partitions of $n$ (non-increasing sequences of positive integers summing to $n$), and the subspaces $V^\lambda$ are \emph{isotypic components}, which we described explicitly in \Cref{sec:sensitivity-sn}. In contrast to the Fourier decomposition of the Boolean cube, the subspaces $V^\lambda$ are generally not one-dimensional (unless $\lambda = n$ or $\lambda = 1^n$).

By combining isotypic components according to $\lambda_1$, we get the level decomposition of the Fourier expansion:
\[
 \RR[\Sym] = \bigoplus_{d=0}^{n-1} \bigoplus_{\substack{\lambda \vdash n \\ \lambda_1 = n-d}} V^\lambda.
\]
The $d$'th summand is the space of pure degree~$d$ functions.

Going the other way, we can decompose each $V^\lambda$ into an  orthogonal basis known as the Gelfand--Tsetlin (GZ) basis, which is canonical given an ordering of the coordinates $1,\ldots,n$. We refer the reader to~\cite{CSST2} for more details on the GZ basis.

Our goal here is to find a decomposition of the symmetric group into parts that mirror the Fourier levels. In fact, our decomposition will be finer: it will correspond to the decomposition into isotypic components. It can be further refined to obtain a bijection between the GZ basis and the symmetric group, but we do not pursue this here.

We construct the decomposition using the Robinson--Schensted correspondence, which is explained in full detail in~\cite{Sagan}. In order to describe the correspondence, we need to define standard Young tableaux. For a partition $\lambda = \lambda_1,\dots,\lambda_m \vdash n$, a \emph{Young tableau} of shape $\lambda$ consists of left-justified rows of lengths $\lambda_1,\dots,\lambda_m$ filled with the numbers $1,\ldots,n$. For example, here is a Young tableau of shape $3,2,1$:
\[
\ytableausetup{notabloids}
\begin{ytableau}
1&2&3 \\ 4&5 \\ 6
\end{ytableau}
\]
A Young tableau is \emph{standard} if the numbers increase along rows and columns.

\begin{proposition}[Robinson--Schensted, Schensted, Greene] \label{pro:rsk-sn}
There is a bijection between permutations $\pi \in \Sym$ and pairs of standard Young tableaux of the same shape $\lambda$, where $\lambda$ goes over all partitions of $n$.

Furthermore, $\lambda_1$ is the length of a longest increasing subsequence of $\pi$, and more generally, $\lambda_1 + \cdots + \lambda_k$ is the maximal length of the union of $k$ increasing subsequences in $\pi$.
\end{proposition}

Let $d_\lambda$ be the number of standard Young tableaux of shape $\lambda$. It is known that $\dim(V^\lambda) = d_\lambda^2$, and furthermore there is a canonical way to convert a pair of standard Young tableaux of shape $\lambda$ into a GZ basis vector inside $V^\lambda$.

We conclude that the number of permutations whose longest increasing subsequence has length $n-d$ coincides with the dimension of pure degree~$d$ functions, and this gives us a decomposition of the symmetric group that mirrors the decomposition into Fourier levels. We can obtain a more refined decomposition by considering the rest of $\lambda$, which also has an interpretation in terms of increasing subsequences. As an example, we give the decomposition of $\Sym[4]$ in \Cref{fig:S4-decomposition}.

\begin{figure}
\[
\begin{array}{|c|c|lr|} \hline
\text{Level} & \text{Shape} & \multicolumn{2}{c|}{\text{Permutations}} \\\hline
0 & 4 & \multicolumn{2}{c|}{1234} \\\hline
\multirow{5}{*}{1} & \multirow{5}{*}{3,1} &
1243 & 2134 \\
& & 1324 & 2314 \\
& & 1342 & 2341 \\
& & 1423 & 3124 \\
& & & 4123 \\\hline
\multirow{7}{*}{2} & \multirow{2}{*}{2,2} &
2143 & 3124 \\
& & 2413 & 3412 \\\cline{2-4}
& \multirow{5}{*}{2,1,1} &
1432 & 4132 \\
& & 2431 & 4213 \\
& & 3214 & 4231 \\
& & 3241 & 4312 \\
& & 3421 & \\\hline
3 & 1,1,1,1 & \multicolumn{2}{c|}{4321} \\\hline
\end{array}
\]
\caption{Decomposition of $\Sym[4]$ according to the Robinson--Schensted correspondence}
\label{fig:S4-decomposition}
\end{figure}

\paragraph{Perfect matching scheme} In \Cref{sec:sensitivity-pm}, we described the analog of the Fourier expansion on the perfect matching scheme, which we present according to the level decomposition:
\[
 \RR[\pms] = \bigoplus_{d=0}^{n-1} \bigoplus_{\substack{\lambda \vdash n \\ \lambda_1 = n-d}} V^{2\lambda},
\]
where $2\lambda$ is the partition obtained by doubling each part of $\lambda$. 
As in the case of the symmetric group, the $d$'th summand is the space of pure degree~$d$ functions.

We stress that the isotypic components here differ from the isotypic components in the decomposition of the symmetric group. In particular, $V^{2\lambda}$ has dimension $d_{2\lambda}$ rather than $d_{2\lambda}^2$. The isotypic components are described explicitly in \Cref{sec:sensitivity-pm}.

We can represent a perfect matching $m \in \pms$ as a permutation in $\Sym[2n]$ which is a product of $n$ many $2$-cycles, corresponding to the edges in $m$. We call this the \emph{permutation representation} of $m$. It turns out that if we apply the Robinson--Schensted correspondence to such a permutation, we get a pair of equal tableaux, and moreover, their shape after \emph{transposition} (reflection along the main diagonal)\footnote{This ``reversed situation'' is an artifact of the representation-theoretic fact that each of the one-dimensional irreducibles of the hyperoctahedral group induces to a multiplicity-free representation of the symmetric group, in particular, the sign representation of the hyperoctahedral group.} is always of the form $2\lambda$.

\begin{proposition}
\label{pro:rsk-pms}
There is a bijection between perfect matchings $m \in \pms$ and  standard Young tableaux of shape $2\lambda$, where $\lambda$ goes over all partitions of $n$.

Furthermore, $2\lambda_1$ is the length of a longest decreasing subsequence of the permutation representation of $m$, and more generally, $2(\lambda_1 + \cdots + \lambda_k)$ is the maximal length of the union of $k$ decreasing subsequences in the permutation representation of $m$.
\end{proposition}

Just as in the case of the symmetric group, this results in a decomposition of the perfect matching scheme that mimics the level decomposition. We demonstrate this for $\pms[6]$ in \Cref{fig:M6-decomposition}.

\begin{figure}
\[
\begin{array}{|c|c|lr|} \hline
\text{Level} & \text{Transposed Shape} & \multicolumn{2}{c|}{\text{Perfect Matchings}} \\\hline
0 & 6 & \multicolumn{2}{c|}{[16][25][34]} \\\hline
\multirow{5}{*}{1} & \multirow{5}{*}{4,2} &
[12][36][45] & [15][23][46] \\
& & [13][26][45] & [15][24][36] \\
& & [14][23][56] & [15][26][34] \\
& & [14][26][35] & [16][23][45] \\
& & & [16][24][35] \\\hline
\multirow{3}{*}{1} & \multirow{3}{*}{2,2,2} &
[12][34][56] & [13][24][56] \\
& & [12][35][46] & [13][25][46] \\
& & & [14][25][36] \\\hline
\end{array}
\]
\caption{Decomposition of $\pms[6]$ according to the Robinson--Schensted correspondence}
\label{fig:M6-decomposition}
\end{figure}

\paragraph{Multislices} Recall that for a partition $\mu = \mu_1,\ldots,\mu_m \vdash n$, the multislice $M(\mu)$ consists of all vectors in $\{1,\ldots,m\}^n$ with exactly $\mu_i$ coordinates equal to $i$. If $m = 2$, then we also use the name \emph{slice}.

In \Cref{sec:sensitivity-ms} we described the analog of the Fourier expansion on multislices, which we present according to the level decomposition:
\[
 \RR[M(\mu)] = \bigoplus_{d=0}^{n-\lambda_1} \bigoplus_{\substack{\lambda \trianglerighteq \mu \\ \lambda_1 = n-d}} V^d,
\]
where $\lambda \trianglerighteq \mu$ (read: $\lambda$ dominates $\mu$) if $\lambda_1 + \cdots + \lambda_r \geq \mu_1 + \cdots + \mu_r$ for all $r$, extending the partitions with infinitely many zeroes. In particular, $\lambda$ has at most as many parts as $\mu$. As in the preceding cases, the $d$'th summand is the space of pure degree $d$ functions. The isotypic components $V^\lambda$ are described explicitly in \Cref{sec:sensitivity-ms}.

We can represent an element of the multislice as a permutation in $\Sym$ which lists all coordinates colored $1$ in increasing order, then all coordinates colored $2$ in increasing order, and so on. For example, $(1,2,2,1,1) \in M(3,2)$ corresponds to the permutation $14523$.

When we run the Robinson--Schensted correspondence, we do not get all possible pairs of tableaux. In order to fix that, we treat the elements of the multislice as a generalized permutation (defined below) with a fixed top row, and use the Robinson--Schensted--Knuth (RSK) correspondence, a version of the Robinson--Schensted correspondence for generalized permutations.

A \emph{generalized permutation} is a $2 \times n$ array $\begin{pmatrix} i_1 & \cdots & i_n \\ j_1 & \cdots & j_n \end{pmatrix}$ such that $i_1 \leq i_2 \leq \cdots \leq i_n$, and whenever $i_r = i_{r+1}$ then $j_r \leq j_{r+1}$.

We view an element of the multislice as a generalized permutation by fixing the top row to $1^{\lambda_1} 2^{\lambda_2} \cdots m^{\lambda_m}$, and by using the permutation representation as the second row. Continuing our previous example, $(1,2,2,1,1) \in M(3,2)$ is encoded as the generalized permutation
\[
 \begin{pmatrix}
 	1 & 1 & 1 & 2 & 2 \\
 	1 & 4 & 5 & 2 & 3
 \end{pmatrix}
\]

Running the RSK correspondence still produces a pair of tableaux of the same shape, but only the second one is standard. The first one is \emph{semistandard}, meaning that it is nondecreasing along rows and increasing along columns, and has \emph{content} $\mu$, that is, exactly $\mu_r$ entries are equal to $r$.


\begin{proposition}
\label{pro:rsk-ms}
There is a bijection between elements of the multislice $M(\mu)$ and pairs of Young tableaux $(P,Q)$ of the same shape $\lambda$, where $P$ is a semistandard tableau having content $\mu$, $Q$ is a standard tableau, and $\lambda$ goes over all partitions of $n$ dominating $\mu$.

Furthermore, $\lambda_1$ is the length of a longest increasing subsequence of the permutation representation of the element, and more generally, $\lambda_1 + \cdots + \lambda_k$ is the maximal length of the union of $k$ increasing subsequences in the permutation representation of the element.	
\end{proposition}

As in the preceding cases, the number of elements corresponding to pairs of tableaux of shape $\lambda$ is exactly $\dim(V^\lambda)$. We illustrate the resulting decomposition of the multislice in the case of the slice $M(2,3)$ in \Cref{fig:M23-decomposition}.

\begin{figure}
\[
\begin{array}{|c|c|lr|} \hline
\text{Level} & \text{Shape} & \multicolumn{2}{c|}{\text{Elements}} \\\hline
0 & 5 & \multicolumn{2}{c|}{\{1,2\}} \\\hline
\multirow{2}{*}{1} & \multirow{2}{*}{4,1} &
\{1,3\} & \{1,4\} \\
& & \{2,3\} & \{1,5\} \\\hline
\multirow{3}{*}{2} & \multirow{3}{*}{3,2} &
\{2,4\} & \{3,4\} \\
& & \{2,5\} & \{3,5\} \\
& & & \{4,5\} \\\hline
\end{array}
\]
\caption{Decomposition of the slice $M(2,3)$ according to the RSK correspondence, where each element is represented as the set of elements colored~$1$}
\label{fig:M23-decomposition}
\end{figure}

\paragraph{Product domains} The product domain $[m]^n = H(m,\ldots,m)$ (where $m$ appears $n$ times) decomposes under the action of the symmetric group $\Sym$ into a direct sum of multislices. Each multislices further decomposes as we have shown above. This gives rise to the following decomposition:
\[
 \mathbb{R}[H(\overbrace{m,\ldots,m}^{n \text{ times}})] =
 \bigoplus_{\mu \vdash n} \bigoplus_{\lambda \trianglerighteq \mu} V^{\mu,\lambda}.
\]
Unfortunately, this decomposition doesn't quite correspond to the degree decomposition of $[m]^n$. This is because we should really be considering the action of $\Sym[m] \wr \Sym[n]$ rather than that of $\Sym[n]$. Nevertheless, it gives rise to an interesting decomposition of $[m]^n$.

In order to obtain the decomposition, we encode each element of $[m]^n$ as a generalized permutation in which the top row is $1,\ldots,n$ and the bottom row is an arbitrary word in $[m]^n$, and run the RSK correspondence on it.

\begin{proposition}
\label{pro:rsk-ms}
There is a bijection between elements of $[m]^n$ and pairs of Young tableaux $(P,Q)$ of the same shape $\lambda$, where $P$ is a semistandard tableau containing numbers in $[m]$, $Q$ is a standard tableau, and $\lambda$ goes over all partitions of $n$.

Furthermore, $\lambda_1$ is the length of a longest nondecreasing subsequence of the element, and more generally, $\lambda_1 + \cdots + \lambda_k$ is the maximal length of the union of $k$ nondecreasing subsequences in the element.	
\end{proposition}

We illustrate the resulting decomposition of $H(3,3,3)$ in \Cref{fig:H333-decomposition}.

\begin{figure}
\[
\begin{array}{|c|c|lccr|} \hline
\text{Level} & \text{Shape} & \multicolumn{4}{c|}{\text{Elements}} \\\hline
\multirow{3}{*}{0} & \multirow{3}{*}{3} &
111 & 122 & 222 & 333 \\
& & 112 & 123 & 223 & \\
& & 113 & 133 & 233 & \\\hline
\multirow{4}{*}{1} & \multirow{4}{*}{2,1} &
121 & 212 & 232 & 322 \\
& & 131 & 213 & 311 & 323 \\
& & 132 & 221 & 312 & 331 \\
& & 211 & 231 & 313 & 332 \\\hline
2 & 1,1,1 & \multicolumn{4}{c|}{321} \\\hline
\end{array}
\]
\caption{Decomposition of the slice $H(3,3,3)$ according to the RSK correspondence}
\label{fig:H333-decomposition}
\end{figure}

\end{document}